\documentclass[11pt]{article}

\usepackage[letterpaper]{geometry}
\usepackage{caption}%labelsep=space
\usepackage{mathtools,amsmath}
\usepackage{dsfont}  
\usepackage{endnotes}
\usepackage{indentfirst}  
\usepackage{graphicx}
\usepackage{epstopdf}
\usepackage{natbib}
\usepackage{multibib}
\usepackage{amsmath, amsthm, amssymb}   
\usepackage{amsfonts}
\usepackage{rotating}
\usepackage{bbm}
\usepackage{tikz}
\usetikzlibrary{arrows, decorations.markings}
\usetikzlibrary{positioning,calc,decorations.pathreplacing}
\usetikzlibrary{decorations.pathreplacing,calligraphy}
\usepackage{varwidth}
\usepackage{array}
\usepackage{xcolor}
\usepackage{comment}
\usepackage{float}
\usepackage{booktabs}
\usepackage{subcaption}

\usepackage{algorithm}
\usepackage{algpseudocode}

\floatstyle{ruled}
\restylefloat{algorithm}

\usepackage[normalem]{ulem}

\definecolor{Maroon}{RGB}{128,0,0}
\definecolor{Mygreen}{RGB}{0, 100, 0}
\definecolor{Pool}{RGB}{28,107,160}
\definecolor{Chase}{RGB}{88,14,235}

\usetikzlibrary{shapes,arrows}
\tikzstyle{vecArrow} = [thick, decoration={markings,mark=at position
   1 with {\arrow[semithick]{open triangle 60}}},
   double distance=1.4pt, shorten >= 5.5pt,
   preaction = {decorate},
   postaction = {draw,line width=1.4pt, white,shorten >= 4.5pt}]
\tikzstyle{innerWhite} = [semithick, white,line width=1.4pt, shorten >= 4.5pt]

\usepackage{xr-hyper}
\usepackage[colorlinks  = true,
            linkcolor   = blue,
            urlcolor    = blue,
            citecolor   = blue,
            anchorcolor = blue]{hyperref}
\usepackage{accents}
\long\def\symbolfootnote[#1]#2{\begingroup%
\def\thefootnote{\fnsymbol{footnote}}\footnote[#1]{#2}\endgroup}
\makeatletter
\def\@biblabel#1{\hspace*{-\labelsep}}
\makeatother

\usepackage{bigints}

\newtheorem{prop}{Proposition}
\newtheorem{cor}{Corollary}
\newtheorem{lem}{Lemma}
\newtheorem{definition}{Definition}

\newtheorem{rem}{Remark}

\newtheorem*{linearexample*}{Linear Example}

\usepackage{soul}

\newcommand{\E}[0]{\mathbb{E}}
\renewcommand{\P}[0]{\mathbb{P}}

\newcommand{\cZ}{\mathcal{Z}}

\DeclareMathOperator*{\argmax}{arg\,max}

\usepackage{threeparttable}% For Notes below table

\usepackage{siunitx}
\def\yyy{%
  \bgroup\uccode`\~\expandafter`\string-%
  \uppercase{\egroup\edef~{\noexpand\text{\llap{\textendash}\relax}}}%
  \mathcode\expandafter`\string-"8000 }

\def\xxxl#1{%
\bgroup\uccode`\~\expandafter`\string#1%
\uppercase{\egroup\edef~{\noexpand\text{\noexpand\llap{\string#1}}}}%
\mathcode\expandafter`\string#1"8000 }

\def\xxxr#1{%
\bgroup\uccode`\~\expandafter`\string#1%
\uppercase{\egroup\edef~{\noexpand\text{\noexpand\rlap{\string#1}}}}%
\mathcode\expandafter`\string#1"8000 }

\usepackage{pdflscape}
\usepackage{pdfpages}
\usepackage{setspace}
\begin{document}

%\setstretch{2} %this sets the space between lines

\title{Latent Fragility and Clustered Withdrawals in \\Dynamic Banks Runs}

\author{Jodi Dianetti, Giorgio Ferrari, Yunzhi Hu,  and  Hao Xing
\footnote{Dianetti: Department of Economics and Finance, University of Rome Tor Vergata; \href{jodi.dianetti@uniroma2.it}{jodi.dianetti@uniroma2.it}; Ferrari: Center for Mathematical Economics (IMW), Bielefeld University, Germany; \href{giorgio.ferrari@uni-bielefeld.de}{giorgio.ferrari@uni-bielefeld.de}; Hu: Kenan-Flagler Business School, UNC Chapel Hill; \href{yunzhi_hu@kenan-flagler.unc.edu}{Yunzhi\_Hu@kenan-flagler.unc.edu}. Xing: Questrom School of Business, Boston University; \href{haoxing@bu.edu}{haoxing@bu.edu}.\\  \hspace*{14pt}  }}
\date{\today}
% \author{}
% \date{}

\maketitle

\thispagestyle{empty} % this line does not skip numbering the first page but hides the page number

\begin{abstract}

\noindent Using a mean-field game framework, we study a dynamic model of bank runs in which more withdrawals raise the risk of bank failure.  Even though depositors receive gradual and idiosyncratic shocks, withdrawals occur in clusters. The main mechanism is latent fragility: run-prone depositors accumulate gradually over time and may prefer to wait individually, but they withdraw together once collective exit becomes self-fulfilling. We establish equilibrium existence and characterize earliest-run and latest-run equilibria. The clustering mechanism arises whether depositor heterogeneity is discrete or continuous. A common aggregate state coordinates withdrawal timing and leads to a unique threshold equilibrium.

\vspace{3mm}
\end{abstract}

\newpage
\section{Introduction}
%\yunzhi{A few things we need to take care: 1) introduce heterogeneity; 2) welfare implication; 3) policy implications}
Bank runs unfold as sudden but sometimes partial events. For an extended period, observed withdrawals may remain modest; then a large group of depositors exits within a short window while others remain in the bank. This pattern is inherently dynamic: fragility can build while observed withdrawals remain low and then materialize in a concentrated wave. Modern payment technology and social media can compress the time required to communicate and execute withdrawals, while depositor responses vary with insurance, bank relationships, and information networks \citep{iyer2012understanding,cookson2025social}. Yet neither faster execution nor cross-sectional heterogeneity explains why many depositors withdraw together at a particular moment, why others remain, or what determines the size of the run.

This pattern is particularly puzzling when the underlying shocks are gradual and idiosyncratic. Depositors experience changes in liquidity needs, outside options, and information at different times. In a large population, one would expect these independent shocks to average out, producing smooth aggregate outflows. A standard sunspot can coordinate withdrawal timing, but it typically generates an all-or-nothing run and does not link the size of the run to an evolving distribution of withdrawal incentives. Why, then, do gradual and heterogeneous withdrawal incentives aggregate into a discontinuous cluster? What determines which depositors join, when they withdraw, and how large the cluster becomes? How do aggregate conditions coordinate this process?

We answer these questions using a dynamic mean-field game of bank runs. A continuum of depositors choose when to withdraw, and the bank's failure intensity increases with the withdrawn share. Heterogeneous private states build \emph{latent fragility}: run-prone depositors accumulate inside the bank while waiting individually. Strategic complementarity then converts this latent fragility into a clustered withdrawal, because an increase in withdrawals raises failure risk and makes further withdrawals optimal. Finally, a common state coordinates when the accumulated run-prone depositors withdraw. The mean-field framework allows us to study this sequence while retaining a rich distribution of depositor states.

We first establish equilibrium existence and ordering in a non-stationary environment. A larger withdrawn share induces remaining depositors to withdraw earlier. This monotonicity leads to an earliest-run equilibrium and a latest-run equilibrium, with all other equilibria sandwitched between them. It also implies that stronger strategic complementarity leads to earlier withdrawals and a larger withdrawn share. 

We next characterize clustered withdrawals with both discrete and continuous depositor heterogeneity. In a tractable two-type model, depositors become impatient after gradual idiosyncratic shocks but may continue to wait. Once enough impatient depositors have accumulated, their collective withdrawal raises failure risk sufficiently to make the cluster self-fulfilling. We characterize the timing and size of this cluster in closed form. With continuous private states, a distribution of run-prone depositors instead accumulates near an endogenous withdrawal boundary. Strategic complementarity can make this boundary jump, causing a positive but partial mass of depositors to withdraw together. Thus, clustering does not rely on a discrete patient-impatient distinction or a mass point in depositor types.

We then introduce a common state that coordinates withdrawal timing and can generate multiple clusters. As shown by \cite{FrankelPauzner2000}, when depositors receive idiosyncratic withdrawal opportunities at a finite rate, the common state selects a unique threshold equilibrium. When immediate withdrawal is feasible in this frictionless setting, however, their original uniqueness argument does not apply. One contribution of our paper is to further establish uniqueness when depositors can instead withdraw at arbitrary stopping times by using a strong maximum principle. Each time the common state crosses the equilibrium boundary, the run-prone depositors accumulated since the previous crossing may withdraw together, generating multiple clusters.

The model delivers several empirical predictions. Smooth depositor-level shocks can generate discrete and partial jumps in aggregate withdrawals. Holding other conditions fixed, later clusters should typically be larger because more run-prone depositors have had time to accumulate. Cluster participants should be drawn disproportionately from depositors whose characteristics or recent shocks place them near the withdrawal margin, with the relevant characteristics depending on the source of stress. Stronger strategic complementarity brings a cluster forward but need not make it larger, because cluster size depends on how many run-prone depositors have accumulated when coordination occurs. Finally, the same aggregate deterioration should generate a larger response when latent fragility is already high, implying an interaction between depositor-level withdrawal propensities and bank-level balance-sheet exposure.

\medskip

The paper contributes to several literatures. First, it contributes to the theoretical literature on bank runs. The classic model of \citet{diamond1983bank} shows how demand deposits can create self-fulfilling runs, while \citet{goldstein2005demand} and \citet{rochet2004coordination} study global-game and coordination perspectives on run risk. Our contribution is to make depositor heterogeneity and withdrawal timing central. Rather than asking whether a run occurs at a point in time, we study how latent fragility accumulates and how the size and timing of clustered withdrawals are determined endogenously.

Second, the paper relates to dynamic models of runs and synchronized exit. \citet{abreu2003bubbles} show how agents may delay exit before a coordinated crash, and dynamic run models such as \citet{he2012dynamic}, \citet{he2016information}, and \citet{zhong2026dynamic} also feature waiting before a run. These papers emphasize the dynamic nature of runs and the role of information, rollover timing, or regulation. We share the focus on dynamics, but our mechanism is different. In our model, rich depositor heterogeneity creates a stock of run-prone depositors, and strategic complementarity determines when that stock exits as a cluster. This allows us to predict both the timing and the size of clustered withdrawals.

The paper also relates to dynamic models of strategic complementarities outside banking. \citet{alvarez2026p2p} study the adoption of peer-to-peer digital payments, where the benefit of adoption increases with the share of adopters. Their model generates gradual adoption, and their quantitative exercise shows that no jumps (similar to our withdrawal waves) occur in the aggregate adoption path. Moreover, much of their analysis starts from the stationary distribution of a reflected Brownian motion, whereas in our model, latent fragility builds up endogenously. In the two-type Poisson case, we prove analytically that discontinuous equilibria exist; when there are exogenous withdrawals, all equilibria are discontinuous. In the continuous-state case, we provide numerical examples in which strategic complementarity generates a jump in the withdrawn share. 

Third, the paper is related to empirical work on depositor behavior and bank fragility. \citet{iyer2012understanding} use depositor-level data to show that insurance, relationships, and social networks affect withdrawal decisions. \citet{egan2017deposit} show that uninsured deposit demand is sensitive to bank distress. \citet{jiang2024monetary} emphasize the interaction between uninsured deposits and balance-sheet losses in the 2023 banking stress. Our model provides a theoretical interpretation of these findings: such variables shift the distribution of depositor withdrawal incentives and affect how much latent fragility is present before a run.

Finally, the paper contributes methodologically by applying mean-field games to optimal stopping problems. Mean-field methods have been widely used in heterogeneous-agent macroeconomics and related continuous-time models, as in \citet{achdou2022income}, and have recently been used to study macroeconomic strategic complementarities \citep{alvarez2024price}. Stationary equilibria with entry and exit decisions are studied by \cite{miao2005optimal} and \cite{luttmer2007selection} to understand the size distribution of firms and their capital structure choices. Mean-field methods have also been applied to study the optimal timing problems. \cite{carmona2017mean} study a bank run model in a mean-field game setting, where the bank fails when its withdrawn share exceeds its liquidation value. Under a strategic complementarity condition, they establish the existence of mean-field equilibria. \cite{nutz2018mean} models the bank failure as the first jump time of a Cox process, whose intensity depends on a common, a private component, and the withdrawn share. Using the similar intensity-based approach to model bank failure risk, we obtain an easy to check condition for the strategic complementarity, which allows us to cast the bank run problem as a supermodular mean-field game studied by \cite{dianetti2021submodular, dianetti2023unifying}. Weak formulation of the mean-field games with optimal stopping, where agents randomize their stopping time, has been studied by \cite{carmona2017mean}, \cite{bouveret2020mean}, \cite{dumitrescu2021control}, and has been used to model energy transition by \cite{aid2021entry} and \cite{dumitrescu2024energy}. Contributing to this literature, we highlight the discontinuity of the mean-field dynamics in our bank run model and establish the uniqueness of threshold equilibrium when a common state is present. 

The rest of the paper is as follows. Section \ref{Section: model setup} presents the model, and Section \ref{Section: equilibrium existence} establishes equilibrium existence and ordering. Section \ref{Section: clustered withdrawals} develops the clustering mechanism and studies equilibrium selection through a common state. Section \ref{Section: robustness empirical relevance} extends the mechanism to continuous private states and discusses empirical implications. Section \ref{Section: conclusion} concludes with policy implications.

\section{Model}
\label{Section: model setup}

\subsection{Depositors, Bank, and Payoffs}

The model is in continuous time. A continuum of depositors $i \in [0,1]$ each begin with an endowment of \$1 at time $t=0$. All depositors are risk-neutral and deposit their endowments in a bank. At time $t$, depositor $i$ discounts future payoffs according to
\begin{equation}
\label{def:beta_i}
\beta^i_t = \beta(Z_t, X^i_t).
\end{equation}
The common state $Z_t$ is publicly observable and represents macroeconomic conditions affecting all depositors symmetrically. The private state $X^i_t$ represents depositor-specific conditions. We assume the private states $\{X^i_t\}_{i \in [0,1]}$ are i.i.d. across depositors and independent of $Z_t$.

The bank operates mechanically without making strategic decisions. At time $t=0$, it collects all depositor endowments and invests them in a long-term project. The bank pays interest at rate $r > 0$ to each depositor, which is consumed immediately rather than reinvested. Depositors also receive a convenience yield at rate $\delta \geq 0$ on their deposited funds, representing non-pecuniary benefits such as payment services.

Depositors are allowed to withdraw their funds at any time. We distinguish between two types of withdrawal: exogenous and endogenous. Exogenous withdrawal occurs at an idiosyncratic random time $\tau_{\nu}^i$ arriving at rate $\nu \geq 0$, when depositor $i$ experiences a liquidity shock and must withdraw the entire deposit for consumption. These shocks are independent across depositors. Depositor $i$ may also strategically withdraw at an endogenous time $\tau_w^i$, reflecting both changes in her discount rate that reduce patience and a coordination motive: the decision depends on beliefs about other depositors' withdrawal behavior, which affects the likelihood of bank failure.

Let
\begin{equation}
\label{def:m_int}
m_t = \int_0^1 \mathbbm{1}_{\{\tau_w^i \wedge \tau_{\nu}^i \leq t\}} \, di
\end{equation}
denote the proportion of depositors who have withdrawn by time $t$, which we refer to as the \emph{withdrawn share}.\footnote{The results of this paper do not depend on whether $m_t$ is unobservable. There is no hidden information or learning, so individual depositors can infer $m_t$ on the equilibrium path.} As more depositors withdraw, the bank faces greater difficulty meeting its obligations, raising the probability of failure. Specifically, bank failure occurs at a random time $t_{\eta}$ with instantaneous arrival rate $\eta m_t$, where $\eta \geq 0$ measures the sensitivity of failure risk to funding shortfalls. When the bank fails, remaining depositors receive a recovery value $\gamma \in [0,1)$, implying a loss of $1-\gamma$ per dollar deposited. While the bank remains solvent, any withdrawing depositor receives the full principal of \$1.

For a given withdrawn share, a representative depositor's optimization problem is
\begin{equation}
\label{pro:dep}
\sup_{\tau_w} \mathbb{E} \left[\int_0^{\tau_\nu \wedge \tau_w \wedge t_\eta} e^{- \int_0^t \beta_u \, du} (r+\delta) \, dt + e^{-\int_0^{\tau_\nu \wedge \tau_w} \beta_u \, du} \mathbbm{1}_{\{\tau_\nu \wedge \tau_w < t_\eta\}} + e^{-\int_0^{t_\eta} \beta_u \, du} \gamma \, \mathbbm{1}_{\{t_\eta < \tau_\nu \wedge \tau_w \}}\right],
\end{equation}
where $\beta_u = \beta(Z_u, X_u)$ and the superscript $i$ on $X_u$ is omitted to simplify notation. Before withdrawal or bank failure, the depositor receives a flow payoff of $r+\delta$ per unit time. The depositor receives her principal of \$1 upon withdrawal if this occurs before bank failure; otherwise, she receives only the recovery value $\gamma<1$.

\subsection{Information Structure}

Each depositor $i$ observes three sources of information when deciding whether to withdraw: the common state $Z_t$, her private state $X^i_t$, and the withdrawn share $m_t$. Importantly, results from the law of exact large numbers imply that private states are idiosyncratic and average out across the continuum of depositors with distributional symmetry.\footnote{See \cite{sun2006exact} and \cite{sun2009individual} for the exact law of large numbers among a continuum of essentially pairwise i.i.d. random variables. In particular, these papers establish the exact law of large numbers on an extended probability space. See \cite[Section 3]{nutz2018mean} for a summary and applications to mean-field games of optimal stopping.} This implies that the withdrawn share $m_t$ in \eqref{def:m_int} depends only on the common state $Z_t$:
\begin{equation}
\label{def:m_prob}
m_t = \mathbb{P}\left(\tau_w \wedge \tau_{\nu} \leq t \,\big|\, \mathcal{F}^Z_t\right),
\end{equation}
where $\tau_{\nu}$ and $\tau_w$ are the exogenous and endogenous withdrawal times for a representative depositor, and $\{\mathcal{F}^Z_t\}_{t\geq 0}$ denotes the filtration generated by $Z$.\footnote{When there is no common state, \eqref{def:m_prob} becomes $m_t = \mathbb{P}(\tau_w \wedge \tau_{\nu}\leq t)$, so $m$ is deterministic.}

\subsection{Mean-Field Equilibrium}

The optimal withdrawal decision of each depositor depends on the likelihood of bank failure, which is determined by the withdrawn share $m_t$. However, $m_t$ itself emerges from the aggregation of individual withdrawal decisions as in \eqref{def:m_int}. This creates a fixed-point problem. The mean-field game framework of \cite{lasry2007mean} and \cite{caines_large_2006} provides a tractable approach for analyzing such a model. Specifically, mean-field games characterize equilibrium through a fixed-point problem between a representative depositor and the aggregate withdrawn share, rather than by tracking individual strategies. We now define the equilibrium notion in our model.

\begin{definition}[Mean-Field Equilibrium]
\label{def:mfe}
A mean-field equilibrium is a pair $(\tau_w, m)$ such that:
\begin{enumerate}
\item\label{Cond:individual} Given the withdrawn share process $(m_t)_{t\geq 0}$, the withdrawal time $\tau_w$ solves the optimization problem \eqref{pro:dep} for a representative depositor.
\item\label{Cond:consistent} Given the optimal withdrawal strategy $\tau_w$, the withdrawn share $m$ is consistent with \eqref{def:m_prob}.
\end{enumerate}
\end{definition}

Condition \ref{Cond:individual} requires that each depositor optimally responds to aggregate withdrawal behavior, accounting for her liquidity needs and bank failure risk. Condition \ref{Cond:consistent} ensures aggregate behavior is consistent with individual optimal strategies.

\subsection{Discussion of Assumptions}

\paragraph{Deposit contract.} The depositor's payoff structure features two-sided strategic complementarity \citep{rochet2004coordination}: a depositor's withdrawal incentive always increases when more depositors have withdrawn. This contrasts with standard deposit contracts featuring only one-sided complementarity \citep{goldstein2005demand}. Our approach simplifies the problem and allows us to focus on dynamic aspects of bank runs. Hence, our setup constitutes a fully dynamic regime-shift game where regime change, namely bank failure, occurs endogenously based on aggregate behavior.

\paragraph{Heterogeneity.} While we introduce heterogeneity through time-varying discount rates, the model can be equivalently interpreted with alternative sources of heterogeneity in other factors, such as convenience yields, expected returns from the bank's investment, or returns from safe storage, as long as these factors drive the marginal withdrawal decision.\footnote{For instance, suppose all depositors share a common discount rate $\beta$ but experience heterogeneous convenience yields
\begin{equation}
\label{def:delta_heterogeneous}
\delta^i_t = \delta(Z_t, X^i_t),
\end{equation}
for a function $\delta(\cdot, \cdot)$. Alternatively, consider a bank investing in a long-term project maturing at random time $t_{\lambda}$ with arrival rate $\lambda$. When the project matures before the bank fails, it pays out to remaining depositors. Depositors have heterogeneous beliefs about the expected payoff:
\begin{equation}
\label{def:R_heterogeneous}
R^i_t = R(Z_t, X^i_t),
\end{equation}
for a function $R(\cdot, \cdot)$.} These specifications lead to equivalent mathematical structures.

\paragraph{Bank failure probability.} We assume the instantaneous bank failure rate is proportional to $m_t$, the withdrawn share. Our choice reflects the view that bank solvency depends on the total funding shortfall. When a larger proportion of depositors has exited, the bank holds fewer liquid assets to meet obligations, is more vulnerable to liquidity shocks, and may be forced to liquidate long-term investments at fire-sale prices. This specification makes the representative depositor's best response monotone in $m$: a higher withdrawn share lowers the value of waiting for all remaining depositors. This monotonicity is the key property used below to order equilibria and apply the fixed-point argument.

\section{Equilibrium Existence and Characterization}
\label{Section: equilibrium existence}
In this section, we establish the existence of mean-field equilibria and characterize their structure. We show that there exists an earliest-run equilibrium and a latest-run equilibrium, and that all other equilibria, if they exist, are sandwiched between these two extremes. We also show that stronger strategic complementarity, captured by a higher $\eta$, leads to earlier withdrawals and a larger withdrawn share.

\subsection{The Depositor's Optimal Withdrawal Problem}

To characterize mean-field equilibria, we first examine the representative depositor's optimal stopping problem \eqref{pro:dep}. For a given withdrawn share $(m_t)_{t\geq 0}$, we define the value function for \eqref{pro:dep} as
\begin{equation}\label{def:dep_V}
\begin{split}
V(t, Z_t, X_t; m):=  \sup_{\tau_w} \mathbb{E}_t \Big[&\int_t^{\tau_\nu \wedge \tau_w \wedge t_\eta} e^{- \int_t^s \beta_u \, du} (r+\delta) \,ds + e^{-\int_t^{\tau_\nu \wedge \tau_w} \beta_u \, du} \mathbbm{1}_{\{\tau_\nu \wedge \tau_w < t_\eta\}} \\
&+ e^{-\int_t^{t_\eta} \beta_u \, du} \gamma \, \mathbbm{1}_{\{t_\eta < \tau_\nu \wedge \tau_w \}}\Big],
\end{split}
\end{equation}
where $\mathbb{E}_t[\cdot]$ denotes the conditional expectation $\mathbb{E}\big[\cdot \,|\, t< \tau_{\nu} \wedge t_{\eta}\big]$. The following result presents an equivalent formulation that makes the depositor's tradeoff more transparent.

\begin{lem}
\label{lem:obj}
For a given withdrawn share $(m_t)_{t\geq 0}$,
 \begin{align}
 V(t, Z_t, X_t; m) = & \, 1 + \sup_{\tau_w} J (t, Z_t, X_t; m, \tau_w), \quad \text{where}\label{pro:os}\\
  J (t, Z_t, X_t; m, \tau_w):=&\, \mathbb{E}_t \Big[\int_t^{\tau_w} e^{-\int_t^s (\nu + \beta_u + \eta m_u) \, du} \Big(r+ \delta - \beta_s - (1-\gamma) \eta m_s \Big) \, ds\Big].\label{def:hat_J}
 \end{align}
\end{lem}

Lemma \ref{lem:obj} separates the value of immediate withdrawal, $1$, from $J$, the net benefit from staying in the bank. This net benefit weighs the interest and convenience yield from keeping funds deposited against two forces that make waiting costly: the depositor's own impatience and the expected loss if the bank fails. A higher withdrawn share lowers this net benefit because it raises the bank's failure intensity. Thus, withdrawals are strategic complements: when more depositors have already exited, remaining depositors have a stronger reason to exit as well.

For a given withdrawn share, the optimal withdrawal time $\tau_w$ for \eqref{pro:dep} may not be unique. We focus on the earliest and latest optimal stopping times, denoted by $\underline{\tau}_w$ and $\overline{\tau}_w$.\footnote{For the current generality, the optimal stopping time may not be unique. Therefore, we focus on the extreme stopping times in Lemma \ref{lemma monotonicity of brm}. These extreme stopping times are characterized using a representation theorem by \cite{bank2004stochastic}, which provides a convenient representation for the latest optimal stopping time.} The next lemma says that both move earlier when the withdrawn share is higher.

\begin{lem}\label{lemma monotonicity of brm}
Suppose that $\beta$ is bounded and
\begin{equation}
    \label{ass:Rmin}
   r+\delta + (1-\gamma) \nu -\gamma \beta_{\max} \geq 0, \quad \text{where} \quad \beta_{\max} := \sup_{z, x\in \mathbb{R}} \beta(z, x).
\end{equation}
Then the minimal $\underline{\tau}_w$ and maximal $\overline{\tau}_w$ optimal stopping times are nonincreasing with respect to the withdrawn share $m$. That is, if $\bar{m}_t \leq m_t$ for all $t \geq 0$, then
$\underline{\tau}_w(m) \leq \underline{\tau}_w(\bar{m})$ and $\overline{\tau}_w(m) \leq \overline{\tau}_w(\bar{m})$.
\end{lem}

The economic meaning of condition \eqref{ass:Rmin} becomes clearer if we split the promised repayment of $1$ into two components. The amount $\gamma$ is received whether the deposit ends in withdrawal or bank failure. The remaining amount, $1-\gamma$, is received only if withdrawal occurs before failure and therefore captures the value of obtaining full repayment. The net value of waiting in \eqref{def:hat_J} can be written as
\begin{equation}\label{decomp:1+J}
\begin{split}
&J(t, Z_t, X_t; m, \tau_w)=\\
& \underbrace{\mathbb{E}_{t}\Big[\int_{t}^{\tau_{w}}e^{-\int_{t}^{s}(\nu+\beta_{u}+\eta m_{u})du}\big(r+\delta+(1-\gamma)\nu-\gamma\beta_s\big)ds\Big]}_{\text{Net payoff while waiting}}
-\underbrace{(1-\gamma)\Big\{1-\mathbb{E}_{t}\Big[e^{-\int_{t}^{\tau_{w}}(\nu+\beta_{u}+\eta m_{u})du}\Big]\Big\}}_{\text{Value lost by postponing strategic withdrawal}}.
\end{split}
\end{equation}
The first term collects the net payoff received while the depositor waits. She receives the interest and convenience yield $r+\delta$. Because the recovery amount $\gamma$ is paid under either withdrawal or failure, delaying its receipt imposes a discounting cost $\gamma\beta_s$. An exogenous liquidity shock instead leads to full repayment and realizes the additional amount $1-\gamma$; because such shocks arrive at rate $\nu$, they contribute $(1-\gamma)\nu$ to the payoff from waiting.

The second term concerns the timing of strategic withdrawal. Withdrawing immediately secures the additional amount $1-\gamma$ now. If the depositor plans to withdraw at $\tau_w$, she receives this amount at that time only if neither an exogenous withdrawal nor bank failure has already occurred. Its expected discounted value is therefore $(1-\gamma)\mathbb{E}_t[e^{-\int_t^{\tau_w}(\nu+\beta_u+\eta m_u)du}]$. Condition \eqref{ass:Rmin} ensures that the net payoff in the first term is nonnegative in every state. We impose \eqref{ass:Rmin} throughout the rest of the paper to focus on bank-failure risk and coordination motives rather than withdrawals driven only by high impatience.

\subsection{Equilibrium Existence and Comparative Statics}

The monotonicity result in Lemma \ref{lemma monotonicity of brm} helps deliver equilibrium existence. If depositors expect a larger withdrawn share, they withdraw earlier; earlier withdrawals, in turn, generate a larger withdrawn share. This feedback admits fixed points, and the monotone structure lets us order them.

\begin{prop}
    \label{theorem existence minimal maximal MFGE}
    Suppose that \eqref{ass:Rmin} holds.
    There exist an earliest-run equilibrium $(\underline \tau_w, \underline m) $ and a latest-run equilibrium $(\overline \tau_w, \overline m) $ such that, for any other equilibrium $( \hat \tau_w,  \hat m)$,
    $$
   \underline \tau_w \leq \hat \tau_w \leq  \overline \tau_w \quad \text{and} \quad
   \underline m_{t} \geq \hat m _t \geq \overline m _{t}, \quad \text { for any } t\geq 0.
    $$
\end{prop}

We offer a heuristic proof for the existence of these equilibria by starting from an arbitrary withdrawal rule $\tau_{w}$ and letting $m_t(\tau_{w}) := \mathbb{P}(\tau_{\nu} \wedge \tau_{w} \leq t\,|\, \mathcal{F}^Z_t)$ be the implied withdrawn share. The best response map is
    \begin{equation*}
        \mathcal{R}(\tau_{w}) := \tau_w\big(m(\tau_{w})\big),
    \end{equation*}
where $m\mapsto \tau_w(m)$ maps an aggregate withdrawal path to the representative depositor's optimal withdrawal time. Fixed points of $\mathcal{R}$ are mean-field equilibria. Tarski's fixed point theorem then gives a smallest and largest fixed point in the natural ordering, which correspond to the earliest-run and latest-run equilibria. Proposition \ref{theorem existence minimal maximal MFGE} therefore gives an equilibrium interval. Runs cannot occur before the earliest-run equilibrium, and equilibria cannot be delayed beyond the latest-run equilibrium. The iterative scheme in Appendix \ref{app:m_schemes} constructs these two extremes and provides the numerical procedure used in the examples presented later.

The same argument also gives comparative statics in the strength of strategic complementarity, $\eta$. A larger $\eta$ means that each additional withdrawal has a stronger effect on the withdrawal incentives of the remaining depositors. As a result, stronger complementarity moves withdrawals earlier and raises the withdrawn share.

\begin{prop}\label{prop comparative statics strength parameter}
Let $(\underline{\tau}^{\eta}_w,\underline m^{\eta})$ and $(\overline{\tau}^{\eta}_w,\overline m^{\eta})$ denote the earliest-run and latest-run equilibria under $\eta$. If $\eta'\geq \eta$, then
\[
\underline{\tau}^{\eta'}_w \leq \underline{\tau}^{\eta}_w,
\quad
\overline{\tau}^{\eta'}_w \leq \overline{\tau}^{\eta}_w,
\]
and, for every $t\geq0$,
\[
\underline m^{\eta'}_t \geq \underline m^{\eta}_t,
\quad
\overline m^{\eta'}_t \geq \overline m^{\eta}_t.
\]
\end{prop}

\section{Clustered Withdrawals}
\label{Section: clustered withdrawals}

\subsection{A Two-Type Private State Model}
\label{Section: two type illustration}

We first illustrate the clustered-withdrawal mechanism in a simple environment with two private states and no common state. The depositor's discount rate is either $\beta_L$ or $\beta_H$, with $\beta_L<\beta_H$. We refer to depositors with discount rate $\beta_L$ as patient and depositors with discount rate $\beta_H$ as impatient. All depositors are patient at time $0$. A patient depositor becomes impatient at an idiosyncratic exponential time with rate $\theta_{LH}$. In this subsection, we focus on the case in which this transition is absorbing, so that the intensity, $\theta_{HL}$, of transition from impatient to patient is zero. State transitions are independent across depositors. We discuss the result under $\theta_{HL}>0$ at the end of this subsection.

The purpose of this example is to illustrate the economic force behind clustered withdrawals. We assume
\begin{equation}\label{ass:Dis_betaH}
r+\delta + (1-\gamma)\nu - \gamma \beta_H \geq 0
\end{equation}
so that \eqref{ass:Rmin} holds, and
\begin{equation}\label{ass:two_type_main}
\beta_{L}<r+\delta-(1-\gamma)\eta<\beta_{H}.
\end{equation}
The first part of the inequality in \eqref{ass:two_type_main} implies that patient depositors do not withdraw strategically even if all other depositors have already withdrawn. The second part implies that impatient depositors are run-prone: they may prefer to wait when few others have withdrawn but prefer to exit after sufficiently many withdrawals. When a depositor becomes impatient, she therefore does not need to withdraw immediately. Over time, the mass of impatient depositors grows in the background. This creates \emph{latent fragility}: a stock of run-prone depositors who are individually willing to wait as long as others wait, but who would all like to exit once enough others exit. Formally, if the withdrawn share is sufficiently high, the net benefit from staying, $r+\delta - \beta_H - (1-\gamma) \eta \, m_t$ in \eqref{def:hat_J}, becomes negative. Given that impatient depositors do not become patient again and $m_t$ is weakly increasing, the net benefit from staying does not return to positive once the withdrawn share exceeds the threshold 
\begin{equation}\label{def:m_star_two_type}
m^*:=\frac{r+\delta-\beta_H}{(1-\gamma)\eta}.
\end{equation}
When $m_t<m^*$, an impatient depositor is willing to stay; when $m_t \geq m^*$, it is optimal for an impatient depositor to withdraw.

Before any strategic withdrawal, the withdrawn share consists only of exogenous withdrawals and is equal to $1-e^{-\nu t}$. Meanwhile, a mass $e^{-\nu t}-e^{-(\nu+\theta_{LH})t}$ of depositors have become impatient but remain in the bank. This accumulated mass captures latent fragility: it is not yet included in $m_t$, but it determines how much the withdrawn share would jump if all impatient depositors withdrew together. 

In the earliest-run equilibrium, a cluster occurs as soon as the sum of exogenous withdrawals and accumulated impatient depositors reaches $m^*$. In the latest-run equilibrium, by contrast, accumulated impatient depositors continue to wait until exogenous withdrawals alone reach $m^*$. Thus, the latest cluster time $\overline{t}$ is determined by $1-e^{-\nu t}=m^*$. Between $\underline{t}$ and $\overline{t}$, impatient depositors prefer to stay if others stay but prefer to withdraw if the accumulated impatient depositors withdraw together. The next proposition formalizes this range of equilibrium cluster times.

\begin{prop}\label{prop:Poisson}
Suppose \eqref{ass:two_type_main} holds. If $\beta_H<r+\delta$, define
\begin{equation}\label{def:t_poisson}
 \underline{t}:= - \frac{1}{\nu+\theta_{LH}}
\log\left(1-\frac{r+\delta-\beta_H}{(1-\gamma)\eta}\right),
\quad
\overline{t}:=
\begin{cases}
- \frac{1}{\nu}\log\left(1-\frac{r+\delta-\beta_H}{(1-\gamma)\eta}\right), & \nu>0,\\
\infty, & \nu=0.
\end{cases}
\end{equation}
For any $\widetilde{t}\in[\underline{t},\overline{t}]$, define
\begin{equation}\label{def:tilde_m}
\widetilde{m}_s=
\begin{cases}
1-e^{-\nu s}, & s<\widetilde{t},\\
1-e^{-(\nu+\theta_{LH})s}, & s\geq \widetilde{t}.
\end{cases}
\end{equation}
Then any equilibrium is characterized by a cluster time $\widetilde{t}\in[\underline{t},\overline{t}]$ and the associated withdrawn share $\widetilde{m}$. At $\widetilde{t}$, all remaining impatient depositors withdraw together, generating a jump in $m$ of size
\[
e^{-\nu \widetilde{t}}-e^{-(\nu+\theta_{LH})\widetilde{t}}.
\]
After $\widetilde{t}$, depositors withdraw strategically as soon as they become impatient, because the withdrawn share is then higher and the failure risk only increases over time. If $\beta_H\geq r+\delta$, the equilibrium is unique, $\underline{t}=\overline{t}=0$, and depositors withdraw as soon as they become impatient.
\end{prop}

Proposition \ref{prop:Poisson} shows that withdrawal can occur at a cluster even though private shocks arrive idiosyncratically and smoothly over time. The reason is latent fragility. Before the cluster, impatient depositors prefer not to be the first to run, because their own withdrawal has a negligible effect on the withdrawn share. However, when the accumulated mass of impatient depositors is sufficiently high, if all of them withdraw together, their collective exit raises $m_t$ enough to make withdrawal optimal for each of them. The run is therefore not caused by a sudden change in individual states. It is caused by the coordination of a stock of already run-prone depositors.

The cluster date is reminiscent of the sunspot equilibrium in bank runs in a dynamic context. If run-prone depositors expect others to wait until $\widetilde{t}$, waiting is optimal before $\widetilde{t}$; if they expect others to withdraw at $\widetilde{t}$, withdrawing is optimal then. This is related to the self-fulfilling logic in \citet{diamond1983bank}. The difference is that the size of the run is not arbitrary: it is pinned down by the stock of run-prone depositors that has accumulated by the coordination date. Later coordination dates are associated with larger clusters because more run-prone depositors have accumulated inside the bank. The earliest-run equilibrium occurs at $\underline{t}$, the first time at which exogenous withdrawals and accumulated impatient depositors can jointly reach the cutoff $m^*$. The latest-run equilibrium occurs at $\overline{t}$, when exogenous withdrawals alone push the withdrawn share to $m^*$; if $\nu=0$, this latest-run date is infinite.  

\begin{figure}
\centering
\includegraphics[scale=0.6]{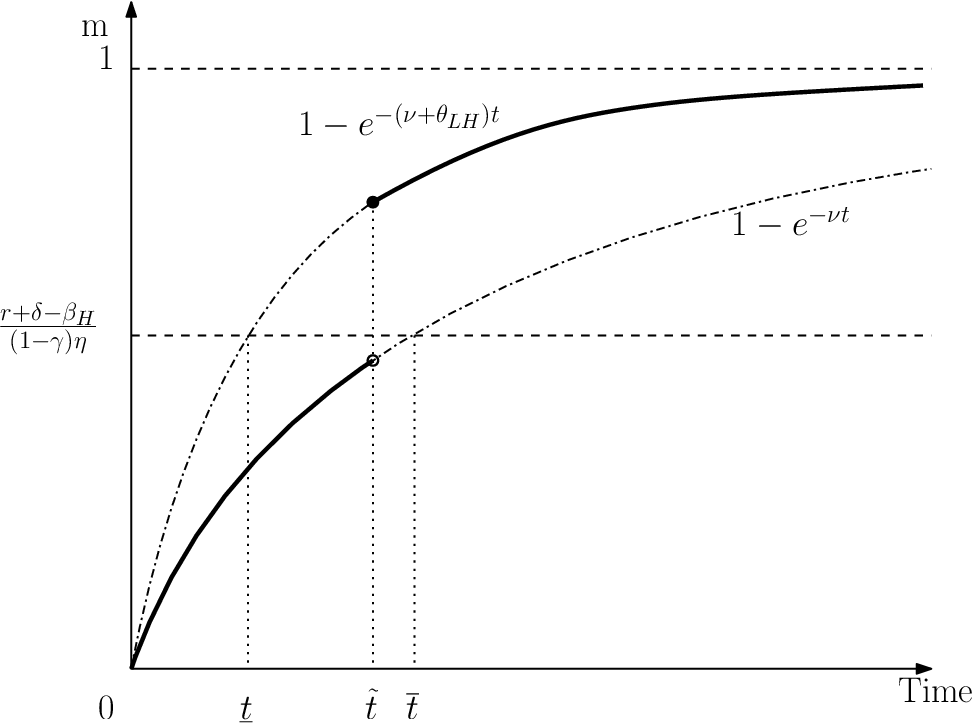}
\caption{Equilibrium withdrawn share}
 \begin{minipage}{\textwidth}
	\footnotesize
	\vspace{2mm}
    This figure presents the equilibrium withdrawn share in Proposition \ref{prop:Poisson} when $\beta_H<r+\delta$. The time $\underline{t}$ is the cluster time in the earliest-run equilibrium, and $\overline{t}$ is the cluster time in the latest-run equilibrium. The solid black line represents a generic equilibrium, whose cluster time, $\tilde{t}$, lies between $\underline{t}$ and $\overline{t}$.
    \end{minipage}
    \label{fig:betaL_low}
\end{figure}

Figure \ref{fig:betaL_low} plots one equilibrium path, with time on the horizontal axis and the withdrawn share on the vertical axis. Before the cluster, the withdrawn share grows only through exogenous withdrawals. At the cluster date, the accumulated impatient depositors withdraw together, so $m_t$ jumps discretely. Heterogeneous private states create a stock of run-prone depositors who may wait individually until $\tilde{t}$, and strategic complementarity turns this latent fragility into a clustered withdrawal once collective withdrawal becomes self-confirming.

\begin{cor}\label{cor:two_type_timing}
Suppose $\beta_H<r+\delta$ and $\nu>0$. The earliest cluster time $\underline{t}$ and the latest cluster time $\overline{t}$ are both decreasing in $\eta$ and $\beta_H$. In addition, $\underline{t}$ is decreasing in $\theta_{LH}$, while $\overline{t}$ is independent of $\theta_{LH}$.
\end{cor}

A larger $\eta$ raises strategic complementarity and lowers the cutoff $m^*$. A larger $\beta_H$ makes impatient depositors less willing to wait. Both forces bring the clustered withdrawal earlier in the earliest-run and latest-run equilibria. A larger $\theta_{LH}$ makes run-prone depositors accumulate faster, so it also brings the earliest-run equilibrium earlier. It does not affect the latest-run date because, in that equilibrium, the run is delayed until exogenous withdrawals alone push the withdrawn share to the cutoff.

We conclude this subsection by allowing impatient depositors to become patient again, so that $\theta_{HL}>0$. The possibility of returning to patience mitigates latent fragility: an impatient depositor may wait for her private state to improve rather than join a coordinated run. When there are no exogenous withdrawals and the return rate to patience is sufficiently high, this force can eliminate clustered withdrawals altogether.

\begin{prop}\label{prop:no_run}
Suppose that $\nu=0$ and $\beta_H > r+\delta - (1-\gamma) \eta \geq \beta_L$. If 
\begin{equation}\label{no_run_cond}
r+\delta - \beta_H - (1-\gamma) \eta \underbrace{\frac{\theta_{LH}}{\theta_{LH} + \theta_{HL}}}_{\text{largest impatient share}} + \theta_{HL} \underbrace{\frac{r+\delta - \beta_L - (1-\gamma) \eta \frac{\theta_{LH}}{\theta_{LH} + \theta_{HL}}}{\beta_L + \theta_{LH} + \eta \frac{\theta_{LH}}{\theta_{LH} + \theta_{HL}}}}_{\text{lowest value in the patient state}} > 0,
\end{equation}
then the unique equilibrium is $\tau^*_w =\infty$ and $m\equiv 0$.
\end{prop}

%To interpret \eqref{no_run_cond}, absent withdrawals, the impatient share converges to $\theta_{LH}/(\theta_{LH}+\theta_{HL})$, its largest possible value. The condition evaluates an impatient depositor's incentive to stay at this worst-case share. It requires that the prospect of returning to the patient state, represented by the last term, more than offsets the depositor's impatience and the failure risk created by a coordinated withdrawal of all impatient depositors. Hence, even in the earliest-run equilibrium, no impatient depositor wishes to initiate or join a run.

If $\nu>0$, exogenous withdrawals keep raising \(m_t\) until the remaining impatient depositors withdraw together. A return to patience can therefore delay, but not eliminate, a clustered withdrawal.

\subsection{Common State as a Coordination Device}
\label{Section: common state coordination}

The two-type model above shows that heterogeneous private states can create a stock of run-prone depositors and that strategic complementarity can turn this stock into a clustered withdrawal. In that environment, however, the timing of the cluster is not pinned down uniquely. Any cluster time between $\underline{t}$ and $\overline{t}$ can be sustained as an equilibrium. We now introduce an aggregate state that coordinates withdrawal timing.

Let the common state $Z_t$ affect the discount rates of patient and impatient depositors, denoted by $\beta_L(Z_t)$ and $\beta_H(Z_t)$, with $\beta_L(z)<\beta_H(z),\;\forall z$. We assume both discount rates are decreasing in $z$, so a lower value of $Z_t$ represents worse aggregate conditions and makes depositors less willing to wait. As in the previous subsection, all depositors are patient at time $0$, patient depositors become impatient at rate $\theta_{LH}$, and impatience is absorbing. Again, we impose
\begin{equation}
\label{ass:beta_H_common}
  r + \delta - (1-\gamma)\eta > \beta_L(z), \quad \text{for any } z,
\end{equation}
so patient depositors never strategically withdraw. For impatient depositors, we assume
\begin{equation}
    \label{ass:beta_L_common}
   r+\delta - (1-\gamma)\eta > \min_z \beta_H(z) \quad \text{and} \quad \max_z \beta_H(z) > r+\delta.
\end{equation}
Thus, when aggregate conditions are sufficiently good, impatient depositors prefer to stay even if the withdrawn share is high. When aggregate conditions are sufficiently bad, impatient depositors prefer to withdraw even if no one else has run. The interesting region is between these extremes, where withdrawal depends jointly on aggregate conditions, the mass of run-prone depositors, and the withdrawn share.

Let $\chi_t$ denote the mass of impatient depositors who remain in the bank. These depositors have not yet contributed to the withdrawn share $m_t$, but they may withdraw together. Thus, $m_t$ measures realized fragility, while $\chi_t$ measures latent fragility. We focus on threshold equilibria with a boundary $\cZ(\chi,m)$ such that an impatient depositor withdraws when the aggregate state falls below the boundary. Her strategic withdrawal time is
\[
\tau_w(H)=\inf\{t\geq \tau_\theta: Z_t\leq \cZ(\chi_t,m_t)\},
\]
where $\tau_\theta$ is the time at which she becomes impatient. The common state therefore does not make all depositors impatient at once. Instead, it coordinates when the already accumulated impatient depositors withdraw.

We motivate this equilibrium using a small withdrawal friction. Suppose impatient depositors receive idiosyncratic opportunities to withdraw at Poisson arrival times with intensity $\vartheta$. For finite $\vartheta$, these opportunities make withdrawals slightly asynchronous. As $\vartheta\rightarrow\infty$, withdrawal opportunities become arbitrarily frequent and depositors can choose any stopping time. The finite-$\vartheta$ economy therefore provides an equilibrium selection for the model without withdrawal frictions.

\begin{prop}[Threshold equilibrium: selection and uniqueness]\label{prop:threshold}
For every finite $\vartheta$, there exists a unique threshold equilibrium, and its threshold is nondecreasing in both $\chi$ and $m$. Suppose the finite-$\vartheta$ equilibria converge as $\vartheta\rightarrow\infty$ and the limiting threshold and value function satisfy the regularity conditions stated in Appendix \ref{app:inf_vartheta}. Then the limit is a threshold equilibrium in which depositors can withdraw at any time, and its threshold remains nondecreasing in $\chi$ and $m$. Moreover, there exists at most one threshold equilibrium satisfying these regularity conditions.
\end{prop}

Our main theoretical contribution relative to \cite{FrankelPauzner2000} is to establish equilibrium uniqueness when depositors can withdraw at any time and their withdrawals affect bank-failure risk. For finite $\vartheta$, we first adapt their coordination argument to our setting. The frictionless case, however, requires a different proof. With Poisson withdrawal opportunities, immediate withdrawal is not feasible. With arbitrary stopping times, immediate withdrawal is feasible and the depositor's value is constant throughout the withdrawal region, so the original uniqueness argument no longer applies. We establish uniqueness using a strong maximum principle. Proposition \ref{prop:threshold_finite} in Appendix \ref{app:finite_vartheta} states the finite-$\vartheta$ result formally, while Propositions \ref{proposition intensity limit} and \ref{prop:uniqueness_vartheta_inf} in Appendix \ref{app:inf_vartheta} provide the convergence and uniqueness results for arbitrary stopping times.

The monotonicity of the threshold has a direct economic interpretation. A higher withdrawn share $m$ means the bank is already more fragile, so a less severe aggregate shock is sufficient to trigger a run. A higher $\chi$ means that more impatient depositors have accumulated inside the bank. This also raises the threshold: the larger the stock of latent fragility, the easier it is for collective withdrawal to become self-confirming.

In the limiting equilibrium, when the threshold is crossed, all remaining impatient depositors withdraw at the same time. Thus, the withdrawn share follows
\begin{equation}
    \label{dyn:m_inf_del}
\begin{split}
    dm_t = & \nu(1-m_t)dt, \quad \text{ when } Z_t> \cZ(\chi_t,m_t), \\
    m_t = & 1- e^{-(\theta_{LH}+\nu)t},  \text{ when } Z_t \leq \cZ(\chi_t,m_t).
\end{split}
\end{equation}
When the aggregate state is above the threshold, depositors only withdraw exogenously. Suppose that the aggregate state hits the threshold at time $t$, the mass of patient depositors is $e^{-(\theta_{LH} + \nu)t}$. Therefore, the mass of run-prone depositors right before the hitting time is 
\[
\chi_{t-} = 1- e^{-(\theta_{LH} + \nu) t} - m_{t-}.
\]
All accumulated run-prone depositors coordinate and withdraw together at the hitting time, generating a clustered withdrawal of size $\chi_{t-}$, and the withdrawn share jumps to $1- e^{-(\theta_{LH} +\nu) t}$ at the hitting time. After the hitting time, if the aggregate state remains below the threshold, depositors withdraw immediately when they become impatient. Therefore, the withdrawn share follows $1-e^{-(\theta_{LH} + \nu)t}$.

\paragraph{Multiple Clustered Withdrawals.} The common state can generate multiple clustered withdrawals. After the first threshold crossing, all run-prone depositors withdraw and $m_t$ jumps to $1- e^{-(\theta_{LH} +\nu)t}$. If aggregate conditions remain bad, newly impatient depositors withdraw as soon as they become impatient, so $m_t$ and $1- e^{-(\theta_{LH} +\nu)t}$ move together. If aggregate conditions recover, however, strategic withdrawals stop. Patient depositors continue to become impatient, so $1- e^{-(\theta_{LH} +\nu)t}$ keeps rising, while $m_t$ grows only through exogenous withdrawals. The gap $1- e^{-(\theta_{LH} +\nu)t}-m_t =\chi_t$, which is the new stock of run-prone depositors, rebuilds. If the common state later falls back to the threshold, this newly accumulated stock withdraws together, generating another cluster.

\begin{figure}
\hspace{-1cm}
\includegraphics[scale=0.7]{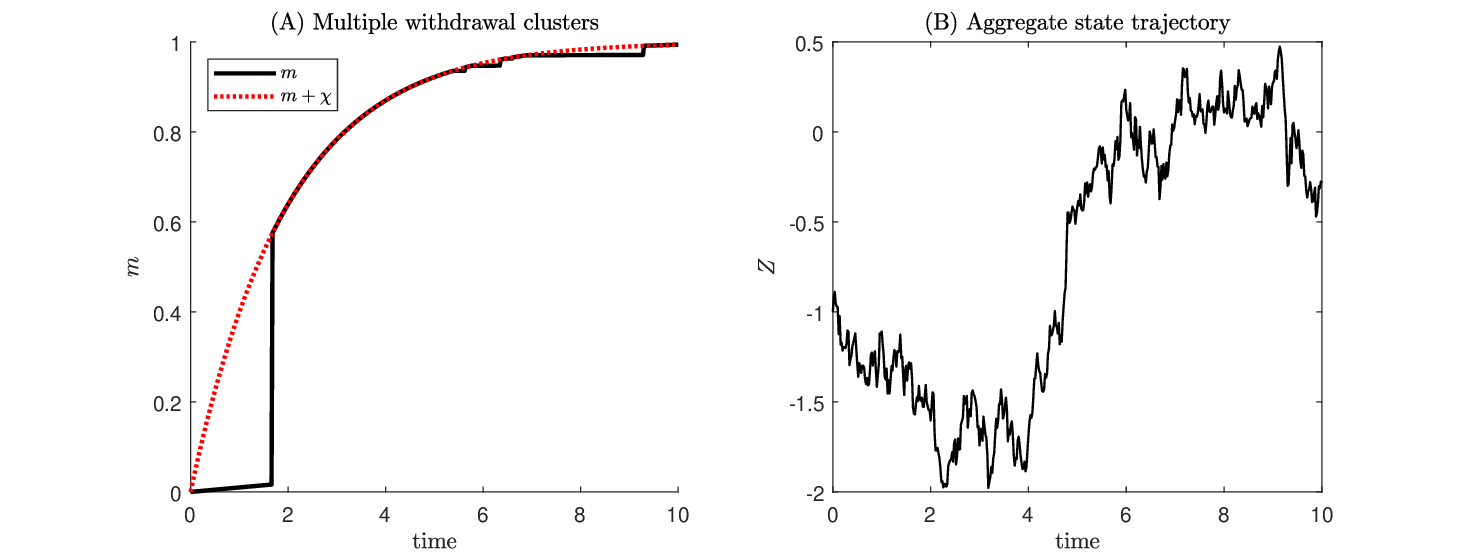}
\caption{Multiple clustered withdrawals}
 \begin{minipage}{\textwidth}
	\footnotesize
	\vspace{2mm}
    This figure presents an example of multiple withdrawal clusters. Panel (A) presents the dynamics of withdrawn share and Panel (B) shows the trajectory of the aggregate state. Parameters used are $\sigma = 0.5, r=0.05, \delta = 0.01, \nu = 0.01, \gamma = 0, \eta = 0.03$, $\theta_{LH} = 0.5$, $\theta_{HL}=0$, and $\vartheta = 10^4$. The threshold remains the same when $\vartheta$ is further increased. The discount rate for impatient depositors is $\beta_H(Z) = \beta_{\text{min}} + \frac{\beta_{\text{max}} -  \beta_{\text{min}}}{1+ e^Z}$, where $\beta_{\text{min}} = 0.01$ and $\beta_{\text{max}} = 0.08$.
    \end{minipage}
    \label{fig:threshold_sim}
\end{figure}

Figure \ref{fig:threshold_sim} illustrates an example of multiple clustered withdrawals. After the aggregate state hits the threshold for the first time and triggers a cluster, the aggregate state then recovers in several episodes, so run-prone depositors who arrive afterward do not immediately withdraw. During each recovery period, $\chi$ rebuilds. When the common state later deteriorates and touches the threshold again, the newly accumulated run-prone depositors withdraw together. The size of each cluster is therefore determined by how many depositors became run-prone between consecutive threshold crossings.

%This is the main difference from the two-type benchmark without an aggregate state. Without a common state, the model identifies an interval of possible cluster times. With a common state, aggregate conditions coordinate the timing of runs. The boundary tells depositors when waiting is no longer sustainable. Each time the common state crosses the boundary, the economy can experience a new clustered withdrawal among vulnerable depositors accumulated between consecutive threshold crossings.

\section{Robustness and Empirical Relevance}
\label{Section: robustness empirical relevance}

\subsection{Continuous State}
\label{Section: continuous state multiple clusters}

The binary private-state structure makes the mechanism transparent, but it also raises a natural question: are clustered withdrawals an artifact of having a mass of depositors with the same type? In this subsection, we show that the answer is no. Clustered withdrawals also arise when private states evolve continuously and depositors have continuously distributed discount rates.

For tractability, we focus on idiosyncratic private states and abstract from the common state.\footnote{When the common state is present, the distribution of discount rates becomes a stochastic flow of continuous probability measures, which is infinite dimensional. This stochastic flow of discount rate distribution becomes a state variable of depositors, whose value function satisfies the so-called master equation. Wellposedness and numerical solvers for master equations are challenging on-going research areas; see \cite{carmona2018probabilistic2}.} Depositor $i$'s private state follows
\begin{equation}
    \label{dyn:X_BM}
     dX^i_t = \sigma dB^i_t, \quad X^i_0 \sim \mathcal{L}_0, \quad i\in[0,1],
\end{equation}
where the Brownian motions $\{B_i\}_{i\in[0,1]}$ are mutually independent and $\mathcal{L}_0=\mathcal{N}(\mu_0,\sigma_0^2)$ is the initial distribution.\footnote{We use the Brownian motion dynamics for the private state as an example. The methodology works for other types of private state dynamics.} The discount rate is
\begin{equation}\label{BM_beta}
 \beta(X_t) = \beta_{\text{min}} + \frac{\beta_{\text{max}} - \beta_{\text{min}}}{1+ e^{\zeta X_t}},
\end{equation}
where $\zeta$ is a scaling factor. A lower private state corresponds to a higher discount rate and a stronger incentive to withdraw.

Because all uncertainty is idiosyncratic, the withdrawn share is deterministic in equilibrium. Given a path of withdrawn share $m$, the representative depositor's net benefit from staying is
\begin{equation}
    \label{def:V_BM}
    \mathcal{V}(t, X_t; m) = \sup_{\tau_w} \mathbb{E}_t \Big[\int_t^{\tau_w} e^{-\int_t^s (\nu + \beta(X_u) + \eta m_u) du} \Big( r + \delta - \beta(X_s) - (1-\gamma) \eta m_s\Big) ds\Big].
\end{equation}
The optimal withdrawal rule is characterized by an optimal withdrawal boundary $b_t$. Depositors withdraw when their private state falls below this boundary:
\[
\tau_w^*=\inf\{t\geq0: X_t\leq b_t\}.
\]
Thus, the continuous-type model replaces the patient/impatient distinction with a boundary that separates depositors who stay from depositors who withdraw.

The boundary is determined jointly with the withdrawn share. On a finite horizon $T$, the value function satisfies the free-boundary problem
\begin{align}
 0 = & \partial_t \mathcal{V} + \tfrac12 \sigma^2 \partial^2_{XX} \mathcal{V} - \big(\nu + \beta(X) + \eta m \big) \mathcal{V} + r +\delta - \beta(X) - (1-\gamma) \eta m, \quad X > b_t,\label{pro:free_boundary}\\
 0 = & \mathcal{V}(t, b_t; m), \quad  0 = \partial_X \mathcal{V}(t, b_t; m), \quad \text{and} \quad \mathcal{V}(T, X;m) =0, \label{bc:free}
\end{align}
The first equation is the HJB equation in the continuation region, where depositors still wait. The first two boundary conditions in \eqref{bc:free} are the value-matching and smooth-pasting conditions at the optimal stopping boundary. The remaining condition is the terminal condition.
For a given boundary, the density of remaining depositors, $g(t,X)$, evolves according to 
\begin{align}
 \partial_t g =& \tfrac12 \sigma^2 \partial^2_{XX} g - \nu \, g, \quad X > b_t,\label{eq:FP}\\
 g(0,X) = & \phi(X) \mathbbm{1}_{\{X>b_0\}}, \quad \text{and} \quad g(t,b_t)=0, \label{bc:FP}
\end{align}
The first equation is the Fokker-Planck equation for the density of surviving depositors who have not withdrawn. The absorbing boundary condition $g(t,b_t)=0$ reflects that depositors withdraw immediately once their private state reaches the boundary. Here, $\phi$ is the density of the initial distribution $\mathcal{L}_0$. The withdrawn share is then
\begin{equation}
    \label{FP:m}
    m_t = 1-\int_{b_t}^{\infty} g(t,X)\,dX.
\end{equation}
A mean-field equilibrium is a fixed point between the boundary $b_t$ and the withdrawn share $m_t$.

The economic mechanism is the same as before. Depositors gradually accumulate near the withdrawal boundary, creating latent fragility. If the boundary jumps upward at time $t$ from $b_{t-}$ to $b_t$, all depositors with private states between these two boundaries withdraw together. The withdrawn share therefore jumps even though individual private states and their cross-sectional distribution evolve continuously. Because the model contains only idiosyncratic shocks, no aggregate fundamental selects the time of this jump. Instead, the jump time acts as a self-fulfilling sunspot. If depositors expect a clustered withdrawal at time $t$, the most run-prone depositors withdraw, raising $m_t$ and making withdrawal optimal for depositors with slightly better private states. This response validates the anticipated upward jump $\Delta b_t=b_t-b_{t-}$ and coordinates a partial run among depositors near the withdrawal margin.

\begin{figure}
\centering
\begin{subfigure}{\textwidth}
\includegraphics[scale=0.65]{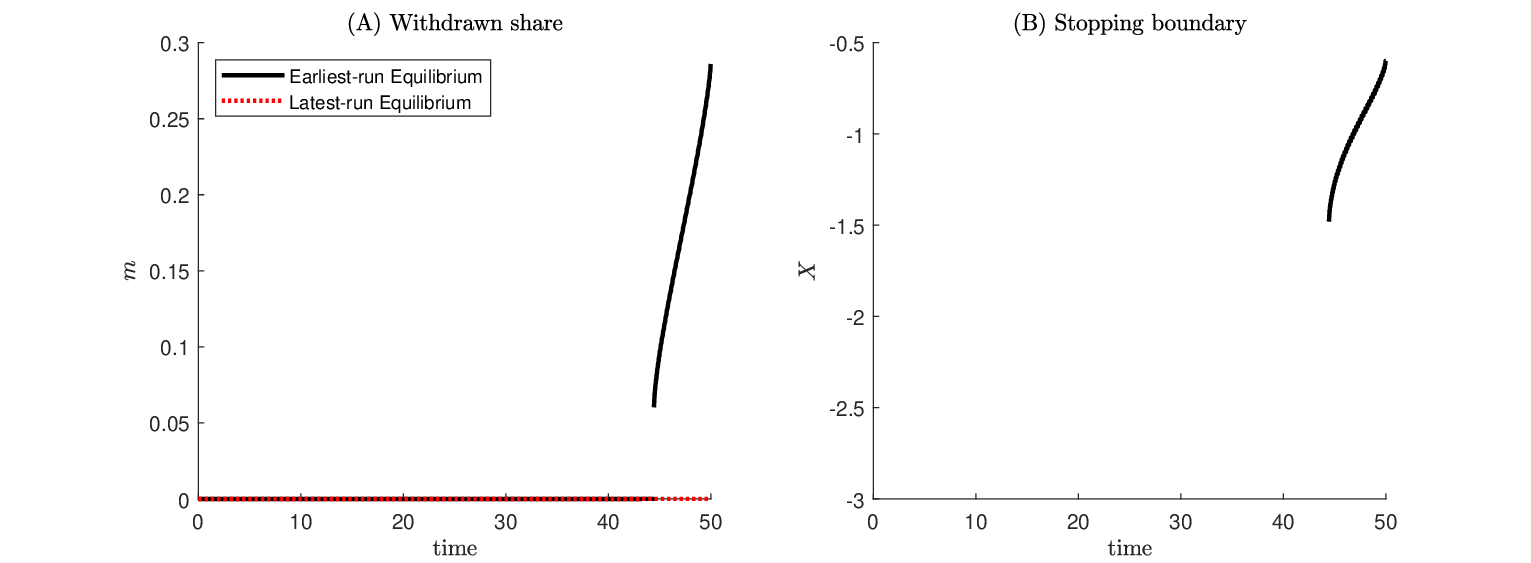}
\end{subfigure}

\begin{subfigure}{\textwidth}
\centering
\includegraphics[scale=0.65]{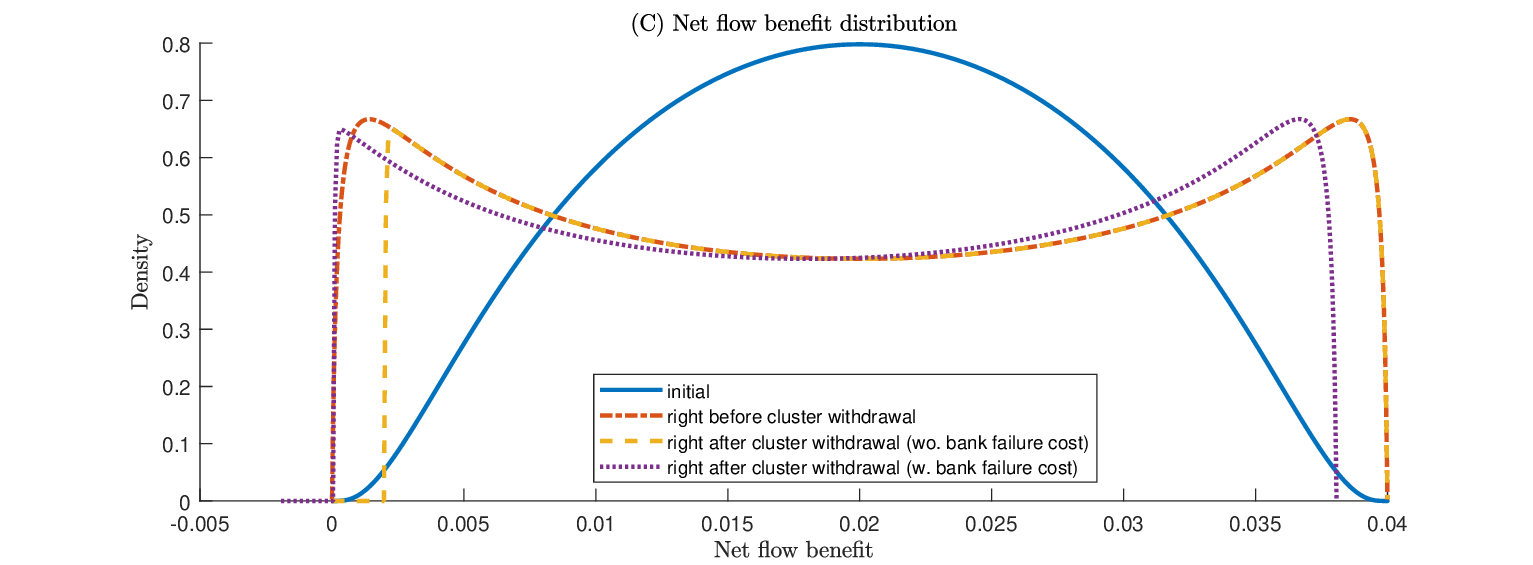}
\end{subfigure}
\caption{Equilibria with continuous private states}
 \begin{minipage}{\textwidth}
	\footnotesize
	\vspace{2mm}
    Panels (A) and (B) present the withdrawn share and the stopping boundary in the earliest-run and latest-run equilibria. Panel (C) presents densities of the net flow benefit, $r+\delta-\beta(X_t)-(1-\gamma)\eta m_t$, in the earliest-run equilibrium. The blue solid curve is the density at $t=0$, and the red dot-dashed curve is the density immediately before the cluster. The yellow curve is the density immediately after the cluster, after adding back the bank-failure term $\eta m_t$. The purple dotted curve is the actual post-cluster density. Parameters are $r=0.05, \delta=0.01, \nu=0, \gamma=0, \eta=0.032, \sigma=0.12, \mu_0=0, \sigma_0=0.5$, and $\zeta=2$. The function $\beta$ is given in \eqref{BM_beta} with $\beta_{\text{min}}=0.02$ and $\beta_{\text{max}}=0.06$. Because $r+\delta - \beta_{\text{max}} =0$ and exogenous withdrawal is absent, the latest-run equilibrium is $m\equiv 0$.
    \end{minipage}
    \label{fig:BM}
\end{figure}

Figure \ref{fig:BM} illustrates this mechanism. In the earliest-run equilibrium, the solid black lines in Panels (A) and (B) show that the withdrawal boundary and the withdrawn share jump at the same time. Panel (C) decomposes this jump. The blue solid curve gives the initial density of net flow benefits. By the instant before the cluster, shown by the red dot-dashed curve, depositors have accumulated near the withdrawal margin even though their private states remain continuously distributed. The yellow curve shows the distribution immediately after the cluster but adds back the bank-failure term $\eta m_t$. Because $m_{t-}=0$ in this example, comparing the red and yellow curves holds failure risk fixed. Their difference, concentrated at the left end of the distribution, is the mass of depositors who leave in the cluster. The purple dotted curve restores the bank-failure term $\eta m_t$ and therefore gives the actual distribution immediately after the cluster. It is the yellow curve shifted to the left by $\eta m_t$ and begins at zero. Thus, the initial withdrawals raise failure risk and reduce the net benefit of staying for every remaining depositor, validating the clustered withdrawal. The cluster is not caused by a mass point in private types; it arises from strategic complementarity acting on a continuous distribution of depositors near the withdrawal margin.

The continuous-state model also preserves the equilibrium ordering from Section \ref{Section: equilibrium existence}. There is an earliest-run equilibrium and a latest-run equilibrium. The earliest-run equilibrium has a higher withdrawal boundary and a larger withdrawn share, while the latest-run equilibrium has a lower boundary and a smaller withdrawn share. Stronger strategic complementarity, measured by $\eta$, shifts the boundary upward and brings the cluster earlier. Changes in private-state volatility and initial dispersion affect how much mass is near the boundary, and therefore affect the size and timing of the cluster; see the comparative statics in Appendix \ref{app:comp_statics}.

This section shows that clustered withdrawals are not an artifact of discrete private states. What matters is not whether depositors are exactly patient or impatient. What matters is that heterogeneous private incentives create a set of marginal depositors. Strategic complementarity then moves the withdrawal boundary in a way that can make these marginal depositors withdraw together.

\subsection{Empirical Implications}
\label{Section: empirical implications}

The model's most direct empirical implication is that smooth depositor-level shocks can generate discrete and partial jumps in aggregate withdrawals. High-frequency deposit-flow data should therefore exhibit concentrated withdrawal clusters rather than only gradual changes in average outflow rates. The response to the same aggregate news should also be state-dependent: it should be larger when more run-prone depositors have already accumulated inside the bank. Holding other conditions fixed, a later cluster should typically be larger because more depositors have had time to move toward the withdrawal margin.

This prediction makes the measurement of latent fragility central. Cluster participants should not be a random sample of depositors; they should be drawn disproportionately from those whose characteristics or recent shocks place them close to the withdrawal margin. The relevant characteristics depend on the source of stress rather than on a universal depositor category. Following bank-specific news, for example, depositors who become informed through social or information networks may become more run-prone. Following an increase in market rates, depositors with rate-sensitive accounts or better outside options may move closer to withdrawal. Empirically, one can estimate depositor-level withdrawal propensities using such characteristics and recent exposures, and use the predicted mass of depositors near the withdrawal margin as a measure of latent fragility. Insurance status, depositor relationships, and social networks are useful predictors in relevant settings \citep{iyer2012understanding,cookson2025social}, but their role should depend on the shock under study.

Latent fragility should be distinguished from the strength of strategic complementarity. The parameter $\eta$ captures how strongly withdrawals increase bank-failure risk and can be proxied by balance-sheet characteristics such as low liquidity, illiquid or long-duration assets, and unrealized losses. A higher $\eta$ brings a cluster earlier, but it need not make the cluster larger. Cluster size is determined by how many run-prone depositors have accumulated when coordination occurs. The model therefore predicts an interaction between depositor-level latent fragility and bank-level balance-sheet exposure: the same stock of run-prone depositors should generate an earlier response when withdrawals are more damaging to the bank. This distinction is related to evidence on the joint role of uninsured deposits and balance-sheet losses in bank fragility \citep{jiang2024monetary} and on the sensitivity of uninsured deposit demand to bank distress \citep{egan2017deposit}.

\section{Conclusion and Policy Implications}
\label{Section: conclusion}

This paper studies bank runs in a dynamic mean-field game with heterogeneous depositors. The main result is that runs can occur in clustered withdrawals. Heterogeneous private states create a stock of run-prone depositors whose incentives to withdraw are high, but who may still wait when the bank is sufficiently stable. Strategic complementarity then turns this latent fragility into a discrete run: once enough other depositors are expected to withdraw, the run-prone depositors withdraw together. A common state coordinates the timing of these clusters and can generate multiple waves of withdrawals. We also establish uniqueness among regular threshold equilibria when depositors can withdraw at arbitrary stopping times. The same clustering mechanism survives when private types are continuous, so clustered withdrawals are not an artifact of assuming that depositors are either patient or impatient.

Although we do not formally evaluate policy interventions, the model provides a useful way to interpret their economic channels. A policy can reduce runs by changing one of three objects: the stock of run-prone depositors, the strength of strategic complementarity, or the coordination device that turns latent fragility into a run. Deposit insurance mainly works by reducing the loss from bank failure for protected depositors. In the model, this raises the value of waiting and lowers the mass of depositors close to the withdrawal margin. It also weakens strategic complementarity because insured depositors have less reason to respond to other depositors' withdrawals. This is consistent with the classic role of deposit insurance in \citet{diamond1983bank}, while also emphasizing that partial insurance may leave a large uninsured or rate-sensitive depositor base exposed to run risk.

Suspension of convertibility works through a different channel. By preventing immediate withdrawals after a threshold is reached, suspension can break the feedback from current withdrawals to future failure risk. In our framework, this intervention can be represented by capping the jump in the withdrawn share. Such a policy may delay rather than eliminate a clustered withdrawal. By keeping run-prone depositors inside the bank, a temporary suspension can allow latent fragility to continue building and may lead to a larger cluster after the suspension is lifted.

Lender-of-last-resort policies operate by reducing the sensitivity of failure risk to withdrawals. In the model, this corresponds to lowering $\eta$. A smaller $\eta$ makes each withdrawal less damaging to the remaining depositors and therefore weakens strategic complementarity. As shown in Proposition \ref{prop comparative statics strength parameter}, this delays withdrawals and lowers the withdrawn share. This channel is close to the logic in \citet{rochet2004coordination}: liquidity support can prevent coordination failures when the bank is fundamentally viable but vulnerable to self-fulfilling withdrawals.

The framework also suggests that policy evaluation should focus not only on total outflows, but also on the timing and concentration of outflows. Policies that look similar in terms of cumulative withdrawals may differ sharply in whether they prevent a clustered withdrawal. A natural next step is to use the numerical model to compare deposit insurance, suspension of convertibility, and lender-of-last-resort support under the same primitives, focusing on the probability, timing, and size of withdrawal clusters. More broadly, low realized outflows need not imply low run risk: latent fragility may build while run-prone depositors continue to wait. Monitoring current withdrawals alone can therefore understate bank fragility.

\newpage{}
%%%%%%%%%%%%%%%%%%%%%%
%\section*{References}
%%%%%%%%%%%%%%%%%%%%%
\bibliographystyle{apalike}
\bibliography{references}

\newpage
%\pagenumbering{roman}\renewcommand{\thepage}{A\arabic{page}}

\setcounter{section}{0}
\renewcommand{\thesection}{APPENDIX \Alph{section}}
\renewcommand{\theequation}{\Alph{section}.\arabic{equation}}
\renewcommand{\thefigure}{\Alph{section}.\arabic{figure}}
\renewcommand{\thetable}{\Alph{section}.\arabic{table}}

\setcounter{equation}{0} 
\setcounter{figure}{0}
\setcounter{table}{0}

\numberwithin{equation}{section}
\numberwithin{figure}{section}

\appendix

\section{Monotone schemes}\label{app:m_schemes}

In order to identify the earliest-run and the latest-run equilibrium in Section \ref{Section: equilibrium existence}, we introduce a decreasing and an increasing iterative schemes, whose limits are the earliest-run and the latest-run equilibrium, respectively. For both schemes, we use the superscript to indicate the step of iterations. For example, $(m^{(n)}, \tau_w^{(n)})$ is the withdrawn share and the optimal withdrawal time for the $n$-th iteration. 

\begin{lem}\label{lem:mono_scheme}
Suppose that \eqref{ass:Rmin} holds.
\begin{enumerate}
\item[(i)] M-decreasing scheme: 
\begin{itemize}
\item[-] Initialize $m^{(0)}_t = 1$ for any $t\geq 0$, i.e., the entire population withdraws at time $0$. 
\item[-] Step $n$: let $\tau_w^{(n)}$ be the earliest optimal stopping time for the problem $\sup_{\tau_w} J(t, Z_t, X_t; m^{(n-1)}, \tau_w)$ and $m^{(n)}_t := \mathbb{P}\big(\tau_w^{(n)} \wedge \tau_{\nu}\leq t\,|\, \mathcal{F}^Z_t\big)$ for any $t$.
\end{itemize}
In this m-decreasing scheme,
\begin{equation}\label{m:dec}
m_t^{(n)} \geq m_t^{(n+1)}\quad \text{and} \quad \tau^{(n)}_w \leq \tau^{(n+1)}_w, \quad \text{ for any } t \text{ and } n.
\end{equation}
\item[(ii)] M-increasing scheme:
\begin{itemize}
\item[-] Initialize $m^{(0)}_t = 0$ for any $t\geq 0$, i.e., the entire population stays all the time.
\item[-] Step $n$: let $\tau_w^{(n)}$ be the  latest optimal stopping time for the problem $\sup_{\tau_w} J(t, Z_t, X_t; m^{(n-1)}, \tau_w)$ and $m^{(n)}_t := \mathbb{P}\big(\tau_w^{(n)} \wedge \tau_{\nu} \leq t\,|\, \mathcal{F}^Z_t\big)$ for any $t$.
\end{itemize}
In the m-increasing scheme,
\begin{equation}\label{m:inc}
m_t^{(n)} \leq m_t^{(n+1)} \quad \text{and} \quad \tau^{(n)}_w \geq \tau^{(n+1)}_w, \quad \text{ for any } t \text{ and } n.
\end{equation}
\end{enumerate}
\end{lem}

In either scheme, the previous result shows that $(\tau_w^{(n)})_{n}$ is a monotone sequence. Therefore, we can define the limit as 
\begin{equation}\label{eq:definition limits}
\tau^*_w := \lim_n \tau^{(n)}_w \quad \text{and} \quad m^*_t := \mathbb{P} \big(\tau^*_w  \wedge \tau_{\nu} \leq t \,|\, \mathcal{F}^Z_t\big). 
\end{equation}

The following result shows that the limits of the m-decreasing and the m-increasing scheme are the earliest-run and the latest-run equilibrium, respectively.

\begin{prop}\label{prop:convergence}
$\,$
\begin{enumerate}
\item[(i)] In the m-decreasing scheme, $(\tau^*_w, m^*)$ is the earliest-run equilibrium.
\item[(ii)] In the m-increasing scheme, $(\tau^*_w, m^*)$ is the latest-run equilibrium.
\end{enumerate}
\end{prop}

\section{Microfoundation for the threshold equilibrium}\label{app:microfoundation}
In order to microfound the equilibrium selection in Section \ref{Section: common state coordination}, we follow \cite{FrankelPauzner2000} to consider idiosyncratic withdrawal opportunities for depositors. We restrict the withdrawal opportunity for each depositor to be the arrival time $\{\tau_n\}_{n\geq 1}$ of a Poisson process with the intensity $\vartheta$. In this case, a depositor can only withdraw at one of $\{\tau_n\}_{n\geq 1}$, which arrives at an exponential random time (with the parameter $\vartheta$) after the previous one. The sequence of arrival time is i.i.d. across different depositors. This withdrawal friction can be interpreted as individual inattentiveness. When $\vartheta$ increases, the individual withdrawal opportunities arrive more often. 

We first examine the finite $\vartheta$ case in Section \ref{app:finite_vartheta} and establish the existence and uniqueness of threshold equilibrium in Proposition \ref{prop:threshold_finite}. Then we study the limiting case $\vartheta \rightarrow \infty$ in Section \ref{app:inf_vartheta}.

\subsection{Finite withdrawal intensity $\vartheta$}\label{app:finite_vartheta}

Following \cite{FrankelPauzner2000}, we will establish the unique threshold equilibrium for the finite $\vartheta$ case. Moreover, the equilibrium threshold is weakly increasing in withdrawn share and the mass of remaining impatient depositors.

Conditions \eqref{ass:beta_H_common} and \eqref{ass:beta_L_common} are assumed throughout this section. The common state $Z$ is modelled by a continuous Markov process. We assume that it satisfies 
\begin{equation}
    \label{ass:Z}
    Z^{t,z}_s = z + Z^{t,0}_s, \quad \text{ for any } s \geq t.
\end{equation}
Here $Z^{t,z}$ represents the common state which starts from $z$ at time $t$. This condition is satisfied when $Z$ has independent increments. All depositors are patient at time 0, but receive discount rate shocks at i.i.d. exponential time with the parameter $\theta_{LH}$ and become impatient afterwards. The private state of impatience is absorbing. 

Given \eqref{ass:beta_H_common}, patient depositors never strategically withdraw early. A mean-field equilibrium is characterized by the strategic withdrawal time of impatient depositors $\tau_w(H)$ and the withdrawn share $m$, which satisfy the following system
\begin{align}
\tau_w(H) =& \argmax_{\tau \in \{\tau_n\}_n} \mathbb{E} \Big[\int_0^{\tau} e^{-\int_0^s(\nu + \beta_H(Z_u) + \eta m_u) du} \Big(r+ \delta - \beta_H(Z_s) - (1-\gamma) \eta m_s \Big)ds\Big], \label{mfg:opt_fin_del}\\
m_t = & \mathbb{P}\big( \tau_w(H) \wedge \tau_{\nu} \leq t \,|\, \mathcal{F}^Z_t\big).\label{mfg:m_fin_del}
\end{align}

An impatient depositor decides the strategic withdrawal time using several state variables: the withdrawn share $m_t$, the mass of remaining impatient depositors, denoted by $\chi_t$, and also the common state $Z_t$. The withdrawn share indicates the proportion of population who has already withdrawn, it impacts the bank failure risk. The mass of remaining impatient depositors measures the magnitude of coordination motive, these impatient depositors could withdraw together, inducing a clustered withdrawal. In the current setting with withdrawal friction, individual impatient depositors need to wait for the arrival of their own withdrawal opportunity. They withdraw at the first $\tau_n$ when the common state $Z_{\tau_n}$ is sufficiently low. 

Introduce $h_t := \chi_t + m_t$, which represents the sum of withdrawn share and the mass of impatient depositors. Then $h_t = 1-p_t$, where $p_t$ is the mass of remaining patient depositors. Given that patient depositors experience exogenous liquidity shocks at the rate $\nu$ and the discount rate shock at the rate $\theta_{LH}$, $p_t$ follows the dynamics 
\[
dp_t = - (\theta_{LH} + \nu) p_t dt, \quad p_0 =1,
\]
which admits an explicit solution $p_t = e^{-(\theta_{LH} +\nu)t}$. Therefore, $h_t = 1- e^{-(\theta_{LH} + \nu)t}$. Given the explicit deterministic dynamics of $h_t$, it is convenient to use $(h_t, m_t, Z_t)$ as the state variables for impatient depositors. In this section, we will work with the state variables $(h_t, m_t, Z_t)$ then translate the results to the state variables $(\chi_t, m_t, Z_t)$, which has an 1-to-1 relation with $(h_t, m_t, Z_t)$.

We consider Markovian threshold equilibria. Given a threshold $\cZ: [0,1]^2 \rightarrow \mathbb{R}$, consider a threshold strategy where a representative impatient depositor withdraws at the first withdrawal opportunity when the common state $Z_t$ is lower than the threshold $\cZ(h_t, m_t)$, i.e., 
\begin{equation}\label{threshold_str}
\tau_w(H; h, m, \mathcal{Z}) = \inf\{\tau_n \geq \tau_{\theta_{LH}}\,:\, Z_{\tau_n} \leq \cZ(h_{\tau_n}, m_{\tau_n})\}.
\end{equation}
Here $\tau_{\theta_{LH}}$ is the arrival time of the discount rate shock when a patient depositor becomes impatient.

Next, we derive the dynamics of $m_t$ when impatient depositors choose the threshold strategy \eqref{threshold_str}. To this end, the threshold strategy indicates that the remaining impatient depositors (with mass $h_t - m_t= \chi_t$) withdraw at the rate $\vartheta$ whenever the common state $Z_t$ is lower than the threshold $\cZ(h_t, m_t)$, and the remaining patient and impatient depositors (with mass $1-m_t$) also withdraw exogenously at the rate $\nu$. Therefore, denote the withdrawn share associated with the threshold strategy in \eqref{threshold_str} by $m^{\cZ}$, we expect that it follows the dynamics
\begin{equation}\label{dyn:m_fin_del}
    d m^{\cZ}_t = \underbrace{\vartheta \mathbbm{1}_{\{Z_t \leq \cZ(h_t, m^{\cZ}_t)\}} (h_t - m^{\cZ}_t) dt}_{\text{strategic withdrawal from remaining impatient depositors}} + \underbrace{\nu (1-m^{\cZ}_t)dt}_{\text{exogenous withdrawal}}.
\end{equation}
If the common state $Z$ starts from $z$ at time $t$, we denote $m^{z,\cZ} =\{m^{z, \cZ}_s\}_{s\geq t}$ to highlight the initial value of $Z$. If we do not emphasize the initial value of $Z$, we suppress the superscript and write $\{m^{\cZ}_s\}_{s\geq t}$.

The following result confirms this intuition behind the dynamics \eqref{dyn:m_fin_del}.
\begin{lem}\label{lem:dyn_m_fin_del}
Given a threshold $\mathcal{Z}: [0,1]^2 \rightarrow \mathbb{R}$, consider a threshold strategy in \eqref{threshold_str}. Then $m_t$ in \eqref{mfg:m_fin_del} follows the dynamics \eqref{dyn:m_fin_del}. 
\end{lem}

Now, a representative impatient depositor treats the withdrawn share $m^{\cZ}_t$ as a state variable, together with two other state variables $h_t$ and $Z_t$, her value function is 
\begin{equation}
    \label{V_Z}
V_H^{\cZ}(h,m,z) := \sup_{\tau \in\{\tau_n\}_n} \mathbb{E}_{h, m, z}\Big[\int_t^{\tau} e^{-\int_t^s (\nu + \beta_H(Z^z_u) + \eta m^{z,\cZ}_u) du} \Big(r+\delta - \beta_H(Z^z_s) - (1-\gamma) \eta m^{z,\cZ}_s \Big) ds \Big], 
\end{equation}
where $\mathbb{E}_{h,m,z}[\cdot]$ represents the conditional expectation $\mathbb{E}[\cdot \,|\, h_t = h, m^{z,\cZ}_t = m, Z^z_t = z, t< \tau_{\nu} \wedge t_{\eta}]$ and $m \in [1- e^{-\nu t}, h]$, i.e., $m$ is  between the mass of exogenous withdrawal and the $h$. 

The following result reports the monotonicity of $V^{\cZ}_H$ in $z, h$, and $m$.
\begin{lem}\label{lem:VZ_mono}
 The value function $V^{\cZ}_H$ is strictly increasing in z. When the threshold function $\cZ$ is weakly increasing in $h$ and $m$, then $V^{\cZ}_H$ is weakly decreasing in $h$ and $m$. 
\end{lem}
For the monotonicity in $z$, recall that the assumption \eqref{ass:Z} implies a higher future value $Z^{z}_s$ from a higher initial value $z$. Hence, less future events where the threshold $\cZ(h_s, m^{z,\cZ}_s)$ is breached, and a lower future value of the withdrawn share $m_s^{z,\cZ}$. Meanwhile, $\beta_H$ decreases with the common states, implying a lower future discounting from higher future common states. Therefore,  $z \mapsto V^{\cZ}_H(h,m,z)$ is increasing. For the monotonicity in $h$, a higher initial value of $h_t$ implies a higher future value. Because the threshold $\cZ$ increases with $h$, the future value of threshold is higher, hence easier to breach. Meanwhile, a higher $h_t$ also increases $h_t - m^{z,\cZ}_t$. Both forces increase the growth of the withdrawn share $m^{z,\cZ}$. Since the objective function $J$ in \eqref{decomp:1+J} decreases with the withdrawn share, $h \mapsto V^{\cZ}_H(h,m,z)$ is weakly decreasing. Similarly, the monotonicity of $V_H^{\cZ}$ in $m$ follows from the decreasing property of $J$ in $m^{z,\cZ}$ and the fact that a higher current value in $m^{z,\cZ}$ leads to a higher future value, due to the Markov dynamics of $m^{z,\cZ}$. 

Given the monotonicity of $V^{\cZ}_H$ in $z$, we define an indifference surface for a representative impatient depositor:
\begin{equation}
    \label{def:indiff_Z}
    \widetilde{\cZ}(h,m) := \inf\big\{z \,:\, V_H^{\cZ}(h,m, z) =0 \big\}, \quad \text{for } h\in[0,1] \text{ and } m \in [1-e^{-\nu t}, h].
\end{equation}
The monotonicity of $V_H^{\cZ}$ in Lemma \ref{lem:VZ_mono} implies that the indifference surface is weakly increasing in $h$ and $m$. Moreover, the monotonicity in $z$ implies
\[
V_H^{\cZ}(h,m,z)>0, \quad \text{for } z> \widetilde{\cZ}(h,m) \quad \text{and} \quad V_H^{\cZ}(h,m,z)<0, \quad \text{for } z< \widetilde{\cZ}(h,m).
\]
Therefore, when the common state is higher than the indifference surface, it is optimal for the impatient depositor to wait; when the common state drops below the indifference surface, it is optimal for the impatient depositor to withdraw at the next $\tau_n$, if the common state is still below the indifference surface then.  

The threshold strategy in \eqref{threshold_str} leads to an equilibrium if it is also optimal. This would be the case if the threshold $\cZ$ used to define the strategy \eqref{threshold_str} is also the indifference surface for the impatient depositor. This leads us to the following definition of a threshold equilibrium.
\begin{definition}
    \label{def:threshold_equilibrium}
    A threshold $\cZ : [0,1]^2 \rightarrow \mathbb{R}$ is an equilibrium threshold if 
    \begin{equation}
        V_H^{\cZ}\big( h, m, \cZ(h,m)\big) =0, \quad \text{ for any } h\in[0,1] \text{ and } m \in [1-e^{-\nu t}, h],
    \end{equation}
    where $V_H^{\cZ}$ is defined in \eqref{V_Z}.
\end{definition}

Extending the proof of \cite{FrankelPauzner2000} to our setting with potential bank failure, we obtain the following result for threshold equilibria. 

\begin{prop}
    \label{prop:threshold_finite}
    When $\vartheta$ is finite, there exists a unique threshold equilibrium. The equilibrium threshold is weakly increasing in $h$ and $m$.
\end{prop}

Given the equilibrium threshold in $(h,m)$, we can change the variable to introduce the equilibrium threshold in $(\chi,m)$:
\[
\widehat{\mathcal{Z}}(\chi,m) = \mathcal{Z}(\chi+m,m).
\]
Given that $\mathcal{Z}$ is weakly increasing in $h, m$, $\widehat{\mathcal{Z}}$ also weakly increasing in $\chi$ and $m$.

In the absence of the common state in the previous section, impatient depositors can coordinate their withdrawal at any time between the earliest and the latest equilibria. The common state serves as a coordination device, which eliminates the equilibrium multiplicity. 

\subsection{The limiting case $\vartheta \rightarrow \infty$}\label{app:inf_vartheta}

For any $\vartheta>0$, let $\cZ^{(\vartheta)}$ be the unique threshold equilibrium as in Proposition \ref{prop:threshold_finite}. The following verification result states that if $\cZ^{(\vartheta)}$ converges to a threshold function $\cZ^{(\infty)}$ in a well behaved way, then $\cZ^{(\infty)}$ is  a threshold equilibrium for the limiting case.  

\begin{prop}\label{proposition intensity limit}
Let $\mathcal D = \{ (h,m) \in [0,1]^2 \, : \, m \leq h \}$ and suppose that: 
 \begin{enumerate}
\item $\cZ^{(\vartheta)}$ converges to some locally Lipschitz function  $\cZ^{(\infty)}$ locally uniformly in $\mathcal D $, as $\vartheta \rightarrow \infty$. 
 \item $V^{(\vartheta)}_H$  converges to some function $V^{(\infty)}_H$  locally uniformly in $ \mathcal D \times \mathbb{R}$ as $\theta \rightarrow \infty$. 
 \end{enumerate} 
 Then $\cZ^{(\infty)}$ is a threshold equlibrium in which each depositor can withdraw at any time.
\end{prop}

As $\vartheta$ approaches infinity, the dynamics of $m^{\cZ}_t$ in \eqref{dyn:m_fin_del} becomes increasingly singular. Intuitively, for a given threshold $\cZ$, it converges to 
\begin{equation}
    \label{dyn:m_inf_del_app}
    \begin{aligned}
        dm_t^{\cZ} = \nu (1-m_t^{\cZ}) dt, \text{ when } Z_t > \cZ(h_t, m_t^{\cZ}), \quad \text{and}
        \quad m_t^{\cZ} = h_t, \text{ when } Z_t \leq \cZ(h_t, m_t^{\cZ}), 
    \end{aligned}     
\end{equation}
with $m^{\cZ}$ jumping from $m^{\cZ}_{\tau-}$ to $h_{\tau}$ at entry times $\tau$ into the withdrawal region
$\{(h, m, Z)\,:\, Z \leq \cZ (h,m^\cZ)\}$. This means that only exogenous withdrawal happens when the common state is higher than the threshold $\cZ$ and all remaining impatient depositors, with mass $h_{\tau} - m^{\cZ}_{\tau-}$, withdraw together once the threshold is hit at time $\tau$.  In the proof of  Proposition \ref{proposition intensity limit}, we prove that $m^{\cZ^{(\infty)}}$, associated with the limiting threshold $\cZ^{(\infty)}$, indeed follows the dynamics \eqref{dyn:m_inf_del_app}.

For the illustration in Figure \ref{fig:threshold_sim}, we need to first identify the equilibrium threshold $\cZ^{(\infty)}$. To this end, we first numerically identify the equilibrium threshold $\cZ^{(\vartheta)}$ for a finite $\vartheta$. When $\vartheta$ increases beyond a sufficiently high level $\hat{\vartheta}$, the updating of the threshold $\cZ^{(\vartheta)}$ stops, hence the value function $V^{(\vartheta)}_H$ remains the same among all $\vartheta \geq \hat{\vartheta}$. Therefore, we verify the convergence assumptions in Proposition \ref{proposition intensity limit} numerically. 

The limiting-existence statement in Proposition \ref{prop:threshold} follows from Proposition \ref{proposition intensity limit}; monotonicity follows from Proposition \ref{prop:threshold_finite} and convergence; and uniqueness follows from Proposition \ref{prop:uniqueness_vartheta_inf}, presented below.   

Is the limiting threshold equilibrium unique? To answer this question, we consider the case where the common state $Z$ is a drifted Brownian motion, i.e.,
\begin{equation}\label{Z:drift_BM}
Z^{t,z}_s = z + \mu (s-t) + \sigma (B^c_s - B^c_t), 
\end{equation}
where $\sigma>0$ and $B^c$ is a one-dimensional standard Brownian motion.\footnote{The dynamics \eqref{Z:drift_BM} ensures that $Z^{t,z}$ satisfies \eqref{ass:Z}. The uniqueness result presented later also works when $\mu$ and $\sigma$ are time-dependent, i.e., $Z^{t,z}_s = z + \int_t^s \mu(u) du + \int_t^s \sigma(u) dB^c_u$ with non-degenerate volatility $\inf_u|\sigma(u)|\geq \epsilon$ for some $\epsilon>0$.} 

For a given threshold $\cZ: \mathcal{D} \rightarrow \mathbb{R}$, consider the dynamics of $m^{\cZ}$ in \eqref{dyn:m_inf_del_app}. If $Z$ starts from $z$ at $t$, we denote by $m^{z,\cZ}$ to highlight the initial condition of $Z$. Recall that $h$, the sum of withdrawn share and the mass of remaining impatient depositors, follows the dynamics
\begin{equation}\label{dyn:h}
dh_t = (\nu + \theta_{LH}) (1-h_t) dt.
\end{equation}
Consider the optimal stopping problem 
\begin{equation}
    \label{V_Z_inf}
V_H^{\cZ}(h,m,z) := \sup_{\tau} \mathbb{E}_{h, m, z}\Big[\int_t^{\tau} e^{-\int_t^s (\nu + \beta_H(Z^z_u) + \eta m^{z,\cZ}_u) du} \Big(r+\delta - \beta_H(Z^z_s) - (1-\gamma) \eta m^{z,\cZ}_s \Big) ds \Big], 
\end{equation}
where $\tau$ can be any stopping time, not necessarily a Poisson arrival time $\{\tau_n\}_n$. The threshold $\cZ$ corresponds to an equilibrium if $\cZ$ is also the optimal stopping boundary of \eqref{V_Z_inf}, i.e.,
\[
\cZ(h,m) = \inf\{z \,:\, V^{\cZ}_H(h,m,z) =0\}, \quad \text{for any } 0\leq m\leq h\leq 1.
\]

Next, we present our result on the uniqueness of threshold equilibria in the case where depositors can withdraw anytime. Let $\cZ$ be a threshold equilibrium. Given the dynamics of $m^{\cZ}$, $Z$, and $h$ in \eqref{dyn:m_inf_del_app}, \eqref{Z:drift_BM}, and \eqref{dyn:h}, we expect that $V^{\cZ}_H$ satisfies a free boundary problem: The state space is split into the continuation region $\mathcal{C}$ and the stopping region $\mathcal{S}$:
\[
\mathcal{C} = \{(h,m,z)\,:\, z > \cZ(h,m)\} \quad \text{and} \quad \mathcal{S} = \{(h,m,z)\,:\, z \leq  \cZ(h,m)\}.
\]
On $\mathcal{C}$, $V^{\cZ}_H$ satisfies the following equation
\begin{equation}\label{PDE:V_H}
    \big(\mathcal{L} - q(m,z) \big) V^{\cZ}_H(h,m,z) + f(z,m) =0,
\end{equation}
where 
\begin{equation}\label{def:cL}
\begin{split}
 \mathcal{L} := & (\theta_{LH} +\nu) (1-h) \partial_h + \nu(1-m) \partial_m + \mu \partial_z + \frac12 \sigma^2 \partial^2_{zz},\\
 q(z,m) := & \nu + \beta_H(z) + \eta m,\\
 f(t,m) := & r+\delta - \beta_H(z) - (1-\gamma) \eta m.
\end{split}
\end{equation}
On $\mathcal{S}$, $V^{\cZ}_H \equiv 0$.
At the optimal stopping boundary, $\cZ(h,m)$, we also expect that $V^{\cZ}_H$ satisfies the smooth-pasting condition 
\begin{equation}
    \label{V_smooth_pasting}
    \lim_{z\downarrow \cZ(h,m)} \partial_z V^{\cZ}_H(h,m,z) = \lim_{z\uparrow \cZ(h,m)} \partial_z V^{\cZ}_H(h,m,z) =0, \quad \text{for any } 0\leq m\leq h\leq 1.
\end{equation}

Now we are ready to state the uniqueness result in the case where depositors can withdraw at an arbitrary stopping time. 

\begin{prop}\label{prop:uniqueness_vartheta_inf}
 There exists at most one threshold equilibrium which satisfies the following regularity conditions: (i) $V^{\cZ}_H \in C^{1,1,2}(\mathcal{C})$ satisfies the equation \eqref{PDE:V_H} on $\mathcal{C}$ and the smooth-pasting condition \eqref{V_smooth_pasting}, and (ii) $\cZ(h,m)$ is locally Lipschitz in $(h,m)$. 
\end{prop}

The uniqueness result in Proposition \ref{prop:uniqueness_vartheta_inf} extends the proof of \cite{FrankelPauzner2000} to the case where the stopping time can be chosen arbitrarily, not necessarily a Poisson arrival time. When the stopping time can only be a Poisson arrival time, $V^{\cZ}_H$ is strictly increasing in $z$ (see Lemma \ref{lem:VZ_mono}) and any stopping time $\tau \in \{\tau_n\}_n$ must satisfy $\mathbb{P}(\tau >0) =1$, because the probability for a Poisson time to arrive immediately is zero. When depositors can choose an arbitrary stopping time, $V^{\cZ}_H \equiv 0$ on $\mathcal{S}$. Hence $V^{\cZ}_H$ is not strictly increasing in $z$ in the entire domain. Moreover, stopping immediately, i.e., $\mathbb{P}(\tau =0)=1$, is an admissible choice for depositors. We extend the argument of \cite{FrankelPauzner2000} by incorporating a strong form of maximum principle to prove Proposition \ref{prop:uniqueness_vartheta_inf}.

For the regularity assumptions, given that the dynamics of $h$ and $m$ are deterministic, equation \eqref{PDE:V_H} can be transformed into a parabolic partial differential equation (PDE). The assumption of $\sigma>0$ and the regularity theory for parabolic PDE allows us to weaken the regularity assumption of  $V^{\cZ}_H$ to the continuity of $V^{\cZ}_H$, i.e., $V_H^{\cZ}(\mathcal{C})$, see Remark \ref{rem:VH_regularity} in Section \ref{app:inf_vartheta}. The local Lipschitz regularity of the optimal threshold is more challenging to prove. For the American put option problem, the regularity of the optimal exercise boundary has been established in the Black-Scholes setting using PDE argument by \cite{chen2007mathematical}. Due to these technical challenges, we leave the existence of threshold equilibria satisfying the regularity conditions in Proposition \ref{prop:uniqueness_vartheta_inf} as a future research question.

\section{Proofs}\label{app:proof}

\subsection{Proof of Lemma \ref{lem:obj}}

We decompose the objective function in \eqref{def:dep_V} into four terms:
\begin{equation}\label{V_decomp}
\begin{split}
&\mathbb{E}_t\Big[\int_t^{\tau_\nu \wedge \tau_w \wedge t_\eta} e^{- \int_t^s \beta_u du} (r+\delta) \,ds + e^{-\int_t^{\tau_\nu } \beta_u du} 1_{\{\tau_\nu < \tau_w \wedge t_\eta\}} + e^{-\int_t^{\tau_w} \beta_u du} 1_{\{\tau_w < \tau_{\nu} \wedge t_\eta\}} \\
& + e^{-\int_t^{t_\eta} \beta_u du} \gamma  1_{\{t_\eta < \tau_\nu \wedge \tau_w \}}\Big] =: \text{I} + \text{II} + \text{III} + \text{IV}. 
\end{split}
\end{equation}
We compute each term separately. For I, due to the independence between $\nu$ and $t_\eta$, we have $\mathbb{P}(\tau_{\nu}\wedge t_{\eta} > s\,|\, \tau_{\nu} \wedge t_{\eta} > t) = e^{-\int_t^s (\nu + \eta m_u) du}$ when $s>t$. Therefore, 
\begin{align*}
\text{I} = &\mathbb{E}_t \Big[\int_t^{\tau_{\nu} \wedge \tau_w \wedge t_\eta} e^{-\int_t^s \beta_u du}(r+\delta)\, ds\Big] = \mathbb{E}_t \Big[\int_t^{\tau_w} e^{-\int_t^s \beta_u du} (r+\delta) 1_{\{s< \tau_{\nu} \wedge t_\eta\}} \,ds\Big] \\
= & \mathbb{E}_t \Big[\int_t^{\tau_w}e^{-\int_t^s (\nu + \beta_u + \eta m_u) du} (r+\delta) \, ds\Big].
\end{align*}
For II, note that $\mathbb{P}(\tau_\nu \in [s, s+ds) \,|\, \tau_\nu \wedge t_\eta >t)= e^{-\int_t^s \nu du} \nu ds$. Then
\[
\text{II} = \mathbb{E}_t \Big[\int_t^{\tau_w \wedge t_\eta} e^{-\int_t^s (\nu+\beta_u) du} \nu ds\Big] = \mathbb{E}_t \Big[\int_t^{\tau_w } e^{-\int_t^s (\nu+\beta_u + \eta m_u) du} \nu ds\Big],
\]
where the second equality follows from the same argument leading for I. For III, 
\[
\text{III} = \mathbb{E}_t \Big[e^{-\int_t^{\tau_w} (\nu + \beta_u + \eta m_u) du} \Big]  = 1 - \mathbb{E}_t \Big[\int_t^{\tau_w} e^{-\int_t^s (\nu + \beta_u + \eta m_u) du} (\nu + \beta_s + \eta m_s) ds\Big].
\]
Finally, note that $\mathbb{P}(t_\eta \in [s, s+ds)\,|\, \tau_\nu \wedge t_\eta >t) = e^{-\int_t^s \eta m_u du} \eta m_s ds$. Then
\[
\text{IV} = \gamma \mathbb{E}_t \Big[\int_t^{\tau_\nu \wedge \tau_w} e^{-\int_t^s (\beta_u + \eta m_u) du} \eta m_s \, ds\Big] = \gamma \mathbb{E}_t \Big[\int_t^{\tau_w} e^{-\int_t^s (\nu + \beta_u + \eta m_u) du} \eta m_s \, ds\Big]
\]
Combining the previous equations, we obtain the decomposition in \eqref{pro:os}.

\subsection{Proof of Lemma \ref{lemma monotonicity of brm}}

The subsequent proof of Lemma \ref{lemma monotonicity of brm} is based on an application of the Bank–El Karoui representation theorem, introduced in \cite{bank2004stochastic} and employed in \cite{BankFoellmer} for the study of optimal stopping problems and American options; see in particular \cite[Theorem 2]{BankFoellmer}. Importantly, the Bank-El Karoui approach provides a handy representation of the latest optimal stopping time (see \eqref{eq:maxminOStimes} below), which, by the general theory of optimal stopping, is given in terms of the first time of increase of the compensator in the Doob-Meyer decomposition of the Snell envelope -- an object that is typically technically challenging to characterize. In the proof of Lemma \ref{lemma monotonicity of brm}, this representation \eqref{eq:maxminOStimes} will play an important role.

We first briefly recall the Bank–El Karoui representation theorem result in relation to optimal stopping problems. Let $T\in (0,\infty]$ be a given time horizon and $(\Omega,\mathcal{F}, \mathbb{F}:=(\mathcal{F}_t)_{t\in [0,T]}, \mathbb{P})$ be a complete probability space. Consider the (not necessarily Markovian) optimal stopping problem
\begin{equation}
    \label{eq:OSgeneral}
    \sup_{\tau \in [0,T]}\mathbb{E}[X_{\tau}],
\end{equation}
for an optional process $X$ which is of class (D)\footnote{A stochastic process $X$ is said to be of class (D) if the family of random variables
$\{ X_\tau : \tau \text{ is a bounded stopping time} \}$
is uniformly integrable.} and upper semicontinuous in expectation, and where the optimization is performed over the set of $\mathbb{F}$-stopping times valued in $[0,T]$ $\mathbb{P}$-a.s. 

Suppose that the process $X$ admits the probabilistic representation
\begin{equation}
    \label{eq:BEK-thm}
    X_{\tau} = \mathbb{E}\bigg[\int_{(\tau,T]} \sup_{v \in [\tau,t)}\xi_v \mu(dt)\bigg],
\end{equation}
for some nonnegative, optional random measure $\mu$ on $([0,T], \mathcal{B}([0,T]))$ and
some progressively measurable process $\xi$ with upper-right continuous paths such that
$$\sup_{v \in [\tau(\omega),t)}\xi_v(\omega) \mathbbm{1}_{[\tau(\omega),T]}(t) \in L^1\big(\mathbb{P}(d\omega)\times \mu(\omega,dt)\big).$$

Then, the level passage times
\begin{equation}
\label{eq:maxminOStimes}
    \underline{\tau}:=\inf\{t\geq0: \xi_t \geq 0\} \wedge T, \quad \overline{\tau}:=\inf\{t\geq0: \xi_t > 0\} \wedge T,
\end{equation}
(with the usual convention $\inf \emptyset = + \infty)$ maximize the reward $\mathbb{E}[X_{\tau}]$. Furthermore, $\underline{\tau}$ is the minimal optimal stopping time, while $\overline{\tau}$ is the maximal optimal stopping time.

This approach is based on a representation of the underlying optional process $X$ in terms of the running supremum of another process $\xi$, which is to be determined. The works \cite{bank2001optimal,BankFoellmer,ferrari2015integral,ferrari2016irreversible} provide instances in which the process $\xi$ can be determined explicitly. This is typically the case in a one-dimensional, stationary, time-homogeneous setting. The process $\xi$ takes over the role of the Snell envelope and allows one to characterize optimal stopping times by means of a level-crossing principle. In \cite{ferrari2015integral} and \cite{ferrari2016irreversible}, it is shown how, in Markovian settings, the process $\xi$ is closely related to the free boundary arising from a PDE analysis of the problem and triggering the optimal stopping rule.

We now use the Bank-El Karoui representation to prove Lemma \ref{lemma monotonicity of brm}.

\begin{proof}
 Without loss of generality, we take $t=0$ in \eqref{pro:os} and \eqref{def:hat_J}. Then, by using the strong Markov property, it follows that
\begin{eqnarray}
\label{eq:newrepr}
&& \sup_{\tau_w \geq 0}\mathbb{E} \Big[\int_0^{\tau_w} e^{-\int_0^s (\nu + \beta_u + \eta m_u) du} \Big(r+\delta - \beta_s -  (1-\gamma) \eta m_s \Big) ds\Big] \nonumber\\
&& = G(0, y, x) + \sup_{\tau_w \geq 0}\mathbb{E}\Big[ - e^{-\int_0^{\tau_w} (\nu + \beta(X^{c,0,y}_u, X^{0,x}_u) + \eta m_u) du} G(\tau_w, X^{c, 0, y}_{\tau_w}, X^{0,x}_{\tau_w})\Big],
\end{eqnarray}
where $X^{c,0,y}$ and $X^{0,x}$ are the common and private state processes which start from $y$ and $x$ at time $0$, respectively. In the previous equation, for $t\geq0$ and $x, y\in \mathbb{R}$, we have set
$$G(t,y,x):=\mathbb{E} \Big[\int_t^{\infty} e^{-\int_t^s (\nu + \beta(X^{c,t,y}_u, X^{t,x}_u) + \eta m_u) du} \Big(r+ \delta - \beta(X^{c,t,y}_s, X^{t,x}_s) -  (1-\gamma) \eta m_s \Big) ds\Big].$$

Define now the progressively-measurable process
\begin{equation}
Y_t:= - e^{-\int_0^{t} (\nu + \beta(X^{c,0,y}_t, X^{0,x}_t) + \eta m_u) du} G(t, X^{c,0,y}_t, X^{0,x}_{t}),
\end{equation}
with the convention
$$Y_{\tau}:=\limsup_{t \uparrow \infty}\Big(- e^{-\int_0^{t} (\nu + \beta(X^{c,0,y}_t, X^{0,x}_t) + \eta m_u) du} G(t, X^{c,0,y}_t, X^{0,x}_{t})\Big) = 0 \quad \text{on} \quad \{\tau=+\infty\}.$$
The fact that the limit superior above is zero is due to the boundedness of $\beta$, which in turn yields boundedness of $G$.
This latter property of $G$ then gives boundedness of $Y$ (uniformly in $(\omega,t)\in \Omega \times \mathbb{R}_+$) and hence the fact that $Y$ is of class (D) and it is lower-semicontinuous (actually continuous) in expectation.

Defining the nonnegative Borel measure on $[0,\infty)$
$$\mu(dt):= (\nu + \beta_t + \eta m_t) e^{-\int_0^{t} (\nu + \beta_u + \eta m_u) du} dt,$$ 
it then follows by Theorem 3 in \cite{bank2004stochastic} (see also Theorem 2.2 in \cite{BankFoellmer}) that, for any given stopping time $\tau$, the representation problem
\begin{equation}\label{def:-Y}
-Y_{\tau} = \mathbb{E}\bigg[\int_{\tau}^{\infty} \Big(-\sup_{\tau \leq s \leq t}\xi_s\Big) \mu(dt)\,\Big|\,\mathcal{F}_{\tau}\bigg]
\end{equation}
admits a unique progressively measurable upper right-continuous solution $\xi$. Furthermore, due to Theorem 1 in \cite{bank2004stochastic}, $\xi$ is such that for any two stopping times $S, \zeta$ with $S\leq \zeta$
\begin{equation}
\label{eq:repr-xi}
\xi_S= \text{essinf}_{\zeta \geq S} \ell_{S,\zeta}, 
\end{equation}
being the random variable $\ell_{S,\zeta}$ is uniquely given (up to optional sections) by
\begin{equation}
    \label{eq:rvell}
    \ell_{S,\zeta}=-\frac{\mathbb{E}\bigg[\int_S^{\zeta} e^{-\int_S^u (\nu + \beta_{\alpha} + \eta m_{\alpha}) d\alpha} \Big(r+\delta  - \beta_u -  (1-\gamma) \eta m_{u} \Big) du\, \Big|\, \mathcal{F}_S\bigg]}{\mathbb{E}\Big[1 - e^{-\int_S^{\zeta} (\nu + \beta_{\alpha} + \eta m_{\alpha}) d\alpha} \, \Big|\, \mathcal{F}_S\Big]}.
\end{equation}

Observing that by an integration by parts
\begin{align*}
 \int_S^{\zeta} e^{-\int_S^u (\nu + \beta_{\alpha} + \eta m_{\alpha}) d\alpha} \eta m_u du = 1 - e^{-\int_S^{\zeta} (\nu + \beta_{\alpha} + \eta m_{\alpha}) d\alpha} - \int_S^{\zeta} e^{-\int_S^u (\nu + \beta_{\alpha} + \eta m_{\alpha}) d\alpha} (\nu + \beta_u) du,
\end{align*}
and plugging the previous equation into \eqref{eq:rvell}, we transform \eqref{eq:rvell} to the following equivalent form
\begin{equation}
\label{eq:rvell-equiv}
\ell_{S,\zeta}= (1-\gamma) -\frac{\mathbb{E}\bigg[\int_S^{\zeta} e^{-\int_0^u (\nu + \beta_u + \eta m_{\alpha}) d\alpha} \Big(r+\delta + (1-\gamma) \nu - \gamma \beta_u\Big) du \, \Big|\, \mathcal{F}_S \bigg]}{\mathbb{E}\Big[1 - e^{-\int_S^{\zeta} (\nu + \beta_u + \eta m_{\alpha}) d\alpha} \, \Big|\, \mathcal{F}_S \Big]}.
\end{equation}

We now aim at studying the monotonicity of $\ell_{S,\tau}$ with respect to the given and fixed nondecreasing function $m$ appearing on the right-hand side of \eqref{eq:rvell-equiv}. Hence, for an arbitrary function $m$, we denote the random variable given through \eqref{eq:rvell-equiv} by $\ell^{(m)}_{S,\tau}$. For $\overline{m}$ and $m$ nondecreasing right-continuous functions such that $\overline{m}_s \leq m_s$ for almost every $s\geq0$, by employing that $r+\delta + (1-\gamma) \nu -\gamma \beta_{\text{max}} \geq 0$ by assumption, \eqref{eq:rvell-equiv} gives that $\ell^{(m)}_{S,\tau} \geq \ell^{(\overline{m})}_{S,\tau}$. Because of \eqref{eq:repr-xi}, now evaluated at the deterministic time $S=s$, we thus have $\xi^{(m)}_s \geq \xi^{(\overline{m})}_s$ for any $s\geq0$ $\mathbb{P}$-almost surely. 

We thus conclude by Theorem 1.3 in \cite{BankFoellmer} that for the earliest optimal stopping time $\underline{\tau}_w$ and the latest optimal stopping time $\overline{\tau}_w$, for which we now stress the dependence with respect to $m$, we have
\begin{equation}
    \label{eq:earliest}
 \underline{\tau}_w(m)=\inf\{s\geq0: \xi^{(m)}_s \geq 0\} \leq \inf\{s\geq0: \xi^{(\overline{m})}_s \geq 0\} =  \underline{\tau}_w(\overline{m}),
\end{equation}
and
\begin{equation}
    \label{eq:latest}
 \overline{\tau}_w(m)=\inf\{s\geq0: \xi^{(m)}_s > 0\} \leq \inf\{s\geq0: \xi^{(\overline{m})}_s > 0\} =  \overline{\tau}_w(\overline{m}).
\end{equation}
The proof is thus complete.
\end{proof}

\subsection{Proof of Proposition \ref{theorem existence minimal maximal MFGE}}

   The set of (possibly infinite) stopping times $\mathcal T$ is a complete lattice. For a given stopping time $\tau$, define $m_t(\tau) := \mathbb{P}(\tau_{\nu} \wedge \tau \leq t\,|\, \mathcal{F}^Z_t)$ for any $t\geq 0$, where $\tau_{\nu}$ is an independent exponential random variable with the parameter $\nu$. Introduce the best response map: 
    \begin{equation}
        \label{def:resp}
        \mathcal{R}(\tau) := \tau_w\big(m(\tau)\big),
    \end{equation}
    i.e., the set-valued optimal withdrawal time for the representative depositor. 
    Note that $\tau \mapsto m(\tau)$ is nonincreasing map.  
    It follows from Lemma \ref{lemma monotonicity of brm} that $\mathcal{R}$ is a monotone nondecreasing map from $\mathcal{T}$ to itself. By Tarski fixed point theorem  (see \cite{tarski1955lattice}), the set of fixed points of $\mathcal{R}$ is a non-empty complete lattice, so that there exist minimal and maximal fixed points of $\mathcal{R}$, respectively denoted by $\underline \tau_w$ and $\overline \tau_w$. Moreover, any other fixed point $\hat{\tau}_w$ of $\mathcal{R}$ must satisfy $\underline \tau_w \leq \hat \tau_w \leq  \overline \tau_w$. Because $\tau \mapsto m(\tau)$ is nonincreasing. Therefore, $\overline m _{t} \leq \hat m _t \leq \underline m_{t}$ for any $t\geq 0$. Finally, the statement of the proposition follows from the one-to-one correspondence between fixed points of the map $\mathcal{R}$ and equilibria introduced in Definition \ref{def:mfe}. Indeed, $(\hat \tau, \hat m)$ is a mean-field equilibrium if and only if $\hat \tau \in R (\hat \tau )$ and $\hat m_t = \mathbb P (\tau_{\nu} \wedge \hat \tau \leq t \,|\, \mathcal{F}^Z_t)$. 

\subsection{Proof of Proposition \ref{prop comparative statics strength parameter}}
 
To stress the dependence on the parameters $\eta, \widehat{\eta}$ of the profit functional $J$ defined in \eqref{def:hat_J}, we will write $J^\eta$ and $J^{\widehat{\eta}}$. Recall the monotone schemes introduced in Appendix \ref{app:m_schemes}.
Given a generic initialization $m^{(0)}$, denote by $(\tau_w^{(n)},m^{(n)})$ the schemes with inizialization $m^{(0)}$ and parameter $\eta$, and by $(\widehat{\tau}_w^{(n)},\widehat{m}^{(n)})$ the schemes with inizialization $m^{(0)}$ and parameter $\widehat{\eta}$.

Since $\eta \, m^{(0)}_t \leq \widehat{\eta} \, m^{(0)}_t$ for all $t$, the monotonicity of the best response map in Lemma \ref{lemma monotonicity of brm} implies that 
$$
\tau_w^{(1)} = \argmax J^\eta (t, Z_t, X_t; \cdot, m^{(0)})  \geq \argmax J^{\widehat{\eta}} (t, Z_t, X_t; \cdot, m^{(0)}) = \widehat{\tau}_w^{(1)},
$$
as well as
$$
m^{(1)} _t = \mathbb P \big(\tau_{\nu}\wedge\tau^{(1)}_w \leq t\,|\, \mathcal{F}^Z_t\big) 
\leq  \mathbb P \big(\tau_{\nu} \wedge \widehat{ \tau}^{(1)}_w \leq t\,|\, \mathcal{F}^Z_t\big) = \widehat{m}^{(1)}_t.
$$
For the second step, the latter inequality in turn implies that $\eta \, m^{(1)}_t \leq \widehat{\eta} \, \widehat{m}^{(1)}_t$ for all $t$.
Thus, again from the monotonicity of the best response map  we obtain $\tau_w^{(2)} \geq \widehat{\tau}_w^{(2)}$. 
Proceeding by induction, we obtain
$$
\tau_w^{(n)} \geq \widehat{\tau}_w^{(n)}
\quad \text{and}\quad
m^{(n)} \leq \widehat{m}^{(n)}.
$$
If now $m^{(0)}_t =1$ for any $t$, we can use the first statement in Proposition \ref{prop:convergence} to take limits in the previous inequality and deduce that
$$
\underline{\tau}_w^{\eta} \geq \underline{\tau}_w ^{\widehat{\eta}}
\quad \text{and}\quad
\underline m ^{\eta} \leq \underline m^{\widehat{\eta}}.
$$
Similarly, when  $m^{(0)}_t =0$ for any $t$, we obtain
$$
\overline{\tau}_w^{\eta} \geq \overline{\tau}_w ^{\widehat{\eta}}
\quad \text{and}\quad
\overline m ^{\eta} \leq \overline m^{\widehat{\eta}}.
$$

\subsection{Proof of Proposition \ref{prop:Poisson}}

For a given withdrawn share $m$, define the value function of a depositor as   
\[
V(t, X_t; m) -1  = \sup_{\tau_w} J(t, X_t; m, \tau_w),
\]
where the objective function $J$ is defined in \eqref{def:hat_J} without the common state $Z$. We denote $V_L(t;m):= V(t, \beta_L; m)-1 $ and $V_H(t; m) = V(t, \beta_H; m)-1$. Let $\tau_\theta$ be the arrival time for the discount rate shock. It follows from the dynamic programming principle and the exponential distribution of $\tau_\theta$ that
\begin{align}
V_L(t; m) = & \sup_{\tau_w} \mathbb{E}_t \Big[\int_t^{\tau_w \wedge \tau_\theta} e^{-\int_t^s (\nu + \beta_L  + \eta m_u) du} \big(r+\delta - \beta_L - (1-\gamma) \eta m_s\big) ds \nonumber\\
& \qquad + e^{-\int_t^{\tau_\theta}(\nu + \beta_L + \eta m_u) du} V_H (\tau_\theta; m) 1_{\{\tau_w \geq \tau_\theta\}}\Big] \nonumber\\
= & \sup_{\tau_w} \int_t^{\tau_w} e^{-\int_t^s (\nu + \beta_L  + \theta_{LH} + \eta m_u ) du} \big(r+\delta -\beta_L - (1-\gamma) \eta m_s\big) ds \nonumber\\
& \qquad + \int_t^{\tau_w} e^{-\int_t^s (\nu + \beta_L + \theta_{LH} + \eta m_u) du} \theta_{LH} e^{-\int_t^s \theta_{LH} du} V_H(s; m) ds \nonumber\\
= & \sup_{\tau_w} \int_t^{\tau_w} e^{-\int_t^s (\nu + \beta_L + \theta_{LH} + \eta m_u) du} \big(r+\delta - \beta_L - (1-\gamma) \eta m_s + \theta_{LH} V_H (s; m)\big) ds. \label{Poisson:VH}
\end{align}
For the impatienting depositor, because the discounting rate never returns to $\beta_L$, the value function is 
\begin{equation}\label{Poisson:VL}
V_H(t;m) = \sup_{\tau_w} \int_t^{\tau_w} e^{-\int_t^s (\nu + \beta_H  + \eta m_u) du} \big(r + \delta - \beta_H - (1-\gamma) \eta m_s\big) ds.
\end{equation}

For the impatienting depositor, given that $t\mapsto m_t$ is nondecreasing and the tie-breaking condition, the optimal withdrawal time is 
\begin{equation}\label{tau_Lm}
\tau_w^*(H; m) = \inf\{s\geq t\,:\, r + \delta - \beta_H - (1-\gamma) \eta m_{s} \leq 0\}.
\end{equation}
Hence, the impatienting value function is 
\[
 V_H(t; m) = \int_t^{\tau_w^*(L;m)} e^{-\int_t^s (\nu + \beta_H + \eta m_u) du} \big(r+\delta-\beta_H - (1-\gamma) \eta m_s\big) ds, \quad \text{ when } t< \tau^*_w(L;m),
\]
and $V_H(t; m) = 0$ when $t\geq \tau^*_w(L;m)$. For the patienting depositor, given the assumption $r+ \delta - \beta_L - (1-\gamma) \eta>0$ and $V_H \geq 0$, therefore $V_L (t; m) \geq 0$ and 
\begin{equation}
 \tau^*_w(L; m) = \infty.
\end{equation}
In summary, for any given $m$, a patienting depositor never withdraws early, and a impatienting depositor optimally withdraw at $\tau^*_w(H;m)$. Putting the two types together, the optimal withdrawal time for a representative depositor is $\tau^*_w = \tau_\theta \vee \tau_w^*(H; m)$, i.e., the later time between $\tau_w^*(H;m)$ and the arrivial time of the discounting rate shock. 

By the consistency condition \eqref{def:m_prob} for the case without common state, 
\begin{align}
 m_s = &\mathbb{P}(\tau_{\nu} \wedge \tau^*_w \leq s) = \mathbb{P}(\tau_\nu \leq s) + \mathbb{P} (\tau_\nu >s, \tau^*_w \leq s) \nonumber\\
 =& 1-e^{-\nu s} + e^{-\nu s} \mathbb{P}(\tau_\theta \vee \tau^*_w(H;m) \leq s)\nonumber\\
 =& 1- e^{-\nu s} + e^{-\nu s} \mathbb{P} (\tau_\theta \leq s, \tau^*_w(H;m) \leq s)\nonumber\\
 =&  1- e^{-\nu s} + e^{-\nu s} (1-e^{-\theta_{LH} s}) 1_{\{\tau^*_w(H;m) \leq s\}}, \quad \text{for any } s\geq 0. \label{Poisson:m_eqn}
\end{align}
Notice that $m$ appears on both sides of \eqref{Poisson:m_eqn}. Hence it is an equation for $m$ and each solution corresponds to an equilibrium. 

To identify all solutions of \eqref{Poisson:m_eqn}, observe that the right-hand side of \eqref{Poisson:m_eqn} is sandwiched between $1- e^{- \nu t}$ and $1- e^{-(\nu + \theta_{LH}) t}$. Denote
\begin{align*}
 \underline{t} := & \inf\{s\geq 0\,:\, r+\delta -\beta_H - (1-\gamma) \eta (1- e^{- (\nu + \theta_{LH}) s}) \leq 0 \} \quad \text{and} \\
 \overline{t} := & \inf\{s\geq 0\,:\, r+\delta - \beta_H - (1-\gamma) \eta (1- e^{- \nu s}) \leq 0 \}.
\end{align*}
They are the lower and upper bound of $\tau^*_w(L;m)$, respectively. These lower and upper bound can be obtained by the following choice of $m$. For $\underline{t}$, consider 
\[
 \underline{m}_s := \left\{\begin{array}{ll}1- e^{-\nu s}, & s< \underline{t}, \\ 1-e^{- (\nu + \theta_{LH}) s}, & s\geq \underline{t}. \end{array}\right.
\]
It follows from \eqref{tau_Lm} that $\tau^*(L;\underline{m}) = \underline{t}$. Meanwhile, for $\overline{t}$, consider
\[
\overline{m}_s := \left\{\begin{array}{ll}1- e^{-\nu s}, & s< \overline{t}, \\ 1-e^{- (\nu+\theta_{LH}) s}, & s\geq \overline{t}. \end{array}\right.
\] 
It then follows that $\tau^*(L; \overline{m}) = \overline{t}$. These two pairs $(\underline{t}, \underline{m})$ and $(\overline{t}, \overline{m})$ correspond to two extreme equilibria: $(\underline{t}, \underline{m})$ is the earliest-run equilbrium and $(\overline{t}, \overline{m})$ is the latest-run equilibrium. 

To identify all solutions of \eqref{Poisson:m_eqn}, we observe from \eqref{Poisson:m_eqn} that any solution for \eqref{Poisson:m_eqn} admits the form 
\[
\widetilde{m}_s = \left\{\begin{array}{ll} 1- e^{-\nu s}, & s< \tilde{t},\\ 1- e^{-(\nu + \theta_{LH}) s}, & s\geq \tilde{t}, \end{array}\right.
\]
for any $\tilde{t} \in [\underline{t}, \overline{t}]$. Because $\tilde{t} \geq \underline{t}$, we have $(1-\gamma) \eta (1- e^{-(\nu + \theta_{LH}) \tilde{t}}) \geq (1-\gamma) \eta (1- e^{-(\nu + \theta_{LH}) \underline{t}}) = r+\delta -\beta_H$. Meanwhile, for any $s<\tilde{t}<\overline{t}$, $(1-\gamma) \eta (1- e^{-\nu s})<(1-\gamma) \eta (1- e^{-\nu \tilde{t}}) < (1-\gamma) \eta (1- e^{-\nu \overline{t}}) = r +\delta - \beta_H$. It then follows from \eqref{tau_Lm} that  $\tau^*_w(H, \tilde{m}) = \tilde{t}$, and hence $(\tilde{t}, \tilde{m})$ is an equilibrium.

\subsection{Proof of Lemma \ref{lem:dyn_m_fin_del}}

Let $\tau_{\theta}$ be the arrival time of discount rate shock, it is an exponential random variable with the parameter $\theta_{LH}$. Let $\{\tau_n\}_{n\geq 1}$ be the jump times of a Poisson process with the intensity $\vartheta$. A representative depositor can only withdraw at $\{\tau_n\}_{n\geq 1}$. Consider the threshold strategy of an impatient depositor in \eqref{threshold_str}. Then
\begin{align*}
 &\mathbb{P}\big(\tau_w(H; h, m, \cZ) >t \,|\, \mathcal{F}_t^Z \big) = \mathbb{P}\big( \tau_{\theta} > t \,|\, \mathcal{F}^Z_t\big) + \mathbb{P}\Big(\tau_\theta < t, \tau_n \mathbbm{1}_{\{Z_{\tau_n} \leq \cZ(h_{\tau_n}, m_{\tau_n})\}} \notin (\tau_\theta, t) \,|\, \mathcal{F}^Z_t\Big)\\
 & = e^{-\theta_{LH}t} + \int_0^t \theta_{LH} e^{-\theta_{LH} s} e^{-\vartheta \int_s^t \mathbbm{1}_{\{Z_u \leq \mathcal{Z}(h_u, m_u)\}} du }ds.
\end{align*}
Therefore, 
\begin{align*}
m^{\mathcal{Z}}_t = &1 - \mathbb{P} \Big( \tau_{\nu} \wedge \tau_w(H; h, m, \cZ) >t\,|\, \mathcal{F}^Z_t\Big)\\
= & 1- \mathbb{P}\big(\tau_{\nu}> t\,|\, \mathcal{F}^Z_t \big) \mathbb{P} \big( \tau_w(H; h, m, \cZ) >t\,|\, \mathcal{F}^Z_t\big)\\
= & 1- e^{-\nu t} \Big[e^{-\theta_{LH}t} + \int_0^t \theta_{LH} e^{-\theta_{LH} s} e^{-\vartheta \int_s^t \mathbbm{1}_{\{Z_u \leq \mathcal{Z}(h_u, m_u)\}} du }ds\Big].
\end{align*}
Taking derivative with respect to $t$, we obtain 
\begin{align*}
dm^{\mathcal{Z}}_t =& \nu (1- m^{\mathcal{Z}}_t) dt - e^{-\nu t} \Big[- \theta_{LH} e^{-\theta_{LH} t} + \theta_{LH} e^{-\theta_{LH} t} \\
&- \vartheta \mathbbm{1}_{\{Z_t \leq \mathcal{Z}(h_t, m_t)\}} \int_0^t \theta_{LH} e^{-\theta_{LH} s} e^{-\vartheta \int_s^t \mathbbm{1}_{\{Z_u \leq \mathcal{Z}(h_u, m_u)\}} du }ds\Big] dt\\
=& \nu (1- m^{\mathcal{Z}}_t) dt + \delta \mathbbm{1}_{\{Z_t \leq \mathcal{Z}(h_t, m_t)\}} \big(1- e^{- (\nu + \theta_{LH})t} - m^{\mathcal{Z}}_t \big) dt.
\end{align*}
Recall that $h_t = 1- e^{- (\nu + \theta_{LH})t}$. We confirm the dynamics \eqref{dyn:m_fin_del}. 

\subsection{Proof of Lemma \ref{lem:VZ_mono}}
First, due to \eqref{ass:Z},  the process $Z^z$ is strictly increasing in $z$. 
Thus, for generic variables $h$ and $m$, we have that $\mathds{1} _{ \{ Z^z_s \leq \cZ (h,m) \} } $ is decreasing in $z$. Therefore,  the comparison principle for ordinary differential equations for the state dynamics  \eqref{dyn:m_fin_del} implies that $m^{z,\cZ}$ is decreasing in $z$. 
In other words, for initial conditions $z < \bar z$ at time $t$, denoting $Z = Z^z, \bar Z = Z^{\bar z}$,  $m^{\cZ} = m^{z, \cZ}$, and $\bar m^{\cZ} = m^{\bar z, \cZ}$ to simplify notation, we have
\begin{equation}
    \label{eq monotonicity Z m}
    Z_s < \bar Z_s, \quad m^{\cZ}_s \geq \bar m^{\cZ}_s, \quad \text{for any } s\geq t. 
\end{equation}

Next, use integration-by-part to rewrite the objective function as 
\begin{equation}\label{eq representation V with F}
V_H^{\cZ}(h,m,z) 
 =  \sup_{\tau \in\{\tau_n\}_n}  \mathbb{E}_{h, m, z} [F(t,\tau; Z,m^{\cZ} )], 
 \end{equation}
 where $F$ is defined as
 \begin{equation} \label{def:F}
 \begin{aligned}
 F(t,\tau; Z,m^{\cZ} ) =&\int_t^{\tau} e^{-\int_t^s (\nu + \beta_H (Z_u) + \eta m^{\cZ}_u)du} \big( r+\delta + (1-\gamma) \nu - \gamma\beta_H (Z_s)\big) ds \\
& + (1-\gamma) \Big( e^{-\int_{t}^{\tau}(\nu + \beta_H(Z_u) + \eta m^{\cZ}_u) du} -1 \Big).
\end{aligned}
\end{equation}
The fact that $\beta_H$ is decreasing together with \eqref{eq monotonicity Z m} 
implies that
$$
 -\beta_H (Z_s) < - \beta_H ( \bar Z_s), \quad - m^{\cZ}_s \leq -\bar m^{\cZ}_s. 
$$
Moreover, since by assumption we have $1-\gamma \geq 0$ and $r+\delta + (1-\gamma) \nu -  \gamma\beta_H (\tilde z) >0$ for any $\tilde{z}$ (see \eqref{ass:Rmin}), we obtain
$$
F(t,\tau; \bar Z,\bar m^{\cZ} ) - F(t,\tau; Z,m^{\cZ} ) > 0,
$$
for any stopping time $\tau > t$.

Finally, choosing $\tau$ to be optimal for the initial condition $z$, we find
\begin{equation*}
\begin{aligned}
V_H^{\cZ}(h,m,\bar z) - V_H^{\cZ}(h,m, z) 
 \geq   \mathbb{E}_{h, m} \Big[F(t,\tau; \bar Z,\bar m^{\cZ} ) - F(t,\tau; Z,m^{\cZ} )\Big] >0,
\end{aligned}
\end{equation*}
where the strict inequality follows from the fact that the Poisson arrival time $\tau$ is almost surely greater than $t$. 
This proves the monotonicity of $V_H^{\cZ}$ in $z$. 

Regarding the monotonicity in $(h,m)$, we first notice that $h_t$ is increasing in $h$. 
Hence,  by the monotonicity of $\cZ$ in $h_t$
and the comparison principle for ordinary differential equations for the state dynamics  \eqref{dyn:m_fin_del}, we deduce that that $m^{\cZ}$ is nondecreasing in $(h,m)$.
Therefore, since $F$ is nonincreasing in the process $m^{\cZ}$, for any stopping time $\tau$ we deduce that 
 $\mathbb{E}_{h, m, z} [F(t,\tau; Z,m^{\cZ} )]$ is nonincreasing in $(h,m)$. 
 We conclude that $V_H^{\cZ}$ is nonincreasing in $(h,m)$, thus completing the proof. 

\subsection{Proof of Proposition \ref{prop:threshold_finite}}

\noindent\underline{Existence}. 
Let $\mathcal{Z}:[0,1]^2 \to \mathbb{R}$ be a given and fixed measurable threshold function.  
Recall from \eqref{eq representation V with F} that
\begin{equation*}
V_H^{\cZ}(h,m,z) 
 =  \sup_{\tau \in\{\tau_n\}_n}  \mathbb{E}_{h, m, z} [F(t,\tau; Z,m^{\cZ} )].
 \end{equation*}
%  where the functional $F$ (taking as entries the optimization's starting time $t \geq0$, the Poissonian withdrawal time $\tau$, the Brownian trajectory $Z$, and the withdrawn share, following the dynamics \eqref{dyn:m_fin_del}) is defined as
%  \begin{equation*}
%  \begin{aligned}
%  F(t,\tau; Z,m^{\cZ} ) :=&\int_t^{\tau} e^{-\int_t^s (\nu + \beta_H (Z_u) + \eta m^{\cZ}_u)du} \big( r+\delta + (1-\gamma) \nu - \gamma\beta_H (Z_s)\big) ds \\
% & + (1-\gamma) \Big( e^{-\int_{t}^{\tau}(\nu + \beta_H(Z_u) + \eta m^{\cZ}_u) du} -1 \Big).
% \end{aligned}
% \end{equation*}
Then, because by assumption $1-\gamma>0$ and $r+\delta + (1-\gamma) \nu - \gamma\beta_H (z)>0$, for all $z\in \mathbb{R}$ (see \eqref{ass:Rmin}), the decreasing map $m^{\cZ} \mapsto F(t,\tau; Z, m^{\cZ})$ implies that 
\[
V_H^{\cZ}(h,m,z) \geq \sup_{\tau \in\{\tau_n\}_n}  \mathbb{E}_{h, m, z} [F(t,\tau; Z, 1)] \quad  \text{and} \quad V_H^{\cZ}(h,m,z) \leq \sup_{\tau \in\{\tau_n\}_n}  \mathbb{E}_{h, m, z} [F(t,\tau; Z, 0)].
\]

Hence, the previous inequalities on $V_H^{\cZ}$ imply that
\begin{equation}
    \label{eq:uniformbounds}
    \underline{\mathcal{Z}}(h,m) \leq \inf\big\{z \,:\, V_H^{\cZ}(h,m,z) =0 \big\} \leq \overline{\mathcal{Z}}(h,m), \quad (h,m) \in [0,1]^2,
\end{equation}
with $\underline{\mathcal{Z}}:[0,1]^2 \to \mathbb{R}$ and $\overline{\mathcal{Z}}:[0,1]^2 \to \mathbb{R}$ such that
\begin{equation}
    \label{def:indiff_Z_underline}
    \underline{\mathcal{Z}} := \inf\big\{z \,:\, \sup_{\tau \in\{\tau_n\}_n}  \mathbb{E}_{h, m, z} [F(t,\tau; Z, 0)] =0 \big\}, 
\end{equation}
and
\begin{equation}
    \label{def:indiff_Z_overline}
    \overline{\mathcal{Z}} := \inf\big\{z \,:\, \sup_{\tau \in\{\tau_n\}_n}  \mathbb{E}_{h, m, z} [F(t,\tau; Z, 1)] =0 \big\}. 
\end{equation}
It is important to notice that $\underline{\mathcal{Z}}$ and $\overline{\mathcal{Z}}$ are actually independent of the given and fixed $\mathcal{Z}$, since the values $\sup_{\tau \in\{\tau_n\}_n}  \mathbb{E}_{h, m, z} [F(t,\tau; Z, 0)]$ and $\sup_{\tau \in\{\tau_n\}_n}  \mathbb{E}_{h, m, z} [F(t,\tau; Z, 1)]$ are such.

Let $[\underline{z}, \overline{z}]\subseteq \mathbb{R}$ be the state space of $Z$. Here $\underline{z}$ can be $-\infty$ and $\overline{z}$ can be $\infty$. Condition \eqref{ass:beta_L_common} implies that $\underline{\cZ}$ and $\overline{\cZ}$ are both in the interior of $[\underline{z}, \overline{z}]$. For $\underline{\cZ}$, it is the optimal stopping boundary for the problem 
\[
\underline{V}(z) = \sup_{\tau_n} \mathbb{E}_z \Big[\int_t^{\tau} e^{-\int_t^s \nu + \beta_H (Z_u) du} \big( r+\delta - \beta_H(Z_s)\big) ds\Big].
\]
The second condition in \eqref{ass:beta_L_common} implies that $\underline{V}(\underline{z}) <0$, hence $\underline{Z} > \underline{z}$. For $\overline{\cZ}$, it is the optimal stopping boundary for the problem 
\[
\overline{V}(z) = \sup_{\tau_n} \mathbb{E}_z \Big[\int_t^{\tau} e^{-\int_t^s \nu + \beta_H (Z_u) + \eta du} \big( r+\delta - \beta_H(Z_s) - (1-\gamma) \eta\big) ds\Big].
\]
The first condition in \eqref{ass:beta_L_common} implies $\overline{V}(\overline{z}) >0$, hence $\overline{Z} < \overline{z}$.

Set
$$\mathcal{B}:=\{\mathcal{Z}:[0,1]^2 \to \mathbb{R}:\,\,h \mapsto \mathcal{Z}(h,m)\,\,\text{and}\,\,m \mapsto \mathcal{Z}(h,m)\,\,\text{are nondecreasing and s.t.}\,\, \underline{\mathcal{Z}} \leq \mathcal{Z} \leq \overline{\mathcal{Z}}\,\,\text{on}\,\,[0,1]^2\}$$
and, as a consequence of \eqref{eq:uniformbounds} and Lemma \ref{lem:VZ_mono}, we can define the map
$\mathcal{T}: \mathcal{B} \to \mathcal{B}$ such that (cf.\ also \eqref{def:indiff_Z})
\begin{equation}
    \label{eq:mathcalT}
    [\mathcal{T}\mathcal{Z}](h,m):= \inf\big\{z \,:\, V_H^{\cZ}(h,m,z) =0 \big\}.
\end{equation}

Notice now that $\mathcal{B}$ is a complete lattice, when endowed with the pointwise order $\mathcal{Z}_1 \leq \mathcal{Z}_2$ if and only if $\mathcal{Z}_1(h,m) \leq \mathcal{Z}_2(h,m)$ for all $(h,m)\in [0,1]^2$. Furthermore, since $\mathcal{Z} \to m^{\mathcal{Z}}$ is nondecreasing (due to a comparison principle applied to the ODE \eqref{dyn:m_fin_del}), $m^{\mathcal{Z}} \mapsto F(t,\tau;Z,m^{\mathcal{Z}})$ is noincreasing, it follows that ${\mathcal{Z}} \mapsto V_H^{\cZ}(h,m,z)$ is noincreasing, and, as a direct consequence of this, that $\mathcal{Z} \mapsto \mathcal{T}\mathcal{Z}$ is nondecreasing. Tarski's fixed point theorem then gives that the set of fixed points of $\mathcal{T}$ is a nonempty complete lattice. Each fixed point of $\mathcal{T}$ is a equilibrium threshold. As a member of $\mathcal{B}$, this threshold is nondecreasing in $h$ and $m$.

\medskip

\noindent \underline{Uniqueness}. 
Suppose that there are two threshold equilibria $\cZ_1$ and $\cZ_2$. Without loss of generality, we assume that there exists $(h,m)$ such that $\cZ_1(h,m) > \cZ_2(h,m)$. 
We shift $\cZ_1$  sufficiently down, so that a parallel version of $\cZ_1$, $\widetilde{\cZ}_1$, intersets $\cZ_2$ at a point $(h^*,m^*)$, i.e., $\widetilde{\cZ}_1(h^*,m^*) = \cZ_2(h^*,m^*)$, but the rest of $\widetilde{Z}_1$ is below  $\cZ_2$, i.e., $\widetilde{\cZ}_1(h,m) \leq \cZ_2(h,m)$ for any $h,m$. 
Because $\cZ_2$ is a threshold equilibrium, we have 
\begin{equation}
    \label{cZ2_id}
    V_H^{\cZ_2}( h^*,m^*, \cZ_2(h^*,m^*)) =0.
\end{equation}

Notice that the mapping $\mathcal{Z} \mapsto m^{\mathcal{Z}}$ is increasing (due to a comparison principle applied to the ODE \eqref{dyn:m_fin_del}) and that, by the representation in \eqref{eq representation V with F}, the value $V_H$ is decreasing in the flow $m^{\mathcal{Z}}$.
It follows that $\mathcal{Z} \mapsto V_H^{\mathcal{Z}}(h,m,z)$ is decreasing, so that
\begin{equation}\label{V_ineq_finite_vartheta}
 V_H^{\cZ_2}(h,m,z) \leq V_H^{\widetilde{\cZ}_1}(h,m,z), \quad \text{ for any } h,m,z.
\end{equation}
Therefore, 
\begin{equation}
   \label{cZ2_tcZ1_ineq}
    V_H^{\cZ_2}( h^*, m^*, \cZ_2(h^*,m^*)) \leq V_H^{\widetilde{\cZ}_1}( h^*,m^*, \cZ_2(h^*,m^*)) =  V_H^{\widetilde{\cZ}_1}(h^*, m^*, \widetilde{\cZ}_1(h^*,m^*)),
\end{equation}
where the equality follows from $\cZ_2(h^*,m^*) = \widetilde{\cZ}_1(h^*,m^*)$.
It then follows from \eqref{cZ2_id} and \eqref{cZ2_tcZ1_ineq} that
\begin{equation}\label{VL_tZ1_ineq}
    0 \leq V_H^{\widetilde{\cZ}_1}(h^*, m^*, \widetilde{\cZ}_1(h^*,m^*)).
\end{equation}

Next, we prove that 
\begin{equation}
    \label{VL_tZ1_Z1_ineq}
     V_H^{\widetilde{\cZ}_1}( h^*,m^*, \widetilde{\cZ}_1(h^*,m^*) ) < V_H^{\cZ_1}(h^*, m^*, \cZ_1(h^*,m^*)). 
\end{equation}
However, because $\cZ_1$ is a threshold equilibrium, $V_H^{\cZ_1}(h^*,m^*, \cZ_1(h^*,m^*))=0$. Therefore, \eqref{VL_tZ1_ineq} and \eqref{VL_tZ1_Z1_ineq} contradict each other. Hence, the threshold equilibrium can only be unique.

To prove \eqref{VL_tZ1_Z1_ineq}, we compare $m_s^{\cZ_1}$ and $m_s^{\widetilde{\cZ}_1}$ for $s\geq t$ with the same initial value $m_t^{\cZ_1} = m_t^{\widetilde{\cZ}_1}= m^*$, but with different initial condition for $Z$. 
Because $m^{\cZ}$ follows the dynamics \eqref{dyn:m_fin_del} and $Z$ satisfies the equation \eqref{ass:Z}, the dynamics of $m^{\cZ}$ only depends on the relative difference of $Z_s$ and $\cZ(h_s, m_s)$ for $s\geq t$. 
Because $\widetilde{\cZ}_1$ is a parallel shift version of $\cZ_1$ (indeed, $\widetilde{\cZ}_1 (h,m) = {\cZ}_1 (h,m) + {\cZ}_2 (h^*,m^*) - {\cZ}_1 (h^*,m^*) $), thanks to \eqref{ass:Z}, the difference between $Z_s^{ \cZ_1(h_t,m_t)}$ and $\cZ_1(h_s,m_s)$ is the same as the difference between $Z_s^{\widetilde{\cZ}_1(h_t,m_t)}$ and $\widetilde{\cZ}_1(h_s,m_s)$. 
Therefore, adding the superscript to highlight the initial condition of $Z$, we have 
\begin{equation}\label{m_equ}
m_s^{\cZ_1(h^*,m^*), \cZ_1} = m_s^{\widetilde{\cZ}_1(h^*, m^*), \widetilde{\cZ}_1}, \quad \text{ for any } s\geq t.
\end{equation}
Using again the representation in  \eqref{eq representation V with F}, it then follows that 
\begin{align*}
 V_H^{\widetilde{\cZ}_1}  ( h^*, m^*, \widetilde{\cZ}_1(h^*,m^*)) 
 = &  \sup_{\tau \in\{\tau_n\}_n}  \mathbb{E} [F(t,\tau; \widetilde{\cZ}_1(h^*,m^*) + Z^{t,0},m^{\widetilde{\cZ}_1(h^*, m^*),\widetilde{\cZ}_1} )] \\
 = &  \sup_{\tau \in\{\tau_n\}_n}  \mathbb{E} [F(t,\tau; \widetilde{\cZ}_1(h^*,m^*) + Z^{t,0},m^{\cZ_1(h^*, m^*){\cZ}_1} )] \\
  < &  \sup_{\tau \in\{\tau_n\}_n}  \mathbb{E} [F(t,\tau; {\cZ}_1(h^*,m^*) + Z^{t,0},m^{\cZ_1(h^*, m^*), {\cZ}_1} )] \\
  = & V_H^{\cZ_1}(  h^*, m^*, {\cZ}_1(h^*,m^*)), 
\end{align*}
where the second equality above follows from \eqref{m_equ}. To see the strict inequality, let $\tau^*$ be the optimal Poisson stopping time for $\mathbb{E}[F(t, \tau; \cZ_1(h^*, m^*)+ Z^{t,0}, m^{\cZ_1})]$. Because $\tau^*$ can only be Poisson arrival time, $\mathbb{P}(\tau^* >0) =1$. Since $\beta_H$ is strictly decreasing in $Z$ and $\widetilde{\cZ}_1(h^*,m^*) < \cZ_1(h^*,m^*)$, we have 
\begin{align*}
 &\sup_{\tau \in\{\tau_n\}_n}  \mathbb{E} [F(t,\tau; {\cZ}_1(h^*,m^*) + Z^{t,0},m^{{\cZ}_1} )] = \mathbb{E} [F(t,\tau^*; {\cZ}_1(h^*,m^*) + Z^{t,0},m^{{\cZ}_1} )] \\
 <& \mathbb{E} [F(t,\tau^*; \widetilde{\cZ}_1(h^*,m^*) + Z^{t,0},m^{{\cZ}_1} )] \leq \sup_{\tau \in\{\tau_n\}_n}  \mathbb{E} [F(t,\tau; \widetilde{\cZ}_1(h^*,m^*) + Z^{t,0},m^{{\cZ}_1} )],
\end{align*}
where $m^{\cZ_1} = m^{\cZ_1(h^*,m^*), \cZ_1}$. Therefore, the claim \eqref{VL_tZ1_Z1_ineq} is confirmed.

\subsection{Proof of Proposition \ref{proposition intensity limit}}

The proof consists of two steps.
In order to simplify the notation, we set $V^{(\theta)} := V_H^{(\theta)}$ and $V^{(\infty)}:=V_H^{(\infty)}$.

\medskip

\noindent\underline{Step 1}.
We first want to show that  
$\mathcal Z ^{(\infty)}$ is the optimal stopping boundary associated to the function $V^{(\infty)}$; that is, 
$\mathcal Z ^{(\infty)} (h,m) = \inf \{ z \in \mathbb R \, | \, V^{(\infty)} (h,m,z) >0 \}$.
We can first show that there exist $\underline{\cZ}$ and $\overline{\cZ}$ such that any threshold $\cZ^{(\vartheta)}$ is sandwiched between. 
Because $V^{(\vartheta)}$ converges to $V^{(\infty)}$ locally uniformly, the convergence is uniform in $ \mathcal D \times [\underline{\cZ}, \overline{\cZ}]$. 
Therefore,
$$
0 = \lim_{\theta \rightarrow \infty}  V^{(\vartheta)}( h, m, \cZ^{(\vartheta)}(h,m)) = V^{(\infty)}(h, m, \cZ^{(\infty)}(h,m)), \quad \text{ for any } (h,m)\in \mathcal D.
$$

We thus remain to show that 
\begin{equation}
\label{eq for being the boundary}
    V^{(\infty)}  (h,m,z) >0
    \quad \text{when}\quad z 
    > \mathcal Z ^{(\infty)} (h,m).
\end{equation}
To this end, fix $z> \mathcal{Z}^{(\infty)}(h,m)$. 
Noticing that $V^{(\vartheta)}  (h,m, \mathcal Z ^{(\vartheta)} (h,m)) =0$ by   the equilibrium property of $\mathcal Z ^{(\vartheta)}$,  it is sufficient to show that there exists a constant $c$, which depends on $z$ but not on $\vartheta$, such that
\begin{equation}
\label{eq V theta over the boundary}
V^{(\vartheta)}  (h,m,z) - V^{(\vartheta)}  (h,m, \mathcal Z ^{(\vartheta)} (h,m))  \geq c >0, \quad \text{ for any $\vartheta$ large enough}.
\end{equation}

In order to show \eqref{eq V theta over the boundary}, 
take  generic $z_1, z_2$ such that $ \mathcal Z ^{(\vartheta)} (h,m) \leq z_1 < z_2 \leq z$.
We first notice that the Markovian dynamics for $m_s$ is decreasing in the $z$-variable; 
that is, that for any $ z_1 < z_2$ one has $m^{z_2 ,\mathcal Z ^{(\vartheta)}}_s \leq m^{ z_1,\mathcal Z ^{(\vartheta)}}_s $, for any $s>t$. Here $\{m_s^{z, \cZ^{\theta}}\}_{s\geq t}$ is the withdrawn share which follows the dynamics \eqref{dyn:m_fin_del} with $Z_t = z$ and $h_t = h$.
The inequality $m^{z_2 ,\mathcal Z ^{(\vartheta)}}_s \leq m^{ z_1, \mathcal Z ^{(\vartheta)}}_s $, together with the fact that the value function $V^{(\vartheta)} $ is nonincreasing in a generic withdrawn share $m$, implies that
$$
\begin{aligned}
V^{(\vartheta)} ( h,m, z_2  ) -  V^{(\vartheta)} ( h,m, z_1) &= 
    V^{(\vartheta)} ( z_2; m^{z_2,\mathcal Z ^{(\vartheta)}}  ) -  V^{(\vartheta)} ( z_1 ; m^{z_1,\mathcal Z ^{(\vartheta)}}  )  \\
    &\geq 
    V^{(\vartheta)} ( z_2; m^{z_1,\mathcal Z ^{(\vartheta)}}  ) -  V^{(\vartheta)} ( z_1; m^{z_1 ,\mathcal Z ^{(\vartheta)}}  ).  
\end{aligned}
$$
To continue estimating from below, we use the optimal stopping time $\tau_1$ for the initial condition $ z_1$ with withdrawn share $m^{z_1 ,\mathcal Z ^{(\vartheta)}}$, which is not necessarily optimal for the initial condition $ z_2$.
We obtain
\begin{equation}
\begin{aligned}
\label{eq:inequality1}
&V^{(\vartheta)} ( h, m, z_2  ) -  V^{(\vartheta)} ( h, m, z_1) \\
&
\quad \geq  
\mathbb{E}_t \bigg[\int_t^{\tau_1} \Big( e^{-\int_t^s (\nu + \beta_H (Z^{z_2}_u) + \eta m^{z_1,\mathcal Z ^{(\vartheta)}}_u)du} \big( r+\delta + (1-\gamma) \nu - \gamma\beta_H (Z^{ z_2 }_s)\big)  \\
&\quad \quad \quad \quad \quad \quad \quad 
-  e^{-\int_t^s (\nu + \beta_H (Z^{z_1}_u) + \eta m^{z_1,\mathcal Z ^{(\vartheta)}}_u)du} \big( r+\delta + (1-\gamma) \nu - \gamma\beta_H (Z^{z_1}_s)\big) \Big) ds \\
& \quad \quad \quad \quad \quad + (1-\gamma) \Big( e^{-\int_{t}^{\tau_1}(\nu + \beta_H(Z^{z_2 }_u) + \eta m^{z_1,\mathcal Z ^{(\vartheta)}}_u) du} - e^{-\int_{t}^{\tau_1}(\nu + \beta_H(Z^{z_1}_u) + \eta m^{z_1,\mathcal Z ^{(\vartheta)}}_u) du} \Big) \bigg] \\
&
\quad \geq  
\gamma \mathbb{E}_t \bigg[\int_t^{\tau_1} \Big( e^{-\int_t^s (\nu + \beta_H (Z^{z_1}_u) + \eta )du} \big( \beta_H (Z^{z_1}_s) -\beta_H (Z^{ z_2 }_s)\big) ds \bigg] \geq 0,  
\end{aligned}
\end{equation}
where the last inequality follows  by the monotonicity of $\beta_H$ and  $m^{z_1, \cZ^{(\vartheta)}}\leq 1$. 

In particular, the previous inequality gives
\begin{equation}\label{eq equation help}
\begin{aligned}    
V^{(\vartheta)}  (h,m,z) - V^{(\vartheta)}  (h,m, \mathcal Z ^{(\vartheta)} (h,m)) & = V^{(\vartheta)}  (h,m,z) -V^{(\vartheta)}  (h,m,z_1) \\
&\quad + V^{(\vartheta)}  (h,m, z_1) - V^{(\vartheta)}  (h,m, \mathcal Z ^{(\vartheta)} (h,m)) \\
&\geq V^{(\vartheta)}  (h,m,z) - V^{(\vartheta)}  (h,m, z_1).
\end{aligned}
\end{equation}

We next look at the case $z_2 = z$ such that $\mathcal Z ^{(\infty)} (h,m) < z_1 < z$.
For all $\vartheta$ large enough such that $\mathcal Z ^{(\vartheta)} (h,m) < z_1 < z$, 
the local Lipschitzianity of $\cZ^{(\infty)}(h,m)$ allows to find a  sufficiently small $\rho >0$ such that
$$
D_{h,m} \times B_z :=  \mathcal D \cap \big( [h,h+\rho] \times [m, m+\rho] \big) \times [z_1-\rho, z_1 +\rho] \subset \{ (l,\nu, y) \in \mathcal D \times \mathbb R \, | \, y > \cZ^{(\vartheta)} (l,\nu) \}.
$$
Notice also that $m^{(0)}$ (the solution of \eqref{dyn:m_fin_del} with $\vartheta = 0$) does not depend on the stopping boundary $\cZ^{(0)}$. 
Moreover, $m^{(\vartheta)}_s = m^{(0)}_s$ for any time $s\leq \hat \tau ^{(\vartheta)}$, where $\hat \tau^{(\vartheta)} =\inf \{ s \geq t \, | \, Z^{z_1}_{s} < \cZ^{(\vartheta)} (h_{s}, m^{(\vartheta)}_s)\}$. 
Thus, we have  that $\hat \tau^{(\vartheta)} =\inf \{ s \geq t \, | \, Z^{z_1}_{s} < \cZ^{(\vartheta)} (h_s, m^{(0)}_{s})\}$,
and setting 
$$
\tau := \inf \{ s \geq t \, | \, (h_s, m_s^{(0)}, Z^{z_1}_s) \notin D_{h,m} \times B_z  \},
$$
it follows that $\hat{\tau}^{(\vartheta)} \geq \tau >t$, for any sufficiently large $\vartheta$.
Therefore, since $\bar \tau ^{(\vartheta)} :=\inf \{ \tau_n \, | \, Z^{z_1}_{\tau_n} < \cZ^{(\vartheta)} (h_{\tau_n}, m^{(\vartheta)}_{\tau_n})\} \geq \hat \tau ^{(\vartheta)}$, the last line of \eqref{eq:inequality1} actually yields
$$
\begin{aligned}
&V^{(\vartheta)} ( h, m, z  ) -  V^{(\vartheta)} (  h, m, z_1)  \geq  
\gamma \mathbb{E}_t \bigg[\int_t^{\tau} \Big( e^{-\int_t^s (\nu + \beta_H (Z^{z_1}_u) + \eta )du} \big( \beta_H (Z^{z_1}_s)  - \beta_H (Z^{z }_s)\big) ds \bigg] =:c (z_1) >0,
\end{aligned}
$$
where $c(z_1)$ is positive due  to the fact that the term inside the integral is positive, by monotonicity of $\beta_H$ and by having assumed $z > z_1$.
This inequality, together with \eqref{eq equation help}, in turn gives \eqref{eq for being the boundary}. 

% and taking limits as $\theta \uparrow \infty$, we obtain
% $$
% \begin{aligned}
% &V^\infty ( z, h, m  ) -  V^\infty ( \bar z, h, m)  \geq  
%     \gamma \mathbb{E}_t \bigg[\int_t^{\tau} \Big( e^{-\int_t^s (\nu + \beta_H (Z^{\bar z^\infty}_u) + \eta )du} \big( \beta_H (Z^{\bar z}_s)  - \beta_H (Z^{z }_s)\big) ds \bigg] =:c (\bar z) >0,  
% \end{aligned}
% $$
% where $c(\bar z)$ is positive due  to the fact that the term inside the the integral is positive, by monotonicity of $\beta_H$ and by having assumed $z > \bar z$.
% This shows \eqref{eq for being the boundary}.

\medskip

\noindent\underline{Step 2}.
In this step we show that $V^{(\infty)}$ is the value function of a representative depositor who can withdraw at any time; that is, that $V^{(\infty)}  (h,m,z) = V^{\mathcal Z ^{(\infty)} } (h,m,z)$.

Since $m^{(\vartheta)}$ is the mean field flow associated to the unique threshold equilibrium 
$\cZ^{(\vartheta)}$, we have from \eqref{dyn:m_fin_del} that
\begin{equation}
\label{eq m theta equal}
 d m^{(\vartheta)}_t =   \nu (1-m^{(\vartheta)}_t)dt + dL^{(\vartheta)}_t, 
 \quad
 \text{with} \quad
 L^{(\vartheta)}_t:=\vartheta \int_0^t \mathbbm{1}_{\{Z_s \leq \cZ^{(\vartheta)} (h_s, m^{(\vartheta)}_s)\}} (h_s - m^{(\vartheta)}_s) ds.
   % m^{(\vartheta)}_s = 1 - (1-m) \exp \Big( {-\nu t - \theta \int_t^s 1_{\{Z_u < \cZ^{(\vartheta)} (m^{(\vartheta)}_u)\}} du}  \Big). 
\end{equation}

Fix an arbitrary sequence $\{\vartheta_n\}_{n\geq 1}$ such that
$\vartheta_n\to\infty$. The following argument is understood pathwise,
outside a common null set. Since $m^{(\vartheta_n)}$ is a withdrawal share,
$
0\leq m^{(\vartheta_n)}_t\leq h_t\leq 1$, for $t\geq0$.
Moreover, each $m^{(\vartheta_n)}$ is nondecreasing in time. Hence, by
Helly's selection theorem and a diagonal argument, there exists a
subsequence, which we do not relabel, and a nondecreasing,
right-continuous function $m^{(\infty)}$ such that
\begin{equation}
\label{eq limits for m}
m^{(\vartheta_n)}_t\longrightarrow m^{(\infty)}_t
\end{equation}
at every continuity point of $m^{(\infty)}$, and therefore for
$dt$-almost all $t\geq0$.

Similarly, the processes $L^{(\vartheta_n)}$ defined in
\eqref{eq m theta equal} are nondecreasing and satisfy, for every
$T>0$,
$
0\leq L^{(\vartheta_n)}_T
=
m^{(\vartheta_n)}_T-m^{(\vartheta_n)}_0
-\int_0^T\nu(1-m^{(\vartheta_n)}_s)ds
\leq 1.
$
Therefore, 
$L^{(\vartheta_n)}_t\to L^{(\infty)}_t$ at every continuity point of a
nondecreasing, right-continuous function $L^{(\infty)}$. Passing to the
limit in the integral form of \eqref{eq m theta equal}, using dominated
convergence, gives
\begin{equation}
\label{eq limiting m decomposition}
m^{(\infty)}_t
=
m^{(\infty)}_0
+\int_0^t \nu(1-m^{(\infty)}_s),ds
+L^{(\infty)}_t.
\end{equation}

%A comparison principle for ODEs and $h_s \geq m^\theta_s$ imply that [Do we need $\cZ^\theta$ is increasing in $\theta$ here?]
%$$
%m^{\bar \theta}_t \geq m^{ \theta}_t, \quad \text{if} \quad  \bar \theta \geq \theta. 
%$$

% [$m^\theta$ is the equilibrium withdrawn share. Hence $m^\theta \leq 1$ for any $\theta$. Therefore, we do not need the following estimate for the upper bound of $m^\theta$, right? GIORGIO: I would also say so]
% Moreover, again by a comparison principle for ODEs, we have that
% $m^\theta_t \leq M ^\theta_t$, with $M^\theta$ solution to 
% $$
% d M^{\theta}_t = \theta (h_t - M^{\theta}_t) dt +  \nu (1-M^{\theta}_t)dt.
% $$
% The latter equation admits the explicit solution
% $$
% M_t^{\theta}
% =
% e^{-(\theta+\nu)t}M_0^{\theta}
% +
% \theta\int_0^t e^{-(\theta+\nu)(t-s)}h_s\,ds
% +
% \frac{\nu}{\theta+\nu}
% \left(1-e^{-(\theta+\nu)t}\right)
% $$
% which is uniformly bounded in $t$ and $\theta$.

%{\color{blue}{Moreover, since $m^\theta$ is the equilibrium withdrawn share, we also have
%$$
%0 \leq m^\theta_t \leq 1, \quad \text{for any} \quad \theta. 
%$$
%This, together with the monotonicity in $\theta$, implies the existence of a limit function $m^\infty$, which is right-continuous, nondecreasing,  and
%\begin{equation}\label{eq limits for m}
%   \lim_\theta m^\theta_t = m^\infty _t, 
%\end{equation}
%for almost all $t\geq 0$. 
%Passing to the limit in \eqref{eq m theta equal}, we find a nondecreasing right-continuous process $L^\infty$  such that $L^\theta_t \to L^\infty_t$ $dt$-a.e.\ and
%$$
%m^\infty_t = m +  \int_0^t \nu (1-m^{\infty}_t)dt + L^\infty_t. 
%$$

To characterize this dynamics as in \eqref{dyn:m_inf_del_app}, notice first that on any open interval such that $Z_t > \cZ^{(\infty)} (h_t, m^{(\infty)}_t)$, for $\vartheta$ large enough we  have
$Z_t > \cZ^{(\vartheta)} (h_t, m^{(\vartheta)}_t)$.
Thus, one has $d L^{(\vartheta)}_t = 0$ and consequently $dL^{(\infty)}_t = 0$, so that
$$
d m^{(\infty)}_t = \nu (1-m^{(\infty)}_t)dt.
$$
Secondly, since $L^{(\vartheta)}$ is bounded in $\vartheta$ (again, because  $
0\leq L^{(\vartheta)}_t
=
m^{(\vartheta)}_t-m^{(\vartheta)}_0
-\int_0^t\nu(1-m^{(\vartheta)}_s),ds
\leq 1
$, for any $t\geq0$), we have
$$
 \int_0^t \mathbbm{1}_{\{Z_t \leq \cZ^{(\vartheta)} (h_t, m^{(\vartheta)}_t)\}} (h_t - m^{(\vartheta)}_t) dt  = \frac{L^{(\vartheta)}_t}\vartheta \to 0, \quad \text{as} \quad \vartheta \to \infty, 
$$
so that
$$
\lim_n \mathbbm{1}_{\{Z_t \leq \cZ^{(\vartheta_n)} (h_t, m^{(\vartheta_n)}_t)\}} (h_t - m^{(\vartheta_n)}_t)  = 0, \quad dt\text{-a.e.} 
$$
Moreover, since $Z_t \ne \cZ^{(\infty)} (h_t, m^{(\infty)}_t)$ $dt$-a.e., we conclude that 
\begin{equation}\label{eq limits pointwise indicato h-m}
    \lim_n \mathbbm{1}_{\{Z_t \leq \cZ^{(\vartheta_n)} (h_t, m^{(\vartheta_n)}_t)\}} (h_t - m^{(\vartheta_n)}_t) = \mathbbm{1}_{\{Z_t \leq \cZ^{(\infty)} (h_t, m^{(\infty)}_{t})\}} (h_t - m^{(\infty)}_{t}) = 0, \quad dt\text{-a.e.} 
\end{equation}
Therefore, 
on open intervals such that $Z_t < \cZ^{(\infty)} (h_t, m^{(\infty)}_t)$ we deduce  that $h_t = m^{(\infty)}_t$.

Finally, we want to show that if $\tau$ is an entry point into the region $Z_\tau \leq \cZ^{(\infty)} (h_\tau, m^{(\infty)}_{\tau})$ and $h_{\tau-} - m^{(\infty)}_{\tau-} >0$, then  
$h_\tau - m^{(\infty)}_{\tau} =0$.
Arguing by contradiction, suppose that $h_\tau - m^{(\infty)}_{\tau} > 0$. 
By right continuity, there exists $\varepsilon > 0$ such that $h_s - m^{(\infty)}_s > 0$ for any $s \in [\tau, \tau + \varepsilon ]$. 
Moreover, by monotonicity of $\cZ ^{(\infty)}$, since $m^{(\infty)}_{\tau} \leq m^{(\infty)}_{s}$ and $ h_\tau \leq h_s$ for any $s \in [\tau, \tau + \varepsilon ]$, we have
$$
\mathbbm{1}_{\{Z_s \leq \cZ^{(\infty)} (h_\tau, m^{(\infty)}_{\tau})\}} \leq \mathbbm{1}_{\{Z_s \leq \cZ^{(\infty)} (h_s, m^{(\infty)}_{s})\}} , \quad \text{for any $s \in [\tau, \tau + \varepsilon ]$}.
$$
Therefore, 
$$
\mathbbm{1}_{\{Z_s \leq \cZ^{(\infty)} (h_\tau, m^{(\infty)}_{\tau})\}} (h_s - m^{(\infty)}_s ) \leq \mathbbm{1}_{\{Z_s \leq \cZ^{(\infty)} (h_s, m^{(\infty)}_{s})\}} (h_s - m^{(\infty)}_s ), \quad \text{for any $s \in [\tau, \tau + \varepsilon ]$}.
$$
However, due to the effect of the Brownian motion (in particular, the law of iterated logarithm), the indicator function on the left hand side of the inequality above is one on a subset of $[\tau, \tau + \varepsilon ]$ which has positive Lebesgue measure, so that the left-hand side is strictly positive on that subset as well. 
This contradicts \eqref{eq limits pointwise indicato h-m}.

Summarizing, we have shown that $m^{(\infty)}$ solves the consistency dynamics \eqref{dyn:m_inf_del_app}. 

Taking limits along the subsequence $\{\vartheta_n\}_n$ as in \eqref{eq limits for m} allows to conclude that
$$
V^{(\infty)}  (h, m, z) = \lim_n V^{(\vartheta_n)}  ( h,m, z)= \lim_n V^{(\vartheta_n)}  (z;m^{(\vartheta_n)}) = V  (z;m^{(\infty)})
= V^{\mathcal Z ^{(\infty)}}  (h,m, z),
$$
which completes the proof.

\subsection{Proof of Proposition \ref{prop:uniqueness_vartheta_inf}}
\noindent \underline{Step 1}.
Assume that there exist two equilibrium threshold $\cZ_1$ and $\cZ_2$. Define their associated continuation and stopping regions:
\[
\mathcal{C}_i := \{(h,m,z)\,:\, z > \cZ_i(h,m)\} \quad \text{and} \quad \mathcal{S}_i := \{(h,m,z)\,:\, z\leq \cZ_i(h,m)\}, \, i=1,2.
\]
Recall that $\mathcal{D} = \{(h,m)\in[0,1]^2\,:\, m\leq h\}$. Without loss of generality, assume 
\[
c^* = \sup_{(h,m)\in \mathcal{D}} \big(\cZ_1(h,m) - \cZ_2(h,m)\big) >0. 
\]
If the previous inequality does not hold, we swap the role of $\cZ_1$ and $\cZ_2$ in the rest of the proof. Due to the continuity of $\cZ_1$ and $\cZ_2$, there exists $(h^*, m^*)\in \mathcal{D}$ such that $c^* = \cZ_1(h^*,m^*) - \cZ_2(h^*,m^*)$. Parallel shift $\cZ_2$ upward to introduce $\widetilde{\cZ}_2 := \cZ_2 + c^*$.  Then definition of $c^*$ implies
\begin{equation}
    \label{tZ2_larger_Z1}
    \widetilde{\cZ}_2 \geq \cZ_1 \quad \text{on } \mathcal{D}.
\end{equation}
Denote $z^* := \widetilde{\cZ}_2(h^*,m^*) = \cZ_1(h^*,m^*)$ and $\widetilde{\mathcal{C}}_2 := \{(h,m,z)\,:\, z > \widetilde{Z}_2(h,m)\}$.

Because $\cZ \mapsto V^{\cZ}_H(h,m,z)$ is decreasing, \eqref{tZ2_larger_Z1} and the same argument leading to \eqref{V_ineq_finite_vartheta} yield
\begin{equation}
    \label{V_ineq_inf_vartheta_1}
    V_H^{\widetilde{\cZ}_2}(h,m,z) \leq V_H^{\cZ_1}(h,m,z), \quad \text{for any } h, m, z.
\end{equation}

\medskip
\noindent\underline{Step 2}.
In order to derive a contradiction, note that the dynamics in \eqref{dyn:m_inf_del_app} depends on the relative position of $Z^{t,z}_s$ and $\cZ(h_s, m_s^{z, \cZ})$, meanwhile the dynamics of $h$ is given exogenously. Therefore, \eqref{ass:Z} implies that 
\begin{align}
 Z^{t,z}_s - \widetilde{\cZ}_2(h_s, m^{z, \widetilde{\cZ}_2}_s) = & Z^{t, z-c^*}_s - \cZ_2(h_s, m^{z-c^*, \cZ_2}_s) \quad \text{and} 
 \nonumber\\
 m^{z, \widetilde{\cZ}_2}_s = & m^{z-c^*, \cZ_2}_s, \, \text{ for any } s\geq t. \label{m_id_inf_vartheta}
\end{align}

Instead of deriving the reverse inequality to \eqref{V_ineq_inf_vartheta_1} at $(h^*, m^*, z^*)$ to obtain a contradiction as in the proof of Proposition \ref{prop:threshold_finite}, we show that 
\begin{equation}
    \label{V_ineq_inf_vartheta_2}
    V_H^{\widetilde{\cZ}_2}(h,m,z) > V_H^{\cZ_2}(h,m, z-c^*), \quad \text{for any } z > \widetilde{\cZ}_2(h,m).
\end{equation}
To this end, for any $(h,m,z) \in \widetilde{\mathcal{C}}_2$, i.e., $z> \widetilde{\cZ}_2(h,m)$, definition of $c^*$ implies $z-c^*>\cZ_2(h,m)$, i.e., $(h,m,z-c^*) \in \mathcal{C}_2$. Let $\tau^*_{\cZ_2}(h,m, z-c^*) := \inf \{s\geq t\,:\, Z^{t, z-c^*}_s  \leq \cZ_2 (h_s, m^{z-c^*, \cZ_2}_s)\}$. It is an optimal stopping time for the problem $V_H^{\cZ_2}(h,m,z-c6*)$. Equation \eqref{m_id_inf_vartheta} and the assumption \eqref{ass:Z} imply that 
\begin{equation}\label{tau_id}
\begin{split}
 \tau^*_{\cZ_2}(h,m, z-c^*) = & \inf \{s \geq t\,:\, Z^{t, z-c^*}_s \leq \cZ_2 (h_s, m_s^{z, \widetilde{\cZ}_2}) \} \\
 = & \inf \{s\geq t\,:\, Z^{t,z}_s \leq c^* + \cZ_2 (h_s, m_s^{z, \widetilde{\cZ}_2}) \} \\
 = & \inf\{s\geq t\,:\, Z^{t,z}_s \leq \widetilde{\cZ}_2 (h_s, m_s^{z, \widetilde{\cZ}_2})\} \\
 =: & \tau^*_{\widetilde{\cZ}_2} (h,m, z),  
\end{split}
\end{equation}
where $\tau^*_{\widetilde{\cZ}_2}(h,m,z)$ is an optimal stopping time for the problem $V_H^{\widetilde{\cZ}_2}(h,m,z)$.
As a consequence, 
\begin{align*}
 V_H^{\cZ_2}(h,m, z-c^*) = & \mathbb{E}\Big[\int_t^{\tau^*_{\cZ_2}(h,m, z-c^*)} e^{-\int_t^s (\nu + \beta_H (Z^{t, z-c^*}_u) + \eta m^{z-c^*, \cZ_2}_u)du} \big( r+\delta + (1-\gamma) \nu - \gamma\beta_H (Z^{t, z-c^*}_s)\big) ds \\
 & + (1-\gamma) \Big( e^{-\int_{t}^{\tau^*_{\cZ_2}(h,m, z-c^*)}(\nu + \beta_H(Z^{t, z-c^*}_u) + \eta m^{z-c^*,\cZ_2}_u) du} -1 \Big)\Big]\\
 = & \mathbb{E}\Big[\int_t^{\tau^*_{\widetilde{\cZ}_2}(h,m, z)} e^{-\int_t^s (\nu + \beta_H (Z^{t, z-c^*}_u) + \eta m^{z, \widetilde{\cZ}_2}_u)du} \big( r+\delta + (1-\gamma) \nu - \gamma\beta_H (Z^{t, z-c^*}_s)\big) ds \\
 & + (1-\gamma) \Big( e^{-\int_{t}^{\tau^*_{\widetilde{\cZ}_2}(h,m, z)}(\nu + \beta_H(Z^{t, z-c^*}_u) + \eta m^{z,\widetilde{\cZ}_2}_u) du} -1 \Big)\Big]\\
 < & \mathbb{E}\Big[\int_t^{\tau^*_{\widetilde{\cZ}_2}(h,m, z)} e^{-\int_t^s (\nu + \beta_H (Z^{t, z}_u) + \eta m^{z, \widetilde{\cZ}_2}_u)du} \big( r+\delta + (1-\gamma) \nu - \gamma\beta_H (Z^{t, z}_s)\big) ds \\
 & + (1-\gamma) \Big( e^{-\int_{t}^{\tau^*_{\widetilde{\cZ}_2}(h,m, z)}(\nu + \beta_H(Z^{t, z}_u) + \eta m^{z,\widetilde{\cZ}_2}_u) du} -1 \Big)\Big]\\
 = & V_H^{\widetilde{\cZ}_2}(h,m,z).
\end{align*}
where the second equality follows from \eqref{m_id_inf_vartheta} and \eqref{tau_id}, the inequality holds thanks to $Z^{t, z-c^*}_{\cdot} < Z^{t,z}_{\cdot}$, $\beta_H$ is strictly decreasing, $r+\delta + (1-\gamma)\nu - \beta_H(\cdot) \geq 0$, and $\mathbb{P} \big(\tau^*_{\cZ_2}(h,m, z-c^*)>0 \big) =1$ because $z - c^* > \cZ_2 (h,m)$.  Therefore, the claim in \eqref{V_ineq_inf_vartheta_2} is confirmed.

\medskip

\noindent\underline{Step 3}. Introduce
\[
W(h,m,z) := V_H^{\cZ_1}(h,m,z) - V_H^{\cZ_2}(h,m,z-c^*).
\]
For any $z> \widetilde{\cZ}_2(h,m)$, \eqref{V_ineq_inf_vartheta_1} and \eqref{V_ineq_inf_vartheta_2} imply that 
\begin{equation}\label{W_ineq}
W(h,m,z) \geq V_H^{\widetilde{Z}_2}(h,m,z) - V_H^{\cZ_2}(h,m,z-c^*) >0.
\end{equation}
Meanwhile,
\begin{equation}\label{W_id}
\begin{split}
W(h^*, m^*, z^*) = & V_H^{\cZ_1}(h^*, m^*, z^*) - V_H^{\cZ_2}(h^*, m^*, z^* - c^*) \\
= & V_H^{\cZ_1}\big(h^*, m^*, \cZ_1(h^*, m^*)\big) - V_H^{\cZ_2}\big(h^*, m^*, \cZ_2(h^*,m^*)\big)\\
= & 0,
\end{split}
\end{equation}
where the second equality follows from $z^* = \widetilde{\cZ}_2(h^*,m^*) = \cZ_1(h^*, m^*)$, $z^*- c^* = \widetilde{\cZ}_2(h^*,m^*) - c^* = \cZ_2(h^*, m^*)$, and the third equality holds because both $\cZ_1$ and $\cZ_2$ are threshold equilibria. 

Now, for any $z> \widetilde{\cZ}_2(h,m)$, we have $z > \cZ_1(h,m)$ due to \eqref{tZ2_larger_Z1}, hence $(h,m,z) \in \mathcal{C}_1$. Moreover, $z-c^* > \cZ_2(h,m)$, hence $(h,m, z-c^*) \in \mathcal{C}_2$. Therefore, 
\begin{equation}\label{Vzi_equ}
\begin{split}
 &\big(\mathcal{L} - q(z,m) \big) V_H^{\cZ_1}(h,m,z) + f(z,m) =0,\\
 & \big(\mathcal{L} - q(z-c^*, m) \big) V_H^{\cZ_2}(h,m,z-c^*) + f(z-c^*,m) =0,
\end{split}
\end{equation}
where $\mathcal{L}, q$, and $m$ are defined in \eqref{def:cL}. For any $(h,m)$, introduce
\begin{equation}\label{h0_m0}
h^0_t := 1- (1-h)e^{-(\nu + \theta_{LH}) t}, \quad  m^0_t := 1- (1-m) e^{-\nu t}, \quad \text{and} \quad \widehat{V}^{\cZ_i}(t,z) := V_H^{\cZ_i}\big(h^0_t, m^0_t, z\big). 
\end{equation}
Then
\[
\partial_t \widehat{V}^{\cZ_i}(t,z) = (\nu + \theta_{LH}) \big(1- h^0_t \big) \partial_h V_H^{\cZ_i}\big(h^0_t, m^0_t, z\big) + \nu \big(1- m^0_t \big) \partial_m V_H^{\cZ_i}\big(h^0_t, m^0_t, z\big),
\]
and we obtain from \eqref{Vzi_equ} that $\widehat{V}^{\cZ_1}\big(t, z \big)$ and $\widehat{V}^{\cZ_2}\big(t, z -c^* \big)$ satisfy parabolic PDEs
\begin{equation}
    \label{PDE:hat_V}
\begin{split}
 &\Big(\partial_t + \mu \partial_z + \frac12 \sigma^2 \partial^2_{zz} - q\big(z, m^0_t\big)\Big) \widehat{V}^{\cZ_1}\big(t, z \big) + f\big(z, m^0_t \big) =0, \\
 & \Big(\partial_t + \mu \partial_z + \frac12 \sigma^2 \partial^2_{zz} - q\big(z-c^*, m^0_t\big)\Big) \widehat{V}^{\cZ_2}\big(t, z-c^* \big) + f\big(z -c^*, m^0_t \big) =0.
\end{split}
\end{equation}
Introduce $\widehat{W}(t, z) := W(h^0_t, m^0_t, z)$, the previous equations imply
\begin{equation}\label{W_diff_ineq}
\begin{split}
 &\Big(\partial_t + \mu \partial_z + \frac12 \sigma^2 \partial^2_{zz} - q(z, m^0_t)\Big) \widehat{W}(t, z) \\
 &=  \big(q(z,m^0_t) - q(z-c^*, m^0_t) \big) V_H^{\cZ_2}\big(h^0_t, m^0_t, z\big) - f(z, m^0_t) + f(z-c^*, m^0_t)\\
 &=  \big(\beta_H (z) - \beta_H(z-c^*) \big) \Big[1+  V_H^{\cZ_2}\big(h^0_t, m^0_t, z\big)\Big]\\
 & <0,
\end{split}
\end{equation}
where the inequality follows from $\beta_H(z) < \beta_H(z-c^*)$ and $1+ V^{\cZ_2}_H >0$.

Let $(h,m) = (h^*,m^*)$ in \eqref{h0_m0}, then $(h^0_0, m^0_0) = (h^*, m^*)$. We have from \eqref{W_ineq} and \eqref{W_id} that
\begin{align*}
    \widehat{W}(t, z) >0, & \quad z >  \widetilde{\cZ}_2\big(h^0_t, m^0_t\big),\\
    \widehat{W}(0, z^*) =0, & \quad z^* = \widetilde{\cZ}_2(h^0_0, m^0_0).
\end{align*}
Therefore, $(0,z^*)$ is a minimal point on the domain $\Xi := \{(t,z)\,:\, t\geq 0, z\geq \widetilde{\cZ}_2(h^0_t, m^0_t)\}$. If $\Xi$ satisfies the parabolic interior ball condition at $(0, z^*)$ \footnote{There exists a ball centered $B$ at $(\hat{t}, \hat{z})$ with $\hat{t}\geq 0$, such that $(\tilde{t}, \tilde{z})\in \Xi$ for any $(\tilde{t}, \tilde{z}) \in B$ with $\tilde{t}\geq0$, moreover, $(0, z^*)\in \partial B$ and the boundaries of $B$  and $\Xi$ only intersect at $(0, z^*)$.}, Hopf Lemma (a strong form of maximum principle; see, e.g., Lemma II.2.8 in \cite{lieberman1996second}) implies that 
\begin{equation}\label{hatW_z_ineq}
 \partial_z \widehat{W}(0, z^*) >0.
\end{equation}
However, this contradicts with smooth-pasting for $\widehat{W}(0, z^*)$, because
\begin{align*}
    \partial_z \widehat{W}(0, z^*) = & \partial_z W(h^*, m^*, z^*) = \partial_z V_H^{\cZ_1}(h^*,m^*, z^*) - \partial_H^{\cZ_2}(h^*,m^*, z^*-c^*)\\
    = &\partial_z V_H^{\cZ_1}\big(h^*, m^*, \cZ_1(h^*, m^*) \big) - \partial_z V_H^{\cZ_2}\big(h^*, m^*, \cZ_2(h^*, m^*)\big) = 0,
\end{align*}
where the last equality follows from the smooth-pasting condition for $V_H^{\cZ_1}$ and $V_H^{\cZ_2}$.

Finally, we do not assume that $\Xi$ satisfies the interior ball condition at $(0,z^*)$. We assume in the statement of Proposition \ref{prop:uniqueness_vartheta_inf} that $\cZ(h,m)$ is continuously differentiable in $(h,m)$, then $\widetilde{\cZ}_2(h^0_t, m^0_t)$ is continuously differentiable in $t$. We can still obtain \eqref{hatW_z_ineq} from the following lemma, whose proof is presented at the end of the next section.
\begin{lem}\label{lem:hatW_z}
Suppose that $\widetilde{Z}_2(h^0_t, m^0_t)$ is locally Lipschitz in $t$, then $\partial_z \widehat{W}(0, z^*)>0$.
\end{lem}

\begin{rem}\label{rem:VH_regularity}
 Using the regularity theory for non-degenerate parabolic PDEs,
 the assumption $V^{\cZ}_H \in C^{1,1,2}(\mathcal{C})$ in Proposition \ref{prop:uniqueness_vartheta_inf} can be weaken to be $\widehat{V}^{\cZ} \in C(\widehat{\mathcal{C}})$, where $\widehat{V}^{\mathcal{C}}$ satisfies 
 \[
 \Big(\partial_t + \mu \partial_z + \frac12 \sigma^2 \partial^2_{zz} - q(z, m^0_t) \Big) \widehat{V}^{\cZ}(t,z) + f(z, m^0_t) =0, \quad (t,z)\in \widehat{\mathcal{C}} := \{(t,z)\,:\, z > \cZ(h^0_t, m^0_t)\}.
 \]
 Indeed, consider a parabolic domain $\Omega_p = [0,T] \times [z_1, z_2] \subset \widehat{\mathcal{C}}$ and the following boundary value problem
 \begin{align*}
&\Big(\partial_t + \mu \partial_z + \frac12 \sigma^2 \partial^2_{zz} - q(z, m^0_t) \Big) \widetilde{V}(t,z) + f(z, m^0_t) =0, \quad (t,z) \in [0,T) \times (z_1, z_2),\\
& \widetilde{V}(t,z) = \widehat{V}^{\cZ}(t,z), \quad (t,z) \in [0,T]\times \{z_1, z_2\} \cup \{T\} \times [z_1, z_2].
 \end{align*}
Given the continuity of $\widehat{V}^{\cZ}$, 
it follows from \cite[Corollary 2.4.3]{krylov2008lectures} that $\widetilde{V} \equiv \widehat{V}^{\cZ}$ on $\Omega_p$ and $\widehat{V}^{\cZ} \in C^{1,2}([0,T)\times (z_1, z_2))$. Because the parabolic domain is chosen arbitrarily, we have $\widehat{V}^{\cZ} \in C^{1,2}(\widehat{\mathcal{C}})$.
\end{rem}

\section{Other proofs}\label{app:other_proofs}

\subsection{Proof of Lemma \ref{lem:mono_scheme}}

We will prove the statement for the decreasing scheme by induction. When $n=0$, given that $m^{(0)}_t = 1$ for any $t$, then $m^{(0)}_t \geq m^{(1)}_t$ is automatically satisfied. Suppose that \eqref{m:dec} holds, we want to show \eqref{m:dec} is satisfied with $n$ replaced by $n+1$. To this end, $m^{(n)}_t \geq m^{(n+1)}_t$ and Lemma \ref{lemma monotonicity of brm} combined implies $\tau^{(n+1)}_w \leq \tau^{(n+2)}_w$. Therefore,
\[
m^{(n+1)}_t = \mathbb{P}\big(\tau_{\nu} \wedge \tau^{(n+1)}_w \leq t\,|\, \mathcal{F}^Z_t\big) \geq \mathbb{P}\big(\tau_{\nu} \wedge \tau^{(n+2)}_w \leq t \,|\, \mathcal{F}^Z_t \big) = m^{(n+2)}_t, \quad \text{for any } t.
\]
As a result,  \eqref{m:dec} holds for any $n$ by induction. 

The proof of \eqref{m:inc} is similar.

\subsection{Proof of Proposition \ref{prop:no_run}}\label{app:non-absorbing}
We prove by contradiction. Assume that the earliest strategic withdrawal time is finite. Recall that the discount rates $\beta_L$ and $\beta_H$ form a two-state Markov chain. It follows the same dynamic programming argument as in \eqref{Poisson:VH} that the value functions for both types of depositors are 
\begin{align*}
V_L(t;m) =& \sup_{\tau_w} \int_t^{\tau_w} e^{-\int_t^s ( \beta_L + \theta_{LH} + \eta m_u) du} \big(r+\delta - \beta_L - (1-\gamma) \eta m_s + \theta_{LH} V_H (s; m)\big) ds,\\
V_H(t;m) = & \sup_{\tau_w} \int_t^{\tau_w} e^{-\int_t^s (\beta_H + \theta_{HL} + \eta m_u) du} \big(r+\delta - \beta_H - (1-\gamma) \eta m_s + \theta_{HL} V_L (s; m)\big) ds.
\end{align*}
Becuase the withdrawn share $m$ is nondecreasing in time, one can use the representation \eqref{decomp:1+J} and Markov property to show that both $V_L$ and $V_H$ are non-increasing in time. Therefore, the optimal stopping time for high and low-types are
\begin{align*}
\tau^*_w(L;m) = &\inf\{t\geq 0\,:\, r+\delta - \beta_L  - (1-\gamma) \eta m_t + \theta_{LH} V_H(t;m) \leq 0\},\\
\tau^*_w(H;m) = & \inf\{t\geq 0\,:\, r+\delta - \beta_H - (1-\gamma) \eta m_t + \theta_{HL} V_L(t;m) \leq 0\}.
\end{align*}
Given that $\beta_L < r+\delta - (1-\gamma) \eta$ and $V_H(t;m)\geq 0$, the patient depositors never withdraw strategically, i.e., $\tau^*_w(L;m) =\infty$. 

Denote $\underline{t}$ the first time of strategic withdrawal from impatient depositors in the earliest-run equilibrium. Because there is no exogenous withdrawals, the mass of both types follows
\begin{align*}
 dH_t =& - \theta_{HL} H_t dt + \theta_{LH} L_t dt, \quad H_0 = 0,\\
 dL_t =& -\theta_{LH} L_t dt + \theta_{HL} H_t dt, \quad L_0 = 1.
\end{align*}
The solution of the ODE system is 
\begin{equation}\label{ODE:HL}
\begin{split}
 L_t = & \frac{\theta_{HL}}{\theta_{LH} + \theta_{HL}} + \frac{\theta_{LH}}{\theta_{LH} + \theta_{HL}} e^{-(\theta_{LH} + \theta_{HL})t},\\
 H_t = & \frac{\theta_{LH}}{\theta_{LH} + \theta_{HL}} - \frac{\theta_{LH}}{\theta_{LH} + \theta_{HL}} e^{-(\theta_{LH} + \theta_{HL})t},
\end{split}
\end{equation}
when $t< \underline{t}$. 

Because $\underline{t}$ is the first strategic withdrawal time, $V_H(\underline{t}; m)=0$. Given that $m_t$ is nondecreasing and $V_L(t;m)$ is nonincreasing in time, $V_H(t; m)$ remains zero when $t\geq \underline{t}$. Therefore, after $\overline{t}$, whenever a patient depositor receives a discount rate shock and becomes impatient, it is optimal to withdrawal immediately. As a result, there is no more transition from impatient depositors to patient depositors after $\underline{t}$, and the mass of patient depositors follows 
$dL_t = - \theta_{LH} L_t dt$, 
and $m_t = 1-L_t$ for $t\geq \underline{t}$. The value of patient depositors after $\underline{t}$ is 
\begin{equation}\label{VL_after_ut}
V_L(t;1-L_t) = \int_t^\infty e^{-\int_t^s (\beta_L + \theta_{LH} + \eta (1-L_u)) du} \big( r+ \delta - \beta_L - (1-\gamma) \eta (1-L_s)\big) ds.
\end{equation}
Now come back to $\underline{t}$. Because all impatient depositors withdrawal at time $\underline{t}$, $m_{\underline{t}} = 1-L_{\underline{t}}$, and the optimality of withdrawal for impatient depositors at $\underline{t}$ implies that 
\begin{equation}\label{VH_ut}
r + \delta - \beta_H - (1-\gamma) \eta (1-L_{\underline{t}}) + \theta_{HL} V_L(\underline{t}; 1-L_{\underline{t}}) =0,
\end{equation}
where $L_t$ is given in \eqref{ODE:HL} and $V_L$ is from \eqref{VL_after_ut}. 

Observe from \eqref{ODE:HL} that $1-L_{\underline{t}} \leq \frac{\theta_{LH}}{\theta_{LH} + \theta_{HL}}$ for any possible value of $\underline{t}$. Meanwhile, because $m \mapsto V_L(t;m)$ is nonincreasing, we obtain
\[
V_L(\underline{t}; 1-L_{\underline{t}}) \geq V_L\big(\underline{t}, \frac{\theta_{LH}}{\theta_{LH} + \theta_{HL}}\big) = \frac{r+\delta - \beta_L - (1-\gamma) \eta \frac{\theta_{LH}}{\theta_{LH} + \theta_{HL}}}{\beta_L + \theta_{LH} + \eta \frac{\theta_{LH}}{\theta_{LH} + \theta_{HL}}},
\]
where the equality follows from \eqref{VL_after_ut} with $1-L_t = \frac{\theta_{LH}}{\theta_{LH} + \theta_{HL}}$. Now the assumption 
\[
r+\delta - \beta_H - (1-\gamma) \eta \frac{\theta_{LH}}{\theta_{LH} + \theta_{HL}} + \theta_{HL} \frac{r+\delta - \beta_L - (1-\gamma) \eta \frac{\theta_{LH}}{\theta_{LH} + \theta_{HL}}}{\beta_L + \theta_{LH} + \eta \frac{\theta_{LH}}{\theta_{LH} + \theta_{HL}}} > 0
\]
implies that 
\[
r+\delta - \beta_H - (1-\gamma) \eta (1-L_{\underline{t}}) + \theta_{HL} V_L(\underline{t} ; 1-L_{\underline{t}}) >0,
\]
for any possible value of $\underline{t}$. This means that the new flow of benefit for impatient depositors is still positive at $\underline{t}$, contradicting with the assumption that it is optimal for impatient depositor to withdraw to time $\underline{t}$. 

Given that $\underline{t}=\infty$ and the strategic withdrawal time for any equilibrium must be at least $\underline{t}$, the only equilibrium is $\tau^*_w = \infty$ and $m\equiv 0$.

\subsection{Proof of Proposition \ref{prop:convergence}}

\noindent\underline{Statement (i)}:
We first show that $(\tau_w^*, m^*)$ defined in \eqref{eq:definition limits} is an equilibrium.
Noticing that $m^*$ already satisfies the consistency condition, we limit ourself to show the optimality of $\tau_w^*$ for $m^*$.

First, the left-continuity of the map $z \mapsto \mathds{1}_{[0,t]}(z)$ and the fact that $\tau_w^{(n)} \leq \tau_w ^{(n+1)}$ implies that
$$
m^*_t = \mathbb{P} (\tau_{\nu} \wedge \tau^*_w \leq t\,|\, \mathcal{F}^Z_t) 
= \mathbb{E} \left[ \mathds{1}_{[0,t]}(\tau_{\nu} \wedge \tau^*_w ) \,|\, \mathcal{F}^Z_t \right] 
= \lim_n \mathbb{E} \left[ \mathds{1}_{[0,t]}(\tau_{\nu} \wedge \tau^n_w )\,|\, \mathcal{F}^Z_t \right]  = \lim_n m_t^{(n)},
$$
where we have used the dominated convergence theorem.
Moreover, using the optimality of $\tau_w^{(n)}$ for $m^{(n)}$, for any stopping time $\tau$ we can write
$$
J (t,Z_t, X_t; \tau_w^{(n)}, m^{(n)}) \geq J (t,Z_t, X_t; \tau, m^{(n)}).
$$
Taking limits as $n \to \infty$ in the previous inequality and using that $m^*_t  = \lim_n m_t^{(n)}$ and $\tau_w^* = \lim_n \tau_w^{(n)}$, from the continuity of $J$, it follows that
$$
J (t,X_t; \tau_w^*, m^*) \geq J (t,X_t; \tau, m^*),
$$
which is the optimality of $\tau^*_w$.
Thus, $(\tau_w^*, m^*)$ is an equilibrium.

We next show that $(\tau_w^*, m^*)$ is the equilibrium with earlier withdrawal time. 
If $(\tau_w,m)$ is another equilibrium, then we have $m_t \leq m^{(0)}_t=1$ for all $t$.
Hence, by monotonicity of the best response map in Proposition \ref{lemma monotonicity of brm} and by optimality of $\tau_w$ for $m$, we have
$ \tau_w^{(1)} \leq \tau_w$.
Thus, by consistency of $m$ we obtain 
$m  _t 
= \mathbb P \big(\tau_{\nu}\wedge\tau_w \leq t\,|\, \mathcal{F}^Z_t\big)
\leq \mathbb P \big(\tau_{\nu}\wedge\tau^{(1)}_w \leq t\,|\, \mathcal{F}^Z_t\big) = m^{(1)} _t
$ for all $t$. 
Repeating this argument, by induction we find $  \tau_w^{(n)} \leq \tau_w$ for all $n$, which in turn implies that $\tau^*_w \leq \tau_w$. 

\medskip

\noindent\underline{Statement (ii)}: 
The argument for the increasing scheme is slightly different.
In particular, we will show  $m^*_t  = \lim_n m_t^{(n)}$ only $dt$-a.e., which however is enough to complete the proof.
%Let $\tau_{\nu}$ be an exponential random variable independent of all other random variables. 

Define the processes 
\[
\xi^{(n)}_t := \mathds{1}_{[\tau_{\nu} \wedge \tau ^{(n)}, \infty )} (t) \quad \text{and} \quad \xi^*_t := \mathds{1}_{[ \tau_{\nu} \wedge \tau^*, \infty )} (t).
\]
Both stochastic processes are nondecreasing, bounded, and right-continuous. 
Moreover, by definition of $m^{(n)}$ and $m^*$, we have
$$
 m^{(n)} _t = \mathbb P \big(\tau_{\nu}\wedge\tau^{(n)}_w \leq t\,|\, \mathcal{F}^Z_t\big)  = \mathbb E [\xi^{(n)} _t\,|\, \mathcal{F}^Z_t] \quad \text{and} \quad m^* _t := \mathbb P \big(\tau_{\nu} \wedge \tau^*_w \leq t\,|\, \mathcal{F}^Z_t\big) = \mathbb E [\xi^* _t\,|\, \mathcal{F}^Z_t].
$$
%Both $m^{(n)}$ and $m^*$ are nondecreasing, bounded, and right-continuous. 
Recall from \eqref{m:inc} that $\tau^{(n+1)}_w \leq \tau^{(n)}_w$, for which we find $\xi^{(n)} \leq \xi^{(n+1)}$ for all $t$.
Moreover, the definition of $\tau^*$ implies that 
$$
\lim_{n}\xi^{(n)}_t = \xi^* _t, \quad dt\otimes \mathbb P\text{-a.e.}
$$
The dominated convergence theorem thus gives, for a generic $\alpha >0$, that $\lim_n  \int_0^\infty e^{-\alpha s} \mathbb E [\xi^*_s - \xi^{(n)}_s\,|\, \mathcal{F}^Z_t] ds =0$.
Thus, since $\xi^*_s - \xi^{(n)}_s \geq 0$, this implies that $\lim _n \mathbb E [\xi^{(n)} _t\,|\, \mathcal{F}^Z_t]= \mathbb E [\xi^* _t\,|\, \mathcal{F}^Z_t]$, $dt$-a.e.
Therefore, by using one of the previous identities we obtain 
\[
\lim_n m^{(n)}_t = \lim _n \mathbb E [\xi^{(n)} _t\,|\, \mathcal{F}^Z_t]= \mathbb E [\xi^* _t\,|\, \mathcal{F}^Z_t] = m^*_t, \quad dt\text{-a.e.}
\]

The rest of the  proof follows closely the proof of statement (i), and is therefore only sketched.
Indeed, the latter  convergence is enough to take limits in the inequality
$
J (t,X_t; \tau_w^{(n)}, m^{(n)}) \geq J (t,X_t; \tau, m^{(n)}),
$
in order to deduce that
$
J (t,X_t; \tau_w^*, m^*) \geq J (t,X_t; \tau, m^*),
$
which is the optimality of $\tau^*_w$.

The fact that $(\tau_w^*, m^*)$ is the equilibrium with latest withdrawal time is deduced by taking an equilibrium $(\tau_w,m)$, observing that  $m_t \geq m^{(0)}_t =0$ for all $t$, and then concluding by induction that $\tau_w^{(n)} \geq \tau_w$ for all $n$.

\subsection{Proof of Lemma \ref{lem:hatW_z}}

Let
$$
 b(s):=\widetilde{\cZ}_2(h_s^0,m_s^0),\qquad s\geq0,
$$
Under the regularity assumptions of Proposition \ref{prop:uniqueness_vartheta_inf}, the derivatives of the boundary exist almost everywhere and are bounded on the
compact state space $\mathcal D$; consequently, $b$ is Lipschitz and there is a
constant $K\geq0$ such that
\begin{equation}\label{eq:K}
 |\mu-\dot b(s)|\leq K\qquad\text{for a.e. }s\geq0.
\end{equation}

Fix $\delta>0$. Let $r>0$ be chosen below and take
$\varepsilon\in(0,r)$. Let the common state start from
$z^*+\varepsilon$:
$$
 Z_s^\varepsilon:=z^*+\varepsilon+\mu s+\sigma B_s^c,
 \qquad s\geq0.
$$
Define
\begin{equation}\label{eq:tau0}
 \tau_0^\varepsilon
 :=\inf\bigl\{s\geq0:Z_s^\varepsilon\leq b(s)\bigr\}.
\end{equation}
For every $s<\tau_0^\varepsilon$,
$$
 Z_s^\varepsilon>b(s)=\widetilde{\cZ}_2(h_s^0,m_s^0)
 \geq\cZ_1(h_s^0,m_s^0)
$$
and, since $\widetilde{\cZ}_2=\cZ_2+c^*$,
$$
 Z_s^\varepsilon-c^*>\cZ_2(h_s^0,m_s^0).
$$
Hence, before $\tau_0^\varepsilon$, neither threshold is reached and the two
withdrawn-share processes coincide with the exogenous-withdrawal trajectory:
\begin{equation}\label{eq:safe-m}
 m_s^{z^*+\varepsilon,\cZ_1}
 =m_s^{z^*+\varepsilon,\widetilde{\cZ}_2}
 =m_s^{z^*+\varepsilon-c^*,\cZ_2}
 =m_s^0,
 \qquad 0\leq s<\tau_0^\varepsilon.
\end{equation}

The next localization procedure allows to take care of a potential jump of $m$. For every $n$ sufficiently large such that
$n^{-1}<\varepsilon$, set
\begin{align}
 \tau_n^\varepsilon
 &:=\inf\bigl\{s\geq0:Z_s^\varepsilon\leq b(s)+n^{-1}\bigr\},
 \label{eq:taun}\\
 \tau_r^\varepsilon
 &:=\inf\bigl\{s\geq0:Z_s^\varepsilon\geq b(s)+r\bigr\}.
 \label{eq:taur}
\end{align}
By continuity, $\tau_n^\varepsilon\leq\tau_0^\varepsilon$, with strict
inequality on $\{\tau_0^\varepsilon<\infty\}$. Define
$$
 \rho_n^\varepsilon:=\tau_n^\varepsilon\wedge\tau_r^\varepsilon\wedge\delta.
$$
On $[0,\rho_n^\varepsilon]$, relation \eqref{eq:safe-m} holds and both value
functions appearing in the definition of $W$ remain in their continuation regions. The interior
regularity of $W$ due to the assumptions in Proposition \ref{prop:uniqueness_vartheta_inf} therefore permits the application of
It\^o's formula to
$
 e^{-\int_0^s q(Z_u^\varepsilon,m_u^0)\,du}
 \widehat W(s,Z_s^\varepsilon).
$
Using \eqref{W_diff_ineq} gives
\begin{align}
 W(h^*,m^*,z^*+\varepsilon)
 =\E_{z^*+\varepsilon}\Bigg[&
 \int_0^{\rho_n^\varepsilon}
 e^{-\int_0^s q(Z_u^\varepsilon,m_u^0)\,du}
 \Big[(q-\partial_s - \mathcal{L})W](s,h^0_s, m^0_s, Z_s^\varepsilon)\,ds \nonumber\\
 &+e^{-\int_0^{\rho_n^\varepsilon}q(Z_u^\varepsilon,m_u^0)\,du}
 W\bigl(h_{\rho_n^\varepsilon}^0,m_{\rho_n^\varepsilon}^0,
 Z_{\rho_n^\varepsilon}^\varepsilon\bigr)\Bigg]. 
 \label{eq:Ito-W}
\end{align}
On the event
$\{\tau_r^\varepsilon<\tau_n^\varepsilon\wedge\delta\}$,
$$
 Z_{\tau_r^\varepsilon}^\varepsilon
 =b(\tau_r^\varepsilon)+r.
$$
The function
$s \mapsto
 W\bigl(h_s^0,m_s^0,b(s)+r\bigr)$
is continuous and strictly positive on $[0,\delta]$. Therefore,
\begin{equation}\label{eq:ar}
 a_r:=\min_{0\leq s\leq\delta}
 W\bigl(h_s^0,m_s^0,b(s)+r\bigr)>0.
\end{equation}
Since $0\leq m_s^0\leq1$ and $\beta_H$ is bounded, with
$$
 \overline q:=\nu+\sup_{z\in\mathbb R}\beta_H(z)+\eta,
$$
we obtain from \eqref{eq:Ito-W}, using also that $W>0$,
\begin{equation}\label{eq:W-prob-n}
 W(h^*,m^*,z^*+\varepsilon)
 \geq a_r e^{-\overline q\delta}
 \P_{z^*+\varepsilon}
 \bigl(\tau_r^\varepsilon<\tau_n^\varepsilon\wedge\delta\bigr).
\end{equation}
As $n\to\infty$, $\tau_n^\varepsilon\uparrow\tau_0^\varepsilon$ almost
surely. Hence
\begin{equation}\label{eq:W-prob}
 W(h^*,m^*,z^*+\varepsilon)
 \geq a_r e^{-\overline q\delta}
 \P_{z^*+\varepsilon}
 \bigl(\tau_r^\varepsilon<\tau_0^\varepsilon\wedge\delta\bigr).
\end{equation}
It remains to estimate the probability on the right-hand side of \eqref{eq:W-prob} from below. The proof is split into three steps.

\medskip
\noindent\underline{Step 1.}
Let
$$Y_s^\varepsilon:=Z_s^\varepsilon-b(s),
 \qquad
\tau^\varepsilon:=\tau_0^\varepsilon\wedge\tau_r^\varepsilon.
$$
Then $Y_0^\varepsilon=\varepsilon$, $0<Y_s^\varepsilon<r$ for
$s<\tau^\varepsilon$, and
$$
 dY_s^\varepsilon=(\mu-\dot b(s))\,ds+\sigma\,dB_s^c.
$$
Take $g(x):=x(r-x)$ on $[0,r]$. For every $T>0$, It\^o's formula yields
\begin{align}
 &\E\bigl[g(Y_{\tau^\varepsilon\wedge T}^\varepsilon)\bigr]-g(\varepsilon)
=
 \E\left[\int_0^{\tau^\varepsilon\wedge T}
 \left\{(\mu-\dot b(s))\bigl(r-2Y_s^\varepsilon\bigr)-\sigma^2\right\}ds\right].
 \label{eq:g-Ito}
\end{align}
On $\{s<\tau^\varepsilon\}$,
$|r-2Y_s^\varepsilon|\leq r$. Choose $r>0$ sufficiently small that
\begin{equation}\label{eq:rK}
 rK\leq\frac{\sigma^2}{2}.
\end{equation}
Since $g\geq0$, equations \eqref{eq:K} and \eqref{eq:g-Ito} imply
$$
 0\leq\E\bigl[g(Y_{\tau^\varepsilon\wedge T}^\varepsilon)\bigr]
 \leq g(\varepsilon)-\frac{\sigma^2}{2}
 \E[\tau^\varepsilon\wedge T].
$$
Letting $T\to\infty$ gives
\begin{equation}\label{eq:Etau}
 \E[\tau^\varepsilon]
 \leq\frac{2\varepsilon(r-\varepsilon)}{\sigma^2}
 \leq\frac{2r}{\sigma^2}\,\varepsilon.
\end{equation}
In particular, Markov's inequality yields
\begin{equation}\label{eq:tau-tail}
 \P(\tau^\varepsilon\geq\delta)
 \leq\frac{2r}{\sigma^2\delta}\,\varepsilon.
\end{equation}

\medskip
\noindent\underline{Step 2.} Recall the constant $K$ from \eqref{eq:K}. Assume first that $K>0$, set
$$
 \alpha:=\frac{2K}{\sigma^2},
 \qquad
 \psi(x):=\frac{e^{\alpha x}-1}{e^{\alpha r}-1},
 \qquad x\in[0,r].
$$
Then $\psi(0)=0$, $\psi(r)=1$, $\psi'>0$, and $\psi''=\alpha\psi'$.
Using \eqref{eq:K}, for almost every $s$ and every $x\in(0,r)$,
$$
 (\mu-\dot b(s))\psi'(x)+\frac{\sigma^2}{2}\psi''(x)
 \geq-K\psi'(x)+\frac{\sigma^2\alpha}{2}\psi'(x)=0.
$$
It\^o's formula and optional sampling therefore give
$$
\begin{aligned}
 \psi(\varepsilon)
 &\leq\E\bigl[\psi(Y_{\tau^\varepsilon}^\varepsilon)\bigr]=\P(\tau_r^\varepsilon<\tau_0^\varepsilon).
\end{aligned}
$$
Since $e^{\alpha\varepsilon}-1\geq\alpha\varepsilon$,
\begin{equation}\label{eq:upper-exit-K}
 \P(\tau_r^\varepsilon<\tau_0^\varepsilon)
 \geq\frac{\alpha}{e^{\alpha r}-1}\,\varepsilon.
\end{equation}
If $K=0$, the same argument with $\psi(x)=x/r$ gives
\begin{equation}\label{eq:upper-exit-zero}
 \P(\tau_r^\varepsilon<\tau_0^\varepsilon)
 \geq\frac{1}{r}\,\varepsilon.
\end{equation}
Thus, in either case,
\begin{equation}\label{eq:c0}
 \P(\tau_r^\varepsilon<\tau_0^\varepsilon)
 \geq c_0(r)\varepsilon,
\end{equation}
where
$$
 c_0(r):=
 \begin{cases}
 \displaystyle\frac{\alpha}{e^{\alpha r}-1},&K>0,\\[1.1ex]
 \displaystyle\frac1r,&K=0,
 \end{cases}
 \qquad\text{and}\qquad
 c_0(r)\sim\frac1r\quad\text{as }r\downarrow0.
$$

\medskip
\noindent\underline{Step 3.}
Combining \eqref{eq:tau-tail} and \eqref{eq:c0},
\begin{align}
 \P(\tau_r^\varepsilon<\tau_0^\varepsilon\wedge\delta)
 &\geq
 \P(\tau_r^\varepsilon<\tau_0^\varepsilon)
 -\P(\tau^\varepsilon\geq\delta)\geq
 \left(c_0(r)-\frac{2r}{\sigma^2\delta}\right)\varepsilon.
 \label{eq:prob-final}
\end{align}
Because $c_0(r)\sim r^{-1}$ as $r\downarrow0$, we may choose $r>0$
small enough so that both \eqref{eq:rK} holds and
$$
 C_1:=c_0(r)-\frac{2r}{\sigma^2\delta}>0.
$$
For all sufficiently small $\varepsilon\in(0,r)$,
\begin{equation}\label{eq:linear-W}
 W(h^*,m^*,z^*+\varepsilon)
 \geq a_r e^{-\overline q\delta}C_1\varepsilon.
\end{equation}
Since $\widehat W(0,z^*)=W(h^*,m^*,z^*)=0$, it follows that
\begin{equation}\label{eq:positive-derivative}
 \liminf_{\varepsilon\downarrow0}
 \frac{\widehat W(0,z^*+\varepsilon)-\widehat W(0,z^*)}{\varepsilon}
 \geq a_r e^{-\overline q\delta}C_1>0.
\end{equation}
By the regularity assumption of $V_H$, $\widehat{W}$ is differentiable in $z$ at $z^*$. Therefore, the statement of Lemma \ref{lem:hatW_z} follows from the previous inequality.

% On the other hand, the assumed smooth-fit conditions for the two equilibrium value
% functions (i.e.\ their continuous differentiability across the relative free boundaries) give
% $$
% \begin{aligned}
%  \partial_z^+\widehat W(0,z^*)
%  &=\partial_z^+V_H^{\cZ_1}
%    \bigl(h^*,m^*,\cZ_1(h^*,m^*)\bigr)-
%  \partial_z^+V_H^{\cZ_2}
%    \bigl(h^*,m^*,\cZ_2(h^*,m^*)\bigr)=0.
% \end{aligned}
% $$
% Hence a contradiction to \eqref{eq:positive-derivative}.

\section{Withdrawal from patient depositors}\label{app:betaL_high}

In this section, we consider the analogue of Proposition \ref{prop:Poisson} when $\beta_L > r+ \delta - (1-\gamma) \eta$. In this case, patient depositors are already at risk at time zero. Because $r+\delta - \beta_L - (1-\gamma) \eta<0$, if all depositors withdraw at time zero, the net flow of benefit is negative for everyone, hence, withdrawing at time zero is optimal for everyone. Therefore, the earliest-run equilibrium happens at time zero. In the latest-run equilibrium, both types of depositors stay as long as possible, until the exogenous withdrawn share is sufficiently high, which forces them to withdraw. In Figure \ref{fig:betaL_high}, those depositors who become impatient before $\overline{t}_H$ do not withdraw until $\overline{t}_H$. At $\overline{t}_H$, the mass of $e^{-\nu \overline{t}_H} - e^{-(\nu + \theta_{LH})\overline{t}_H}$ impatient depositors withdraw together. Between $\overline{t}_H$ and $\overline{t}_L$, depositors withdraw immediately when they experience discount rate shocks and become impatient. At $\overline{t}_L$, all remaining patient depositors withdraw together, hence the withdrawn share turns to one. Therefore, the latest-run equilibrium features two clustered withdrawals: all impatient depositors withdraw together at $\overline{t}_H$ and all remaining patient depositors withdraw together at $\overline{t}_L$. The following result provides a formal statement of the extreme equilibria and the form of any generic equilibrium. 

\begin{figure}[ht!]
\centering
\includegraphics[scale=0.6]{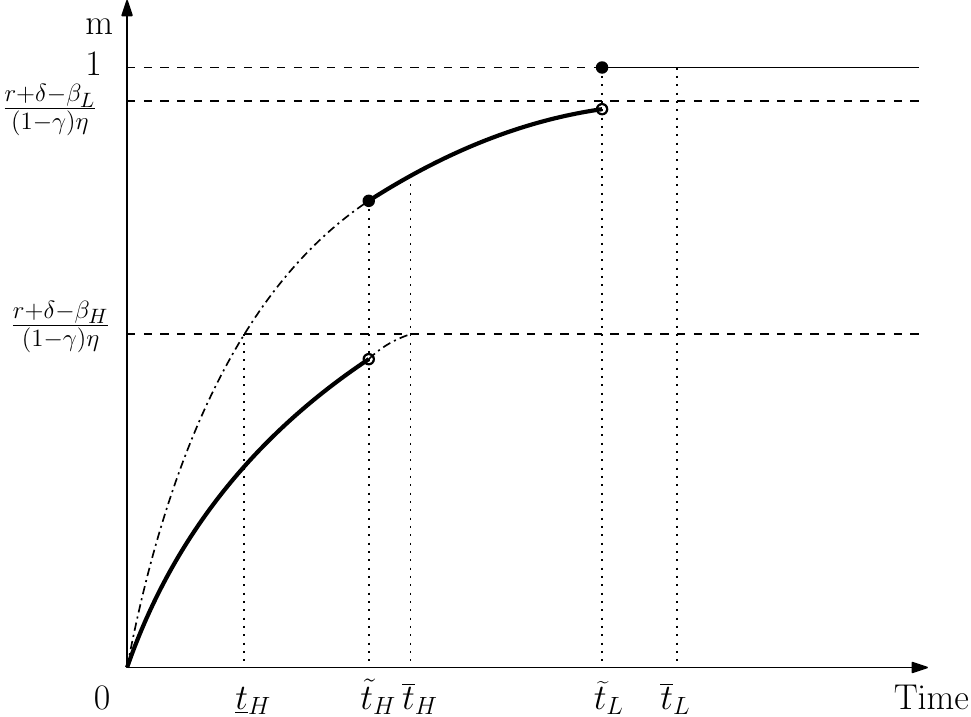}
\caption{Equilibrium withdrawn share}
 \begin{minipage}{\textwidth}
	\footnotesize
	\vspace{2mm}
    This figure presents all mean-field equilibria in Proposition \ref{prop:Possion_high_betaL}. The high discount rate $\beta_H$ satisfies $\beta_H < r+\delta$. In the earliest-run equilibrium, all depositors withdraw at time zero. In the latest-run equilibrium, all impatient depositors do not withdraw before $\overline{t}_L$, they withdraw at time $\overline{t}_L$. When patient depositors become impatient between $\overline{t}_L$ and $\overline{t}_H$, they withdraw immediately. At $\overline{t}_H$, all remaining patient depositors withdraw together.  The solid black line represents the withdrawn share in one equilibrium sandwiched between the earliest-run and the latest-run equilibria.   
    \end{minipage}
    \label{fig:betaL_high}
\end{figure}

\begin{prop}\label{prop:Possion_high_betaL}
Suppose that $r + \delta > \beta_L > r+\delta - (1-\gamma) \eta$. When $\beta_H < r+\delta$, define 
\begin{align}
 \overline{t}_H := & \left\{\begin{array}{ll} -\tfrac{1}{\nu} \log \big[1-\tfrac{r+\delta - \beta_H}{(1-\gamma) \eta}\big], & \nu >0, \\ \infty, & \nu=0, \end{array}\right. \quad \overline{t}_L =\left\{ \begin{array}{ll} -\tfrac{1}{\nu + \theta_{LH}} \log \big[1-\tfrac{r+\delta - \beta_L}{(1-\gamma)\eta}\big] \vee \overline{t}_H, & \nu >0,\\ \infty, & \nu = 0,\end{array}\right.\label{case2:ot_LH}\\
 \underline{t}_H :=& -\tfrac{1}{\nu + \theta_{LH}} \log\big[1-\tfrac{r+\delta - \beta_H}{(1-\gamma) \eta}\big].\label{case2:ut_L}
\end{align}
When $\beta_H \geq r+\delta$, $\overline{t}_H = \underline{t}_H=0$ and $\overline{t}_L$ is defined as in \eqref{case2:ot_LH}.  
Then,
\begin{enumerate}
\item[i)] The earliest-run equilibrium is $\underline{\tau}^*_w = 0$ and $\underline{m}_s = 1$ for any $s\geq 0$, i.e., all depositors strategically withdrawal immediately. The latest-run equilibrium is $\overline{\tau}^*_w = \overline{t}_L \wedge \big( \tau_\theta \vee \overline{t}_H\big)$ and 
\[
\overline{m}_s = \left\{\begin{array}{ll} 1- e^{-\nu s}, & s < \overline{t}_H,\\ 1- e^{-(\nu + \theta_{LH}) s}, & \overline{t}_H \leq s< \overline{t}_L,\\ 1, & \overline{t}_L \leq s. \end{array}\right.
\]
In the latest-run equilibrium, all impatient depositors do not strategically withdraw before $\overline{t}_H$, but all impatient depositors, of mass $e^{-\nu \overline{t}_H} - e^{-(\nu + \theta_{HL})\overline{t}_H}$, strategically withdraw together at $\overline{t}_H$. Between $\overline{t}_H$ and $\overline{t}_L$, depositors strategically withdraw when they become impatient. At $\overline{t}_L$, all remaining depositors withdraw together. 
\item[ii)] Each equilibrium is indexed by a pair $(\tilde{t}_H, \tilde{t}_L)$ with $ 
\tilde{t}_L \leq \overline{t}_L$ and $\underline{t}_H\wedge \tilde{t}_L \leq \tilde{t}_H  \leq \overline{t}_H \wedge \tilde{t}_L$. The optimal withdrawal time is $\tilde{\tau}^*_w = \tilde{t}_L \wedge \big(\tau_\theta \vee \tilde{t}_H \big)$ and 
\begin{equation}\label{tm_betaL_high}
\widetilde{m}_s = \left\{\begin{array}{ll} 1- e^{-\nu s}, & s< \tilde{t}_H,\\ 1- e^{-(\nu + \theta_{HL}) s}, & \tilde{t}_H \leq s< \tilde{t}_L,\\ 1, & \tilde{t}_L \leq s. \end{array}\right.
\end{equation}
When $\tilde{t}_L = \tilde{t}_H$, $\tilde{\tau}^*_w = \tilde{t}_L$ and the middle case in the previous equation is absent. 
\end{enumerate}
\end{prop}

\begin{proof}[Proof of Proposition \ref{prop:Possion_high_betaL}]

For given $m$, the optimal withdrawal time for an impatient depositor is still $\tau^*_w(H;m)$ in \eqref{tau_Lm} and $V_H(t; m) =0$ when $t\geq \tau^*_w(H;m)$. For patient depositors, introduce 
\begin{equation}\label{tau_Hm}
\tau^*_w(L;m) = \inf\{s\geq t\,:\, r+\delta - \beta_L - (1-\gamma) \eta m_{s} \leq 0\}.
\end{equation}
Due to $\beta_L < \beta_H$, $\tau^*_w(H;m) \leq \tau^*_w(L,m)$. We claim that $\tau^*_w(L,m)$ is the optimal withdrawal time for a patient depositor. To see this, notice from \eqref{Poisson:VH} that patient depositors should optimally withdraw at the time
$$\inf\{s\geq t\,:\, r+\delta - \beta_L - (1-\gamma) \eta m_{s} + \theta_{LH} V_H(s;m) \leq 0\}.$$
Given that $V_H\geq 0$, we have that the latter time is larger or equal than $\tau^*_w(L;m)$ as in \eqref{tau_Hm}.
Hence,
\begin{align}
V_L(t;m)=&\int_t^{\tau^*_w(L;m)} 
     e^{-\int_t^s (\nu + \beta_L + \theta_{LH} + \eta m_u) du} \big( r+ \delta -\beta_L - (1-\gamma) \eta m_s + \theta_{LH} V_H (s; m)\big) ds 
     \nonumber \\
& + \sup_{\tau_w} \int_{\tau^*_w(L;m)}^{\tau_w} e^{-\int_t^s (\nu + \beta_L + \theta_{LH} + \eta m_u) du} \big( r+\delta - \beta_L - (1-\gamma) \eta m_s \big) ds,
\label{eq:optimalH}
\end{align}
where we have used that $V_H(s;m)=0$ for any $s\geq \tau^*_w(L;m)\geq \tau^*_w(H;m)$.
By definition of $\tau^*_w(L;m)$ and by the fact that $t\mapsto m_t$ is nondecreasing, we have that the supremum on the right-hand side of \eqref{eq:optimalH} is attained at $\tau^*_w(L;m)$. Hence, the claimed optimality of $\tau^*_w(L;m)$ for patient depositors is shown.

%we recall the high-type value function from \eqref{Poisson:VH}. For any $t\geq \tau^*_w(H; m)$, note that $V_L(t; m) =0$ because $\tau^*_w(H;m) \geq \tau^*_w(L;m)$ and $V_L(t;m) =0$ for any $t \geq \tau^*_w(L;m)$. Meanwhile, $t\mapsto m_t$ is nondecreasing. Therefore $\lambda (R_H -1) - (1-\gamma) \eta m_t + \theta V_L(t; m) \leq 0$ for any $t\geq \tau^*_w(H;m)$. It then follows from the form of $V_H(t; m)$ that it is optimal to stop immediately when $t\geq \tau^*_w(H;m)$. When $t< \tau^*_w(H;m)$, $\lambda(R_H-1) - (1-\gamma) \eta m_t >0$ and $V_L(t;m)\geq 0$. Then $V_H(t;m) >0$ for any $t<\tau^*_w(H;m)$. Putting two cases together, we conclude that $\tau^*_w$ is the optimal withdrawal time for a high-type depositor. 

In summary, for any given $m$, a representative depositor either optimally withdraws at $\tau^*_w(L;m)$ if she is still patient, or withdraws at the later time between $\tau^*_w(H;m)$ and $\tau_\theta$ when she experiences a discount rate shock and becomes impatient. Therefore, the optimal withdrawal time for the representative depositor is 
\[
\tau^*_w = \tau^*_w(L; m) \wedge \big(\tau_\theta \vee \tau^*_w(H;m)\big).
\]

By the consistency condition \eqref{def:m_prob} for the case without common state, for any $s\geq0$, we have 
\begin{align}
m_s = &\mathbb{P}(\tau_\nu \wedge \tau^*_w \leq s) = \mathbb{P} \Big(\tau_\nu \wedge \tau^*_w(L; m) \wedge \big(\tau_\theta \vee \tau^*_w(H;m)\big) \leq s \Big) \nonumber\\
= & \mathbb{P}(\tau_\nu \leq s) + \mathbb{P}\Big(\tau_\nu > s, \tau^*_w(L; m) \wedge \big(\tau_\theta \vee \tau^*_w(H;m)\big) \leq s\Big)\nonumber\\
= & 1- e^{-\nu s} + e^{-\nu s}\,  \mathbb{P} \Big(\tau^*_w(L; m) \wedge \big(\tau_\theta \vee \tau^*_w(H;m)\big) \leq s \Big)\nonumber\\
= & 1- e^{-\nu s} + e^{-\nu s} \Big( (1-e^{-\theta_{LH} s}) 1_{\{\tau^*_w(H;m) \leq s < \tau^*_w(L;m)\}} + 1_{\{\tau^*_w(L;m) \leq s\}}\Big).\label{Poisson:m_eqn_2}
\end{align}
The reason for the last equality is the following. When $s<\tau^*_w(H;m)$, $\tau^*_w(L;m) \wedge \big(\tau_\theta \vee \tau^*_w(H;m)\big)$ is at least $\tau^*_w(L;m)$, hence the set $\big\{\tau^*_w(L; m) \wedge \big(\tau_\theta \vee \tau^*_w(H;m)\big) \leq s \big\}$ is empty. When $\tau^*_w(H;m) \leq s< \tau^*_w(L;m)$, the set $\big\{\tau^*_w(L; m) \wedge \big(\tau_\theta \vee \tau^*_w(H;m)\big) \leq s \big\}$ is $\{\tau_\theta \leq s\}$. Finally, when $\tau^*_w(L;m) \leq s$, the inequality $\tau^*_w(L; m) \wedge \big(\tau_\theta \vee \tau^*_w(H;m)\big) \leq s$ always holds. The withdrawn share $m$ shows up on both sides of \eqref{Poisson:m_eqn_2}, hence it is an equation for $m$. 

To identify all solutions of \eqref{Poisson:m_eqn_2}, we first observe that the right-hand side of \eqref{Poisson:m_eqn_2} is bounded from above by $\underline{m}_t = 1$ for any $t\geq 0$. This corresponds to the case where all depositors withdraw at time $\underline{t}=0$. Due to the assumption $(1-\gamma) \eta \geq r+\delta - \beta_L > r+\delta-\beta_H$, it follows from \eqref{tau_Lm} and \eqref{tau_Hm} that $\tau^*_w(L, \underline{m}) = \tau^*_w(H, \underline{m}) = \underline{t} =0$. Therefore $(\underline{t}, \underline{m}) = (0,1)$ is the earliest equilibrium. 

To identify the latest equilibrium, define
\begin{align*}
 \overline{t}_H := & \inf\{s\geq 0\,:\, r+\delta - \beta_H - (1-\gamma)\eta (1-e^{-\nu s}) \leq 0\} \quad \text{and}\\
 \overline{t}_L := & \inf\{s\geq \overline{t}_H\,:\, r+\delta - \beta_L - (1-\gamma) \eta (1- e^{-(\nu+\theta_{LH}) s})\leq 0\}. 
\end{align*}
They are the upper bounds for $\tau^*_w(H;m)$ and $\tau^*_w(L,m)$, respectively. Introduce  
\[
\overline{m}_s = \left\{\begin{array}{ll} 1- e^{-\nu s}, & s < \overline{t}_H,\\ 1- e^{-(\nu + \theta_{LH}) s}, & \overline{t}_H \leq s< \overline{t}_L,\\ 1, & \overline{t}_L \leq s. \end{array}\right.
\]
From this definition, we can check $\tau_w^*(H; \overline{m}) = \overline{t}_H$. Meanwhile, for any $s< \overline{t}_L$, $r+\delta - \beta_L - (1-\gamma) \eta \overline{m}_s >0$; for any $s\geq \overline{t}_L$, $\overline{m}_s =1$, because $r+\delta - \beta_L - (1-\gamma) \eta \overline{m}_s \leq 0$. Therefore, $\tau^*_w(L; \overline{m}) = \overline{t}_L$. This shows that $\overline{\tau}^*_w = \overline{t}_L \wedge \big( \tau_\theta \vee \overline{t}_H\big)$ together with $\overline{m}$ forms an equilibrium. Given that $\overline{t}_L$ and $\overline{t}_H$ are the upper bounds for the optimal withdrawal time for low-type and high-type depositors, $(\overline{\tau}^*_w, \overline{m})$ is the latest equilibrium. 

The proof for the form of an arbitrary equilibrium in \eqref{tm_betaL_high} is similar to the proof of Proposition \ref{prop:Poisson} ii).
\end{proof}

\section{Comparative statics for equilibria with continuous private states}\label{app:comp_statics}

Figure \ref{fig:comp} presents the impact of the bank failure sensitivity $\eta$, the private state volatility $\sigma$, and the initial private state dispersion $\sigma_0$ on the withdrawn share dynamics and the stopping boundary in the model setting of Section \ref{Section: continuous state multiple clusters}. Panel (a) is consistent with Proposition \ref{prop comparative statics strength parameter} --- the withdrawn share increases with $\eta$ and the clustered withdrawal happens earlier. As the private state becomes more volatile, the left-tail of the net flow benefit, $r+\delta - \beta(X_t)$, forms earlier. Hence, more depositors become marginal sooner, leading to an earlier clustered withdrawal, as Panel (b) shows. The same intuition applies when the initial private state dispersion is larger. When $\sigma_0=1$, Panel (c) shows that a group of marginal depositors, whose initial private states are sufficiently bad, already withdraw together at time zero. 

\begin{figure}
\centering
\begin{subfigure}{\textwidth}
    \includegraphics[scale=0.6]{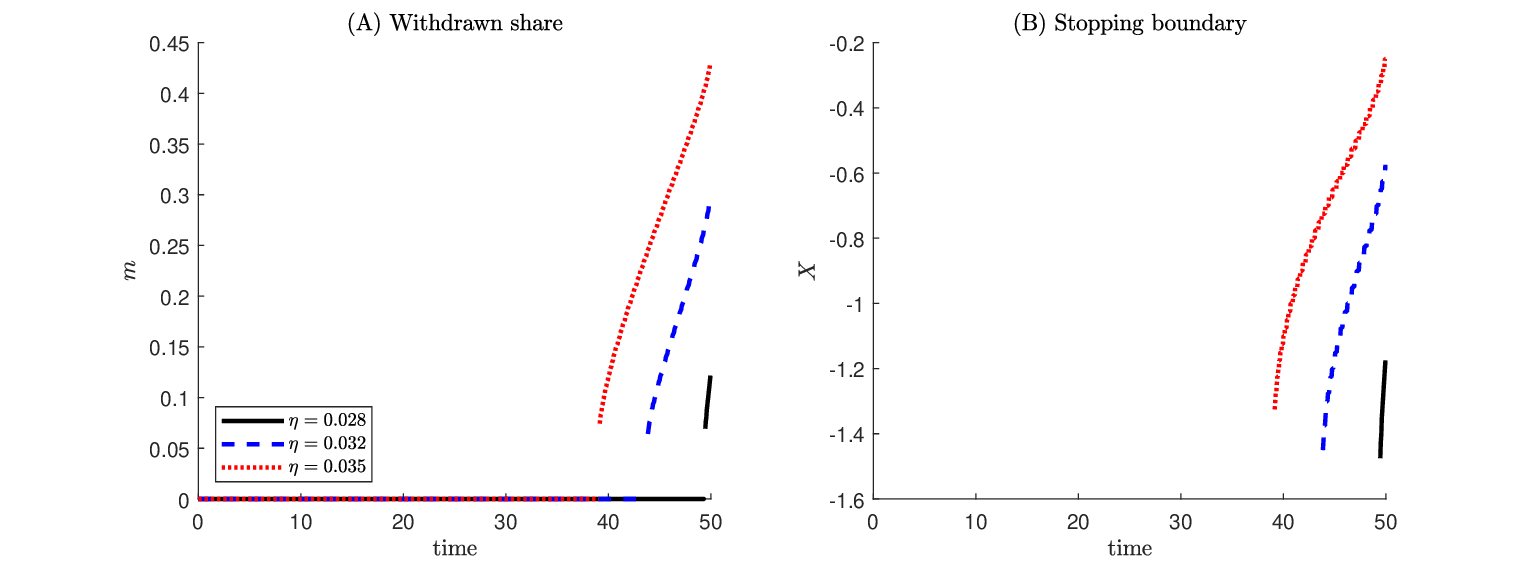}
\caption{Different $\eta$}
\end{subfigure}
\begin{subfigure}{\textwidth}
   \includegraphics[scale=0.6]{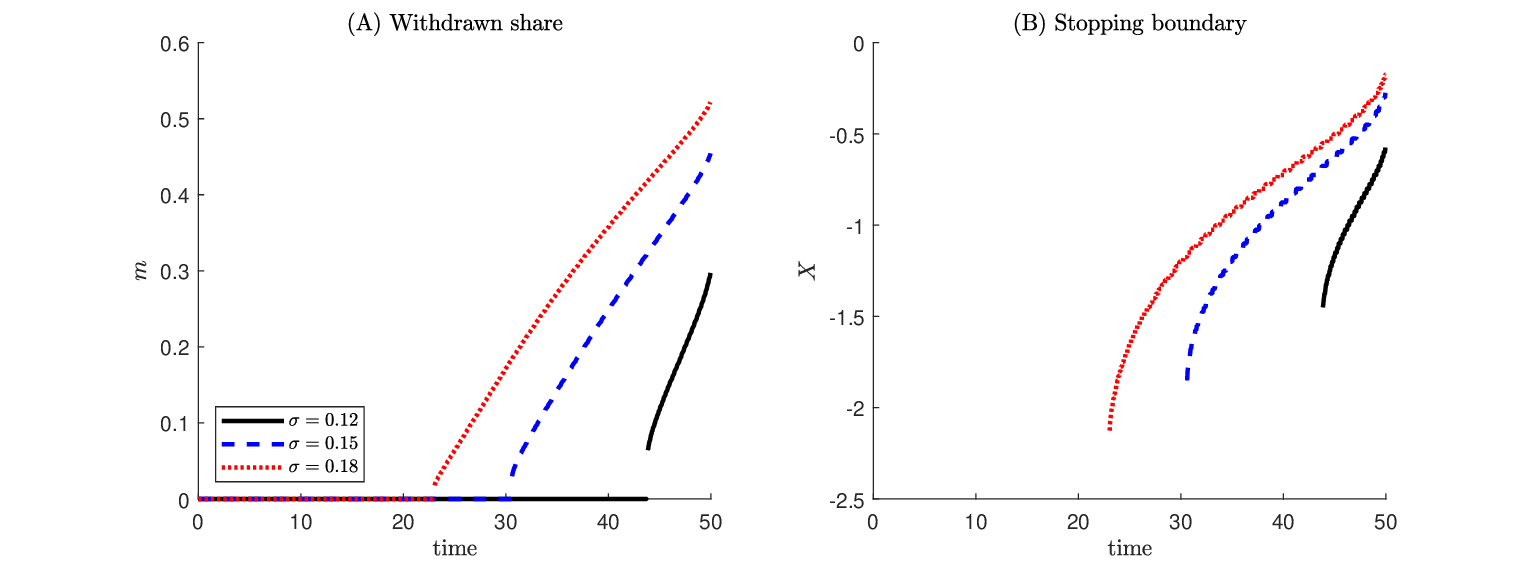}
\caption{Different $\sigma$}
\end{subfigure}

\begin{subfigure}{\textwidth}
  \includegraphics[scale=0.6]{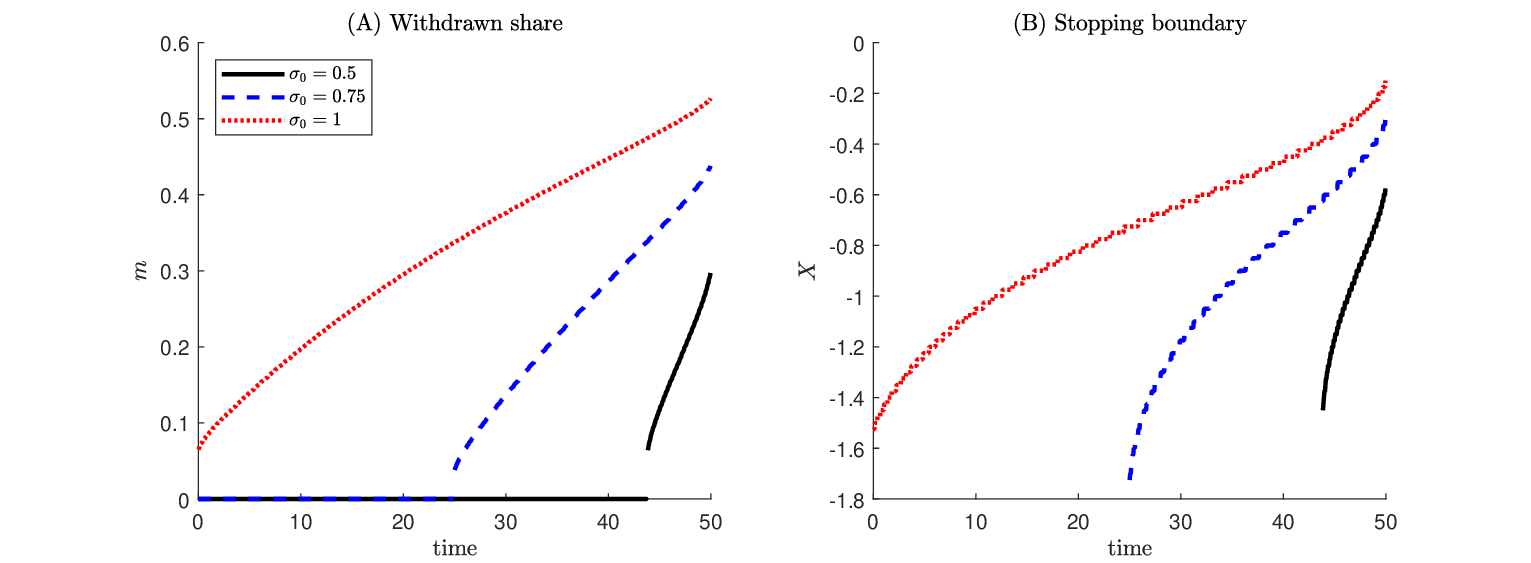}
\caption{Different $\sigma_0$}
\end{subfigure}
\caption{Comparative statics}
 \begin{minipage}{\textwidth}
	\footnotesize
	\vspace{2mm}
    This figure presents the impact of $\eta$, $\sigma$, and $\sigma_0$ on the withdrawn share and the stopping boundary in the earliest-run equilibrium. Parameters are the same as in Figure \ref{fig:BM}.
 \end{minipage}   
    \label{fig:comp}
\end{figure}

\end{document}